\documentclass[11pt]{article}
 
\usepackage[T1]{fontenc}
\usepackage[utf8]{inputenc}

\usepackage{lmodern}
\usepackage{xspace}                                     

\usepackage{amsfonts,amsmath,amssymb, amsthm, mathtools}
\usepackage{thm-restate}
\usepackage{dsfont} 

\usepackage[backref,colorlinks,citecolor=blue,bookmarks=true,linktocpage]{hyperref}
\usepackage[capitalize]{cleveref}
\usepackage[numbered]{bookmark}

\usepackage{algorithmicx,algpseudocode,algorithm}

\usepackage[usenames,dvipsnames,table]{xcolor}

\usepackage{fullpage}

\usepackage[shortlabels]{enumitem}
  \setitemize{noitemsep,topsep=3pt,parsep=2pt,partopsep=2pt} 
  \setenumerate{itemsep=1pt,topsep=2pt,parsep=2pt,partopsep=2pt}
  \setdescription{itemsep=1pt}

\usepackage{mleftright} 
 \makeatletter
\@ifundefined{theorem}{
  \theoremstyle{plain} 
  	\newtheorem{theorem}{Theorem}[section]
  	  \newtheorem{corollary}{Corollary}[section]
  		\newtheorem{lemma}{Lemma}[section]
  		\newtheorem{claim}{Claim}[section]
 	 	\newtheorem{fact}[theorem]{Fact}
  		\newtheorem{proposition}{Proposition}[section]
  		
  \theoremstyle{remark} 
  	\newtheorem{remark}[theorem]{Remark}

  \theoremstyle{definition} 
 		 \newtheorem{definition}{Definition}[section]
}{}
\makeatother
\newenvironment{proofof}[1]{\begin{proof}[Proof of {#1}]}{\end{proof}}

  \newcommand{\new}[1]{{{#1}}}
  \newcommand{\newer}[1]{{{#1}}}
  
  \newcommand{\newerest}[1]{{{#1}}}

\newcommand{\ignore}[1]{\leavevmode\unskip}

\newcommand{\eps}{\ensuremath{\varepsilon}\xspace}

\newcommand{\property}{\ensuremath{\mathcal{P}}\xspace} 
 
\newcommand{\eqdef}{:=}

\newcommand{\accept}{\textsf{accept}\xspace}

\newcommand{\reject}{\textsf{reject}\xspace}
 
\newcommand{\distribs}[1]{\Delta\!\left(#1\right)} 
\newcommand{\yes}{{\sf{}yes}\xspace}
\newcommand{\no}{{\sf{}no}\xspace}

\newcommand{\mutualinfo}[2]{ I\left(#1; #2\right) }

\newcommand{\condmutualinfo}[3]{ I\mleft(#1; #2 \middle| #3\mright) }

\newcommand{\bigO}[1]{{O\mleft( #1 \mright)}}

\newcommand{\bigTheta}[1]{{\Theta\mleft( #1 \mright)}}
\newcommand{\bigOmega}[1]{{\Omega\mleft( #1 \mright)}}

\newcommand{\setOfSuchThat}[2]{ \left\{\; #1 \;\colon\; #2\; \right\} } 			
\newcommand{\indicSet}[1]{\mathds{1}_{#1}}                                              
\newcommand{\indic}[1]{\indicSet{\left\{#1\right\}}}

\newcommand{\dtv}{\operatorname{d}_{\rm TV}}

\newcommand{\totalvardistrestr}[3][]{{\dtv^{#1}\!\left({#2, #3}\right)}}
\newcommand{\totalvardist}[2]{\totalvardistrestr[]{#1}{#2}}

\newcommand{\kldiv}[2]{\operatorname{d_{\rm KL}}\mleft({#1 \mid\mid #2}\mright)}

\newcommand{\proba}{\Pr}
\newcommand{\probaOf}[1]{\proba\!\left[\, #1\, \right]}
\newcommand{\probaCond}[2]{\proba\!\left[\, #1 \;\middle\vert\; #2\, \right]}
\newcommand{\probaDistrOf}[2]{\proba_{#1}\left[\, #2\, \right]}

\newcommand{\expect}[1]{\mathbb{E}\!\left[#1\right]}
\newcommand{\expectCond}[2]{\mathbb{E}\!\left[\, #1 \;\middle\vert\; #2\, \right]}
\newcommand{\shortexpect}{\mathbb{E}}
\newcommand{\var}{\operatorname{Var}}
\newcommand{\cov}{\operatorname{Cov}}

\newcommand{\bernoulli}[1]{\ensuremath{\operatorname{Bern}\!\left( #1 \right)}}

\newcommand{\binomial}[2]{\ensuremath{\operatorname{Bin}\!\left( #1, #2 \right)}}
\newcommand{\poisson}[1]{\ensuremath{\operatorname{Poisson}\!\left( #1 \right) }}

\newcommand{\norm}[1]{\lVert#1{\rVert}}
\newcommand{\normone}[1]{{\norm{#1}}_1}
\newcommand{\normtwo}[1]{{\norm{#1}}_2}

\newcommand{\abs}[1]{\left\lvert #1 \right\rvert}

\newcommand{\vect}[1]{\mathbf{#1}}

\newcommand{\flr}[1]{\left\lfloor #1 \right\rfloor}

\newcommand{\R}{\ensuremath{\mathbb{R}}\xspace}

\newcommand{\N}{\ensuremath{\mathbb{N}}\xspace}

\newcommand{\lp}[1][1]{\ell_{#1}}

\newcommand{\p}{\ensuremath{p}}
\newcommand{\q}{\ensuremath{q}}

\newcommand{\condindprop}[3]{\property_{#1,#2\mid #3}}
\newcommand{\condindrv}[3]{{(#1\perp #2)\mid #3}}
\newcommand{\domx}{\mathcal{X}}
\newcommand{\domy}{\mathcal{Y}}
\newcommand{\domz}{\mathcal{Z}}
\newcommand{\conddistr}[2]{{(#1 \mid #2)}}
 
\usepackage{times}

\begin{document}

\title{Testing Conditional Independence of Discrete Distributions}

\author{
Cl\'{e}ment L. Canonne\thanks{Supported by a Motwani Postdoctoral Fellowship. Some of this work was performed while visiting USC, and a graduate student at Columbia University supported by NSF grants CCF-1115703 and NSF CCF-1319788.}\\
Stanford University\\
{\tt ccanonne@cs.stanford.edu}
\and
Ilias Diakonikolas\thanks{Supported by NSF Award CCF-1652862 (CAREER) and a Sloan Research Fellowship.}\\
University of Southern California\\
{\tt diakonik@usc.edu}\\
\and
Daniel M. Kane\thanks{Supported by NSF Award CCF-1553288 (CAREER) and a Sloan Research Fellowship.}\\
University of California, San Diego\\
{\tt dakane@cs.ucsd.edu}\\
\and
Alistair Stewart\\
University of Southern California\\
{\tt alistais@usc.edu}
}

\maketitle

\setcounter{page}{0}
\thispagestyle{empty}

\begin{abstract}
We study the problem of testing \emph{conditional independence} for discrete distributions.
Specifically, given samples from a discrete random variable $(X, Y, Z)$ on domain $[\ell_1]\times[\ell_2] \times [n]$, 
we want to distinguish, with probability at least $2/3$, between the case that $X$ and $Y$ are conditionally independent 
given $Z$ from the case that $(X, Y, Z)$ is $\eps$-far, in $\lp[1]$-distance, from every distribution that has this property.
Conditional independence is a concept of central importance in probability and statistics with a range of applications
in various scientific domains. As such, the statistical task of testing conditional independence has been extensively studied
in various forms within the statistics and econometrics communities for nearly a century. 
Perhaps surprisingly, this problem has not been previously considered in the framework of distribution property testing 
and in particular no tester with sublinear sample complexity is known, even for the important special case that the domains of $X$ and $Y$ are binary.

The main algorithmic result of this work is the first conditional independence tester with {\em sublinear} sample complexity for 
discrete distributions over $[\ell_1]\times[\ell_2] \times [n]$. 
To complement our upper bounds, we prove information-theoretic lower bounds establishing 
that the sample complexity of our algorithm is optimal, up to constant factors, for a number of settings.
Specifically, for the prototypical setting when $\ell_1, \ell_2 = O(1)$, we show that the sample complexity of testing 
conditional independence (upper bound and matching lower bound) is 
 \[
      \bigTheta{\max\left(n^{1/2}/\eps^2,\min\mleft(n^{7/8}/\eps,n^{6/7}/\eps^{8/7}\mright)\right)}\,.
  \]

To obtain our tester, we employ a variety of tools, including 
(1) a suitable weighted adaptation of the flattening technique~\cite{DK:16},
and (2) the design and analysis of an optimal (unbiased) estimator 
for the following statistical problem of independent interest: 
Given a degree-$d$ polynomial $Q\colon\mathbb{R}^n \to \R$ 
and sample access to a distribution $p$ over $[n]$, 
estimate $Q(p_1, \ldots, p_n)$ up to small additive error. 
Obtaining tight variance analyses for specific estimators of this form
has been a major technical hurdle in distribution testing (see, e.g.,~\cite{CDVV14}). 
As an important contribution of this work, we develop a general theory 
providing tight variance bounds for \emph{all} such estimators. Our lower bounds, established
using the mutual information method, rely on novel constructions of hard instances 
that may be useful in other settings.
\end{abstract}

\newpage

\section{Introduction} \label{sec:intro}

\subsection{Background} \label{ssec:background}

Suppose we are performing a medical experiment. Our goal is to compare a binary response 
($Y$) for two treatments ($X$), using data obtained at $n$ levels of a possibly confounding factor ($Z$).
We have a collection of observations grouped in strata (fixed values of $Z$). 
The stratified data are summarized in a series of $2 \times 2$ contingency tables, one for each strata. 
One of the most important hypotheses in this context is 
conditional independence of $X$ and $Y$ given $Z$. 
How many observations $(X, Y, Z)$ do we need so that we can confidently test this hypothesis?

The above scenario is a special case of the following statistical problem: 
Given samples from a joint discrete distribution $(X, Y, Z)$, are the random variables $X, Y$ 
independent conditioned on $Z$? 
This is the problem of \emph{testing conditional independence} --- 
a fundamental statistical task with applications in a variety of fields, 
including medicine, economics and finance, etc. (see, e.g., ~\cite{MS59, Spirtes2000, WH17} and references therein). 
Formally, we have the following definition:

\begin{definition}[Conditional Independence]
Let $X, Y, Z$ be random variables over discrete domains $\domx,\domy,\domz$ respectively. 
We say that $X$ and $Y$ are \emph{conditionally independent given $Z$}, denoted by $\condindrv{X}{Y}{Z}$, if
for all $(i,j,z)\in \domx \times \domy \times \domz$ we have that:
$\probaCond{X=i, Y=j}{ Z = z} = \probaCond{X=i}{ Z = z}\cdot \probaCond{Y=j}{ Z = z}$. 
\end{definition}

Conditional independence is an important concept in probability theory and statistics, and is a widely used 
assumption in various scientific disciplines~\cite{Dawid79}. Specifically, it is a central notion in 
modeling causal relations~\cite{Spirtes2000} and of crucial importance in graphical modeling~\cite{Pearl88}.
Conditional independence is, in several settings, a direct implication of economic theory. 
A prototypical such example is the Markov property of a time series process.
The Markov property is a natural property in time series analysis and is broadly 
used in economics and finance~\cite{Easley87}.
Other examples include distributional Granger 
non-causality~\cite{Granger80}~---~which is a particular case of conditional independence~---~and exogeneity~\cite{BlundHor07}.

Given the widespread applications of the conditional independence assumption, 
the statistical task of \emph{testing} conditional independence has been studied extensively
for nearly a century. In 1924, R.~A.~Fisher~\cite{Fisher24} proposed the notion of partial correlation coefficient,
which leads to Fisher's classical $z$-test for the case that the data comes from a \emph{multivariate Gaussian
distribution}. For discrete distributions, conditional independence testing 
is one of the most common inference questions that arise in the context of contingency tables~\cite{Agresti1992}. 
In the field of graphical models, conditional independence testing is a cornerstone in the context of 
structure learning and testing of Bayesian networks 
(see, e.g.,~\cite{Neapolitan:2003, Tsamardinos2006,  Natori17, CDKS17} and references therein).
Finally, conditional independence testing is a useful tool in recent applications of machine learning
involving fairness~\cite{HardtPNS16}.

One of the classical conditional independence tests in the discrete setting  is the 
Cochran--Mantel--Haenszel test~\cite{Coch54, MS59}, which requires certain strong assumptions about 
the marginal distributions. When such assumptions do not hold, a common tester used is a linear combination
of $\chi^2$-squared testers (see, e.g.,~\cite{Agresti1992}). However, even for the most basic case
of distributions over $\{0, 1\}^2 \times [n]$, no finite sample analysis is known. 
(Interestingly enough, we note that our tester can be viewed as an appropriately 
weighted linear combination of $\chi^2$-squared tests.)
A recent line of work  in econometrics has been focusing on 
conditional independence testing in \emph{continuous settings}~\cite{LG96, DelgMant01, 
SuWhite07, SuWhite08, Song09, GL10, Huang10, SuWhite14, Zhang11uai, BouezTaam14, dM14, WH17}. 
The theoretical results in these works are asymptotic in nature, while the finite sample 
performance of their proposed testers is evaluated via simulations.

In this paper, we will study the property of conditional independence in the framework
of distribution testing.
The field of \emph{distribution property testing}~\cite{BFR+:00} has seen substantial progress in the past decade,
see~\cite{Rub12, Canonne15,Gol:17} for two recent surveys and books.
A large body of the literature has focused on characterizing the sample size needed to test properties
of arbitrary distributions of a \emph{given} support size. This regime is fairly well understood:
for many properties of interest there exist sample-efficient testers
~\cite{Paninski:08, CDVV14, VV14, DKN:15, ADK15, CDGR16, DK:16, DiakonikolasGPP16, CDS17, Gol:17, DGPP17}.
Moreover, an emerging body of work has focused on leveraging \textit{a priori} structure
of the underlying distributions to obtain significantly improved sample
complexities~\cite{BKR:04, DDSVV13, DKN:15, DKN:15:FOCS, CDKS17, DaskalakisP17, DaskalakisDK16, DKN17}.

\subsection{Our Contributions} \label{ssec:results}
Rather surprisingly, the problem of testing conditional independence
has not been previously considered in the context of distribution property testing. 
In this work, we study this problem for discrete distributions
and provide the first conditional independence tester with sublinear sample complexity.
To complement our upper bound, we also provide information-theoretic 
lower bounds establishing that the sample complexity of our algorithm is optimal  
for a number of important regimes. To design and analyze our conditional independence tester
we employ a variety of tools, including an optimal (unbiased) estimator for the 
following statistical task of independent interest: Given a degree-$d$ polynomial 
$Q\colon\mathbb{R}^n \to \R$ and sample access to a distribution $\p$ over $[n]$, 
estimate $Q(p_1, \ldots, p_n)$ up to small additive error.

In this section, we provide an overview of our results.
We start with some terminology.
We denote by $\distribs{\Omega}$ the set of all distributions over domain $\Omega$.
For discrete sets $\domx,\domy,\domz$, we will use
$\condindprop{\domx}{\domy}{\domz}$ to denote
the property of conditional independence, i.e., 
\[    
\condindprop{\domx}{\domy}{\domz} \eqdef \setOfSuchThat{ \p\in\distribs{\domx\times\domy\times\domz} }{ (X,Y,Z)\sim \p \text{ satisfies } \condindrv{X}{Y}{Z} } \;.
\]
We say that a distribution $\p \in\distribs{\domx\times\domy\times\domz}$ 
is $\eps$-far from $\condindprop{\domx}{\domy}{\domz}$, 
if for every distribution $\q \in \condindprop{\domx}{\domy}{\domz}$ we have that $\totalvardist{\p}{\q} > \eps$. 
We study the following hypothesis testing problem:
\medskip

\fbox{\parbox{6in}{
\smallskip
\noindent $\mathcal{T}(\ell_1, \ell_2, n, \eps)$: 
Given sample access to a distribution $p$ over $\domx \times \domy \times \domz$, 
with $|\domx| = \ell_1$, $|\domy| = \ell_2$, $|\domz| = n$, and $\eps >0$,
distinguish with probability at least $2/3$ between the following cases: 
\begin{itemize}
\item Completeness: $p \in \condindprop{\domx}{\domy}{\domz}$.
\item Soundness: $\dtv(p, \condindprop{\domx}{\domy}{\domz}) \geq \eps$.
\end{itemize}
\smallskip
}}

\bigskip
Even though the focus of this paper is on testing under the total variation distance
metric (or equivalently the $\lp[1]$-distance), we remark that our techniques
yield near-optimal algorithms under the mutual information metric as well.
The interested reader is referred to~\cref{sec:mutualinfo} for a short description of these
implications.

The property of conditional independence captures a number of other important properties
as a special case. For example, the $n=1$ case reduces to the property of independence 
over $[\ell_1] \times [\ell_2]$, whose testing sample complexity was resolved only recently~\cite{DK:16}.
Arguably the prototypical regime of conditional independence corresponds to the other extreme.
That is, the setting  
that the domains $\domx$, $\domy$ 
are binary (or, more generally, of small constant size), 
while the domain $\domz$ is large. This regime exactly captures the well-studied 
and practically relevant setting of $2 \times 2 \times n$ contingency tables 
(mentioned in the motivating example of the previous section). 
For the setting where $\domx$, $\domy$ are small, our tester and our sample complexity
lower bound match, up to constant factors. Specifically, we prove:

\begin{theorem} \label{thm:binary-informal}
There exists a computationally efficient tester for 
$\mathcal{T}(2, 2, n, \eps)$ with sample complexity 
  \[
      \bigO{\max\left(n^{1/2}/\eps^2,\min\mleft(n^{7/8}/\eps,n^{6/7}/\eps^{8/7}\mright)\right)}\;.
  \]
Moreover, this sample upper bound is tight, up to constant factors. 
That is, any tester for  $\mathcal{T}(2, 2, n, \eps)$  requires at least 
$\bigOmega{\max\left(n^{1/2}/\eps^2,\min\mleft(n^{7/8}/\eps,n^{6/7}/\eps^{8/7}\mright)\right)}$ samples.
\end{theorem}

To the best of our knowledge, prior to our work, no $o(n)$ sample algorithm was known for this problem.
Our algorithm in this regime is simple: For every fixed value of $z \in [n]$, we consider the conditional distribution $\p_z$.
Note that $\p_z$ is a distribution over $\domx \times \domy$. We construct an unbiased estimator $\Phi$ 
of the squared $\lp[2]$-distance of any distribution on $\domx \times \domy$ from the product of its marginals. 
Our conditional independence tester uses this estimator in a black-box manner for each of the $p_z$'s.
In more detail, our tester computes a weighted linear combination of $\Phi(p_z)$, $z \in [n]$, 
and rejects if and only if this exceeds an appropriate threshold.

To obtain the required unbiased estimator of the squared $\lp[2]$-distance, 
we observe that this task is a special case of the following more general problem
of broader interest: For a distribution $\p = (p_1, \ldots, p_n)$ and a polynomial $Q: \mathbb{R}^n \rightarrow  \mathbb{R}$, 
obtain an unbiased estimator for the quantity $Q(p_1, \ldots, p_n)$.
We prove the following general result:

\begin{theorem} \label{thm:poly-informal}
For any degree-$d$ polynomial $Q: \mathbb{R}^n \rightarrow  \mathbb{R}$ and distribution $p$ over $[n]$, 
there exists a unique and explicit \emph{unbiased} estimator $U_N$ for $Q(\p)$ given $N\geq d$ samples. 
Moreover, this estimator is linear in $Q$ and its variance is at most
 \[
		\sum_{\substack{ \vect{s}\in\N^n \\1\leq \norm{\vect{s}}\leq d}}
		\left( \prod_{i=1}^n \p_i^{s_i} \right) 
		\left(\frac{\partial^{\norm{\vect{s}}}Q(\p)}{\partial X_1^{s_1} \dots \partial X_n^{s_n}}\right)^2
		\left( \frac{(N-\norm{\vect{s}})!}{N! \prod_{i=1}^n s_i!} \right) \;,
 \]
which itself can be further bounded as a function of $Q^+$, 
the degree-$d$ polynomial obtained by taking the absolute values 
of all the coefficients of $Q$, and its partial derivatives.
\end{theorem}

We note that~\cref{thm:poly-informal} can be appropriately extended 
to the setting where we are interested in estimating $Q(\p, \q)$, where $\p, \q$
are discrete distributions over $[n]$ and $Q$ is a real degree-$d$ polynomial on $2n$ variables.
In addition to being a crucial ingredient for our general conditional independence tester, 
we believe that~\cref{thm:poly-informal} is of independent interest.
Indeed, in a number of distribution testing problems, we need unbiased estimators for some
specific polynomial $Q$ of a distribution $\p$ (or a pair of distributions $\p, \q$).
For example, the $\lp[2]$-tester of~\cite{CDVV14} (which has been used as a primitive to obtain
a wide range of sample-optimal testers~\cite{DK:16}) is an unbiased estimator
for the squared $\lp[2]$-distance between two distributions $\p, \q$ over $[n]$.
While the description of such unbiased estimators may be relatively simple, 
their analyses are typically highly non-trivial. Specifically, obtaining tight bounds for the variance of such 
estimators has been a major technical hurdle in distribution testing. 
As an important contribution of this work, we develop a general theory 
providing tight variance bounds for \emph{all} such estimators.

The conditional independence tester~\cref{thm:binary-informal} 
straightforwardly extends to larger domains $\domx, \domy$, alas
its sample complexity becomes at least linear in the size of these sets. 
To obtain a sublinear tester for this general case,
we require a number of additional conceptual and technical ideas. 
Our main theorem for conditional independence testing for 
domain $[\ell_1] \times [\ell_2] \times [n]$ is the following:

\begin{theorem} \label{thm:general-algo-informal}
There exists a computationally efficient tester for $\mathcal{T}(\ell_1, \ell_2, n, \eps)$ with sample complexity 
\begin{equation} \label{eqn:crazy}
      \bigO{ \max\mleft(
                \min\mleft( \frac{n^{7/8}\ell_1^{1/4}\ell_2^{1/4}}{\eps},\frac{n^{6/7}\ell_1^{2/7}\ell_2^{2/7}}{\eps^{8/7}} \mright),
                 \frac{n^{3/4}\ell_1^{1/2}\ell_2^{1/2}}{\eps},
                \frac{n^{2/3}\ell_1^{2/3}\ell_2^{1/3}}{\eps^{4/3}},
                 \frac{n^{1/2}\ell_1^{1/2} \ell_2^{1/2}}{\eps^2}
            \mright)  } \;,
\end{equation}
where we assume without loss of generality that $\ell_1 \geq \ell_2$.
\end{theorem}

The expression of the sample complexity in~\cref{thm:general-algo-informal} may seem somewhat unwieldy. 
In an attempt to interpret this bound, we consider several important special cases of interest:
\begin{itemize}
\item For $\ell_1=\ell_2=O(1)$, \eqref{eqn:crazy} reduces to the binary case for $X,Y$, 
  recovering the tight bound of~\cref{thm:binary-informal}.
\item For $n=1$, our problem reduces to the task of {\em testing 
independence} of a distribution over $[\ell_1] \times [\ell_2]$, which has been extensively studied
~\cite{BFFKRW:01,LRR11, ADK15, DK:16}. In this case, \eqref{eqn:crazy} recovers 
the optimal sample complexity of independence testing, 
i.e., $\bigTheta{\max\mleft( \ell_1^{2/3}\ell_2^{1/3}/\eps^{4/3},\sqrt{\ell_1\ell_2}/\eps^2 \mright)}$~\cite{DK:16}. 
\item For $\ell_1=\ell_2=n$ (and $\eps=\bigOmega{1}$), the sample complexity of \eqref{eqn:crazy} becomes $O(n^{7/4})$. 
In~\cref{thm:lb-nnn-informal} below, we show that this bound is optimal as well.
\end{itemize}

We conclude with the aforementioned tight sample lower bound for constant values of $\eps$, in the setting where all three coordinates
are of approximately the same cardinality:
\begin{theorem} \label{thm:lb-nnn-informal}
Any tester for $\mathcal{T}(n, n, n, 1/20)$ requires $\Omega(n^{7/4})$ samples.
\end{theorem}

\subsection{Our Techniques} \label{sec:techniques}
\subsubsection{Conditional Independence Tester for Binary \texorpdfstring{$\domx, \domy$}{X, Y}} \label{ssec:binary-techniques}

In the case where $\domx$ and $\domy$ are binary, for each bin $z \in \domz$ we 
will attempt to estimate the squared $\lp[2]$-distance of the corresponding 
conditional distribution and the product of its conditional marginals.
In particular, if $\domx = \domy = \{0, 1\}$ the square of $\p_{00}\p_{11} -
\p_{01}\p_{10}$, where $\p_{ij}$ is the probability that $X=i$ and $Y=j$, for $Z=z$, is
proportional to this difference. Since this square is a degree-$4$
polynomial in the samples, there is an unbiased estimator of this
quantity that can be computed for any value $z \in \domz$ from which we have at
least $4$ samples. Furthermore, for values of $z \in \domz$ for which we have 
more than $4$ samples, the additional samples can be used to reduce the error
of this estimator. The final algorithm computes 
a weighted linear combination of these estimators (weighted
so that the more accurate estimators from heavier bins are given more
weight) and compares it to an appropriate threshold. 
The correctness of this estimator requires a rather subtle
analysis. Recall that there are three different regimes of $\eps$ versus $n$
in the optimal sample complexity and the tester achieves this bound
without a case analysis. As usual, we require a bound
on the variance of our estimator and a lower bound
on the expectation in the soundness case. 

On the one hand, a naive bound on the variance for our estimator for
an individual bin turns out to be insufficient for our analysis. In
particular, let $\p$ be a discrete probability distribution and $Q(\p)$ a polynomial in
the individual probabilities of $\p$. Given $m \geq \deg(Q)$ independent
samples from $\p$, it is easy to see that there is a unique symmetric,
unbiased estimator for $Q(\p)$, which we call $U_m Q$. Our analysis will
rely on obtaining tight bounds for the variance of $U_m Q$. It is not hard
to show that this variance scales as $O(1/m)$, but this bound turns out to 
be insufficient for our purposes. In order to refine this estimate, we
show that $\var(U_m Q) = R(\p)/m + O(1/m^2)$, 
for some polynomial $R$ for which we devise a general formula. 
From this point on, we can show that for our polynomial $Q$ (or in
general any $Q$ which is the square of a lower degree polynomial)
$\var(U_m Q) = O(Q(\p)/m + 1/m^2)$. This gives a much sharper
estimate on the variance of our estimator, except in cases where the
mean is large enough that the extra precision is not necessary.

Another technical piece of our analysis is relating the mean of our
estimator to the total variation distance of our distribution from being
conditionally independent. In particular, our estimator is roughly the sum 
(over the $\domz$-bins with enough samples) of the
squared $\lp[2]$ distance that the conditional distribution is from being independent.
When much of the distance from conditionally independence comes from
relatively heavy bins, this relation is a more or less standard $\lp[1]/\lp[2]$
inequality. However, when the discrepancy is concentrated on very
light bins, the effectiveness of our tester is bounded by the number
of these bins which obtain at least four samples, and a somewhat
different analysis is required. In fact, out of the different cases
in the performance of our algorithm, one of the boundaries is
determined by a transition between the hard cases involving
discrepancies supported on light bins to ones where the discrepancy is
supported on heavy bins.

If the variables $X$ and $Y$ are no longer binary, our estimates for the
discrepancy of an individual bin must be updated. In particular, we
similarly use an unbiased estimator of the $\lp[2]$ distance between the
conditional distribution and the product of its conditional marginals. 
We note however that variance of this estimator is large 
if the marginal distributions have large $\lp[2]$ norms.
Therefore, in bins for which we have a large number of samples, we can
employ an idea from~\cite{DK:16} and use some of our samples
to artificially break up the heavier bins, thus flattening these
distributions. We elaborate on this case, and the required ingredients it entails, in the next subsection.

\subsubsection{General Conditional Independence Tester}  \label{ssec:general-techniques}

Assuming that we take at least four samples from any bin $z\in\domz$, 
we can compute an unbiased estimator for the squared $\lp[2]$ distance between $\p_z$,
the conditional distribution, and $\q_z$ the product of its conditional 
marginals. It is easy to see that this expectation is at least 
$\eps_z^2/(\abs{\domx}\abs{\domy})$, where $\eps_z$ is the $\lp[1]$ distance between the
conditional distribution and the closest distribution with independent
$X$ and $Y$ coordinates. At a high level, our algorithm takes a linear
combination of these bin-wise estimators (over all bins from which we
got at least $4$ samples), and compares it to an appropriate threshold.
There is a number of key ideas that are needed so that this approach 
gives us the right sample complexity.

Firstly, we use the idea of \emph{flattening}, introduced in~\cite{DK:16}. The
idea here is that the variance of the $\lp[2]$ estimator is larger if the $\lp[2]$
norms of $\p$ and $\q$ are large. However, we can reduce this variance by
artificially breaking up the heavy bins. In particular, if we have $m$
samples from a discrete distribution of support size $n$, we can artificially add $m$
bins and reduce the $\lp[2]$ norm of the distribution (in expectation) to at
most $O(1/\sqrt{m})$. We note that it is usually not a good idea to
employ this operation for $m \gg n$, as it will substantially increase the
number of bins. Nor do we want to use all of our samples for flattening
(since we need to use some for the actual tester). Trading off these considerations, 
using $\min(m/2,n)$ of our samples to flatten is a reasonable choice. 
We also remark that
instead of thinking of $\p$ and $\q$ as distributions over $\abs{\domx}\abs{\domy}$ bins, 
we exploit the fact that $\q$ is a two-dimensional product distribution over $\abs{\domx} \times \abs{\domy}$.
By flattening these marginal distributions
independently, we can obtain substantially better variance upper bounds.

Secondly, we need to use appropriate weights for our bin-wise
estimator. To begin with, one might imagine that the weight we should use 
for the estimator of a bin $z \in \domz$  
should be proportional to the probability mass of that bin. 
This is a natural choice because heavier bins will contribute more 
to the final $\lp[1]$ error, and thus, we
will want to consider their effects more strongly. The probability mass of a bin
is approximately proportional to the number of samples obtained from that bin.
Therefore, we might want to weight each bin by the number of
samples drawn from it. However, there is another important effect of
having more samples in a given bin. In particular, having more samples
from a bin allows us to do more flattening of that bin, which
decreases the variance of the corresponding bin-wise estimator. This means that we
will want to assign more weight to these bins based on how much
flattening is being done, as they will give us more accurate
information about the behavior of that bin.

Finally, we need to analyze our algorithm. 
Let $m$ be the number of samples we take and $n$ be the domain of $Z$.
If the bin weights are chosen appropriately, 
we show that the final estimator $A$ has variance 
$O(\min(n,m)+\sqrt{\min(n,m)}\shortexpect[A]+\shortexpect[A]^{3/2})$, 
with high probability over the number of samples falling in each bin as well as the flattening we perform for each bin. 
(This high-probability statement, in turn, is enough for us to apply Chebyshev's inequality to our final estimator in the end.) 
Furthermore, in the completeness case, we have that $\shortexpect[A]=0$. 
In order to be able to distinguish between completeness and soundness, 
we need it to be the case that for all distributions $\eps$-far from conditional independence
it holds that $\shortexpect[A] \gg \sqrt{\min(n,m)}$. 
We know that if we are $\eps$-far from conditional
independence, we must have that $\sum_z \eps_z w_z \gg \eps$, 
where $w_z$ is the probability that $Z=z$. 
In order to take advantage of this fact, we will need to separate the $Z$-bins 
into four categories based on the size of the $w_z$.
Indeed, if we are far from conditional independence, 
then for at least one of these cases 
the sum of $\eps_z w_z$ over bins of that type only will be $\gg \eps$. 
Each of these four cases will require a slightly different analysis:
\begin{itemize}
\item Case 1: $w_z < 1/m$. In this case, the expected number of samples from
bin $z$ is small. In particular, the probability of even seeing $4$
samples from the bin might will be small. Here, the expectation is
dominated by the probability that we see enough samples from the bin.

\item Case 2: $1/m < w_z < \abs{\domx}/m$: In this case, we are likely to get our $4$ samples
from the bin, but probably will get fewer than $\abs{\domx}$. This means that
our flattening will not saturate either of the marginal distributions
and we can reduce the squared $\lp[2]$ norm of $\q$ by a full factor of $m_z$
(where $m_z$ is the number of samples from this bin).

\item Case 3: $\abs{\domx}/m < w_z < \abs{\domy}/m$. In this case, we are likely to saturate our
flattening over the $X$-marginal but not the $Y$-marginal. Thus, our
flattening only decreases the $\lp[2]$ norm of the conditional distribution 
on that bin by a factor of $\sqrt{\abs{\domx} m_z}$.

\item Case 4: $\abs{\domy}/m < w_z$: Finally, in this case we saturate both the $X$- and $Y$-marginals, 
so our flattening decreases the $\lp[2]$ norm by a factor of $\sqrt{\abs{\domx} \abs{\domy}}$.
\end{itemize} 
Within each sub-case, the expectation of $A$ is a polynomial in 
$m, \abs{\domx}, \abs{\domy}$ multiplied by the sum over $z\in\domz$ of some polynomial 
in $\eps_z$ and $w_z$. We need to bound this from below 
given that $\sum_z \eps_z w_z \gg \eps$, and then set $m$
large enough so that this lower bound is more than $\sqrt{\min(n,m)}$. 
We note that only in Case 1 is the case where $m < n$ relevant. Thus, our final bound
will be a maximum over the $4$ cases of the $m$ required in the
appropriate case.

\subsubsection{Sample Complexity Lower Bound Construction for Binary \texorpdfstring{$\domx, \domy$}{X, Y}} \label{ssec:lb-techniques}

We begin by reviewing the lower bound methodology 
we follow: 
In this methodology, a lower bound is shown by adversarially constructing two 
distributions over pseudo-distributions. Specifically, we construct 
a pair of ensembles $\mathcal{D}$ and $\mathcal{D}'$ of pairs of 
nearly-normalized pseudo-distributions such that the following holds: (1) Pseudo-distributions 
drawn from $\mathcal{D}$ satisfy the desired property and 
pseudo-distribution drawn from $\mathcal{D}'$ are $\eps$-far from satisfying the property with high probability, 
and (2) $\poisson{s}$ samples are insufficient to reliably determine 
from which ensemble the distribution was taken from, unless $s$ is large enough. 

To formally prove our lower bounds, we will use the mutual information method, as in~\cite{DK:16}. 
In this section, we provide an intuitive description of our sample complexity
lower bound for testing conditional independence, when $\domx = \domy = \{0, 1\}$
and $\domz = [n]$. (Our lower bound for the regime  $\domx = \domy = \domz  = [n]$ 
is proved using the same methodology, but relies on a different construction.) 
We construct ensembles $\mathcal{D}$ and $\mathcal{D}'$~---~where draws from $\mathcal{D}$ 
are conditionally independent and draws from $\mathcal{D}'$ are $\eps$-far from conditionally independent 
with high probability~---~and show that $s$ samples from a distribution on
$(X,Y,Z)$ are insufficient to reliably distinguish whether the
distribution came from $\mathcal{D}$ or $\mathcal{D}'$, when $s$ is small. 
We define $\mathcal{D}$ and $\mathcal{D}'$ by treating each bin $z \in [n]$ of $Z$ independently. 
In particular, for each possible value $z \in [n]$ for $Z$, we proceed as follows:
(1) With probability $\min(s/n,1/2)$, we assign the point $Z = z$ probability
mass $\max(1/s,1/n)$ and let the conditional distribution on $(X,Y)$ be uniform. 
Since the distribution is conditionally independent on these bins 
and identical in both ensembles, these ``heavy'' bins will create ``noise'' to confuse an estimator. 
(2) With probability $1-\min(s/n,1/2)$, we set the probability that $Z=z$ to be $\eps/n$, 
and let the conditional distribution on $(X,Y)$ be taken from either $C$ or $C'$, for
some specific ensembles $C$ and $C'$. In particular, we pick $C$ and $C'$ 
so that a draw from $C$ is independent and a draw from $C'$ is far from independent. 
These bins provide the useful information that allows us to 
distinguish between the two ensembles $\mathcal{D}$ and $\mathcal{D}'$. 
\emph{The crucial property is that we can achieve the above while guaranteeing that any third moment from $C$ 
agrees with the corresponding third moment from $C'$.} This guarantee implies 
that if we draw $3$ (or fewer) samples of $(X,Y)$ from some bin $Z=z$, 
then the distribution on triples of $(X,Y)$ will be identical if the conditional 
was taken from $C$ or if it was taken from $C'$. That is, all information about
whether our distribution came from $\mathcal{D}$ or $\mathcal{D}'$ will come from bins of type
(2) (for which we have at least $4$ samples) of which there will be
approximately $n(s \eps/n)^4$. On the other hand, there will be about
$\min(s,n)$ bins of type (1) with $4$ samples in random configuration
adding ``noise''. Thus, we will not be able to distinguish reliably unless
$n(s \eps/n)^4 \gg \sqrt{\min(s,n)}$, as otherwise the noise'
due to the heavy bins will drown out the signal of the light ones.

To define $C$ and $C'$, we find appropriate vectors $p, q$ over $\{0,1\}^2$ 
so that $p+q$ and $p+3q$ each are distributions with independent 
coordinates, but $p, p+2q, p+4q$ are not. We let $C$ 
return $p+q$ and $p+3q$ each with probability
$1/2$, and let $C'$ return $p, p+2q$ or $p+4q$ with probability 
$1/8, 3/4, 1/8$ respectively. If we wish to find the probability that $3$ samples from a
distribution $r$ come in some particular pattern, we get $f(r)$ for some
degree-$3$ polynomial $f$. If we want the difference in these
probabilities for $r$ a random draw from $C$ 
and a random draw from $C'$, 
we get $f(p+q)/2+f(p+3q)/2 - f(p)/8 - f(p+2q)(3/4) - f(p+4q)/8$. 
We note that this is proportional to the fourth finite difference of a
degree-$3$ polynomial, and is thus $0$. Therefore, any combination of at
most $3$ samples are equally likely to show up for some $Z$-bin from $\mathcal{D}$ as
from $\mathcal{D}'$.

To rigorously analyze the above sketched construction, 
we consider drawing $\poisson{s}$ samples from a random distribution 
from either $\mathcal{D}$ or $\mathcal{D}'$, and bound the mutual information between 
the set of samples and the ensemble they came from. 
Since the samples from each bin are conditionally independent on
the ensemble, this is at most $n$ times the mutual information coming
from a single bin. By the above, the probabilities of seeing any triple
of samples are the same for either $\mathcal{D}$ or $\mathcal{D}'$ and thus contribute nothing
to the mutual information. For sets of $4$ or more samples, we note that
the difference in probabilities comes only from the case where $4$
samples are drawn from a bin of type (2), which happens with probability
at most $O(s \eps/n)^4$. However, this is counterbalanced by the fact
that these sample patterns are seen with much higher frequency from
bins of type (1) (as they have larger overall mass). Thus, the mutual
information for a combination including $m \geq 4$ samples will be 
$(O(s \eps/n)^m)^2 / \min(s/n,1/2) \cdot \Omega(1)^m$. 
The contribution from $m>4$ can be shown to be negligible, 
thus the total mutual information summed over all bins is
$O(\min(s,n) \cdot (s \eps/n)^8)$. This must be $\Omega(1)$ in order to reliably
distinguish, and this proves our lower bound.

\subsection{Organization} \label{sec:organ}
After setting up the required preliminaries in~\cref{sec:prelims}, 
we give our testing algorithm for the case of constant $\abs{\domx}, \abs{\domy}$ in~\cref{sec:alg-basic}. 
In~\cref{sec:polynomial}, we develop our theory for polynomial estimation.
\cref{sec:alg:flattened} leverages this theory, along with several other ideas,
to obtain our general testing algorithm for conditional independence. 
\cref{sec:lb:appendix} gives our information-theoretic lower bound for the setting of binary $\abs{\domx}, \abs{\domy}$.
In~\cref{sec:lb:nnn}, we give an information-theoretic lower bound matching the sample complexity of our algorithm 
for the regime where $\abs{\domx}=\abs{\domy}=\abs{\domz}$. 
In~\cref{sec:mutualinfo}, we discuss the implications of our results for 
conditional independence testing with regard to the conditional mutual information. 
 
\section{Preliminaries and Basic Facts} \label{sec:prelims}

We begin with some standard notation and definitions 
that we shall use throughout the paper. 
For $m \in \N$, we write $[m]$ for the set $\{1,\dots,m\}$, and $\log$ for the binary logarithm.

\paragraph{Distributions and Metrics} 
A probability distribution over discrete domain $\Omega$ 
is a function $\p\colon\Omega\to[0,1]$ such that $\normone{\p}\eqdef \sum_{\omega\in\Omega}\p(\omega)=1$.
Without the requirement that the total mass be one, 
$\p$ is said to be a \emph{pseudo-distribution}. 
We denote by $\distribs{\Omega}$ the set of all probability distributions over domain $\Omega$.
For two probability distributions $\p,\q\in\distribs{\Omega}$, 
their \emph{total variation distance} (or statistical distance) is defined as 
$
    \dtv(\p,\q) \eqdef \sup_{S\subseteq\Omega} (\p(S)-\q(S)) = \frac{1}{2}\sum_{\omega\in\Omega} \abs{\p(\omega)-\q(\omega)},
$
i.e., $\dtv(\p,\q) = \frac{1}{2}\normone{\p-\q}$,
and their $\lp[2]$ distance is the distance $\normtwo{\p-\q}$ 
between their probability mass functions. Given a subset $\mathcal{P}\subseteq \distribs{\Omega}$ of distributions, 
the \emph{distance from $\p$ to $\mathcal{P}$} is then defined as 
$\dtv(\p,\mathcal{P})\eqdef \inf_{\q\in\mathcal{P}} \dtv(\p,\q)$. 
If $\dtv(\p,\mathcal{P}) > \eps$, we say that $\p$ is \emph{$\eps$-far} 
from $\mathcal{P}$; otherwise, it is \emph{$\eps$-close}. 
For a distribution $\p$ we write $X\sim \p$ to denote 
that the random variable $X$ is distributed according to $\p$. 
Finally, for $\p\in\distribs{\Omega_1},\q\in\distribs{\Omega_2}$, 
we let $\p\otimes\q\in\distribs{\Omega_1\times\Omega_2}$ 
be the product distribution with marginals $\p$ and $\q$.

\paragraph*{Property Testing}
We work in the standard setting of distribution testing: 
a \emph{testing algorithm for a property $\mathcal{P}\subseteq\distribs{\Omega}$} 
is an algorithm which, granted access to independent samples from an 
unknown distribution $\p\in\distribs{\Omega}$ as well as distance 
parameter $\eps\in(0,1]$, outputs either \accept or \reject, with the 
following guarantees:
\begin{itemize}
  \item If $\p\in\mathcal{P}$, then it outputs \accept with probability at least $2/3$.
  \item If $\dtv(\p,\mathcal{P})>\eps$, then it outputs \reject with probability at least $2/3$.
\end{itemize}
The two measures of interest here are the \emph{sample complexity} 
of the algorithm (i.e., the number of samples it draws from the underlying distribution)
and its running time.

\subsection{Conditional Independence} 
We record here a number of notations definitions regarding conditional independence.
Let $X, Y, Z$ be random variables over discrete domains $\domx,\domy,\domz$ respectively.
Given samples from the joint distribution of $(X, Y, Z)$, we want to determine 
whether $X$ and $Y$ are {\em conditionally independent given $Z$}, denoted by $\condindrv{X}{Y}{Z}$,
versus $\eps$-far in total variation distance 
from every distribution of random variables $(X',Y',Z')$ such that $\condindrv{X'}{Y'}{Z'}$.
For discrete sets $\domx,\domy,\domz$, we will denote by $\condindprop{\domx}{\domy}{\domz}$ 
the property of conditional independence, i.e., 
$ 
\condindprop{\domx}{\domy}{\domz} \eqdef \setOfSuchThat{ \p\in\distribs{\domx\times\domy\times\domz} }{ (X,Y,Z)\sim \p \text{ satisfies } \condindrv{X}{Y}{Z} }
$.
We say that a distribution $\p \in\distribs{\domx\times\domy\times\domz}$ is $\eps$-far from $\condindprop{\domx}{\domy}{\domz}$, 
if for every distribution $\q \in \condindprop{\domx}{\domy}{\domz}$ we have that $\totalvardist{\p}{\q} > \eps$. 
Fix a distribution $\q \in \condindprop{\domx}{\domy}{\domz}$ of minimum total variation distance to $p$. 
Then the marginals of $\q$ on each of the three coordinates may have different distributions.
We will also consider testing conditional independence with respect to a different metric, 
namely the~\emph{conditional mutual information}~\cite{Dobrushin59,Wyner78}. 
For three random variables $X,Y,Z$ as above, the conditional mutual information of $X$ and $Y$ 
with respect to $Z$ is defined as
$\condmutualinfo{X}{Y}{Z} \eqdef \shortexpect_{Z}[ (\mutualinfo{X}{Y} \mid Z ) ]$,
i.e., as the expected (with respect to $Z$) K-L divergence 
between the distributions of $(X,Y)\mid Z$ and the product of the distributions 
of $(X\mid Z)$ and $(Y\mid Z)$. In this variant of the problem (considered in~\cref{sec:mutualinfo}), 
we will want to distinguish $\condmutualinfo{X}{Y}{Z}=0$ from $\condmutualinfo{X}{Y}{Z}\geq \eps$.

\paragraph{Notation.}
Let $\p\in\distribs{ \domx\times\domy\times\domz }$. 
For $z\in\domz$, we will denote by $\p_z \in \distribs{ \domx\times\domy}$ the distribution
defined by $\p_z(i,j) \eqdef \probaDistrOf{(X,Y,Z)\sim p}{ X=i, Y=j \mid Z = z}$
and by $\p_Z\in\distribs{\domz}$ the distribution $\p_Z(z) \eqdef \probaDistrOf{(X,Y,Z)\sim p}{ Z = z}$.
By definition, for any $\p\in\distribs{ \domx\times\domy\times\domz }$, we have that
$p(i, j, z) = \p_Z(z) \cdot \p_z(i,j)$.
For $z\in\domz$, we will denote by $\p_{z,X}\in \distribs{ \domx}$ the distribution
$\p_{z,X}(i) =  \probaDistrOf{(X,Y,Z)\sim p}{ X=i \mid Z = z}$ and 
$\p_{z,Y}\in \distribs{ \domy}$ the distribution
$\p_{z,Y}(j) =  \probaDistrOf{(X,Y,Z)\sim p}{ Y=j \mid Z = z}$.

\medskip

We can now define the product distribution of the conditional marginals:
\begin{definition}[Product of Conditional Marginals]\label{def:product:conditional:marginals}
Let $\p\in\distribs{ \domx\times\domy\times\domz }$. For $z\in\domz$, we define the 
\emph{product of conditional marginals of $\p$ given $Z = z$} to be the product distribution
$\q_z\in\distribs{ \domx\times\domy}$ defined by $\q_z \eqdef \p_{z,X} \otimes \p_{z,Y}$, 
i.e., $\q_z(i,j) = \p_{z,X}(i) \cdot \p_{z,Y}(j)$.
We will also denote by $q$ the mixture of product distributions 
$\q \eqdef \sum_{z\in\domz}\p_Z(z)\q_z \in\condindprop{\domx}{\domy}{\domz}$, i.e.,
$\q(i,j,z) \eqdef \p_Z(z)\cdot \q_z(i,j)$.
\end{definition}

\subsection{Basic Facts} \label{ssec:basics}
  
We start with the following simple lemma:
\begin{lemma} \label{lem:dtv-cond}
Let $\p,\p'\in\distribs{\domx\times\domy\times\domz}$.
Then we have that 
\begin{equation}\label{eq:useful:decomposition}
 \totalvardist{\p}{\p'} \leq \sum_{z\in\domz} \p_Z(z) \cdot \totalvardist{\p_z}{\p'_z} + \totalvardist{\p_Z}{\p'_Z} \;,
\end{equation}
with equality if and only if $\p_Z=\p'_Z$.
\end{lemma}  
  
Using~\cref{lem:dtv-cond}, we deduce the following useful corollary:
   
\begin{fact}\label{lemma:distance:either}
If $\p\in \distribs{\domx\times\domy\times\domz}$ is $\eps$-far from $\condindprop{\domx}{\domy}{\domz}$, 
then, for every $\p'\in\condindprop{\domx}{\domy}{\domz}$, either 
(i) $\totalvardist{\p_Z}{\p'_Z} > \eps/2$, or 
(ii) $\sum_{z\in\domz} \p_Z(z) \cdot \totalvardist{\p_z}{\p'_z} > \eps/2$.
\end{fact}
\begin{proof}
Let $\p'\in\condindprop{\domx}{\domy}{\domz}$. Since $\p$ is $\eps$-far from $\condindprop{\domx}{\domy}{\domz}$
we have that $\totalvardist{\p}{\p'} > \eps$. By~\cref{lem:dtv-cond}, we this obtain that 
$\sum_{z\in\domz} \p_Z(z) \cdot \totalvardist{\p_z}{\p'_z} + \totalvardist{\p_Z}{\p'_Z} > \eps$, 
which proves the fact.
\end{proof}
    
The next lemma shows a useful structural property of conditional independence that
will be crucial for our algorithm. It shows that if a distribution $\p\in \distribs{\domx\times\domy\times\domz}$ 
is close to being conditionally independent, then it is also close to an appropriate mixture 
of its products of conditional marginals, 
specifically distribution $q$ from~\cref{def:product:conditional:marginals}:
\begin{lemma}\label{lemma:distance:product}
Suppose $\p\in\distribs{\domx\times\domy\times\domz}$ is $\eps$-close to $\condindprop{\domx}{\domy}{\domz}$. 
Then, $\p$ is 4\eps-close to the distribution $\q  = \sum_{z\in\domz}\p_Z(z) \q_z$.
\end{lemma}

\subsection{Flattening Distributions}\label{ssec:flattening}
We now recall some notions and lemmata from previous work, 
regarding the technique of \emph{flattening} of discrete distributions:
\begin{definition}[Split distribution {\cite{DK:16}}]\label{def:split:distribution}
Given a distribution $\p\in\distribs{[n]}$ and a multi-set $S$ 
of elements of $[n]$, define the \emph{split distribution} 
$\p_S\in\distribs{[n+|S|]}$ as follows:
For $1\leq i\leq n$, let $a_i$ denote $1$ plus the number 
of elements of $S$ that are equal to $i$.
Thus, $\sum_{i=1}^n a_i = n+|S|.$ 
We can therefore associate the elements of $[n+|S|]$ 
to elements of the set 
$B_S\eqdef\setOfSuchThat{(i,j) }{ i\in [n], 1\leq j \leq a_i }$.
We now define a distribution $\p_S$ with support $B_S$, 
by letting a random sample from $\p_S$ be given by $(i,j)$,
where $i$ is drawn randomly from $\p$ and $j$ is drawn uniformly from $[a_i]$.
\end{definition}

\begin{fact}[{\cite[Fact 2.5]{DK:16}}]\label{fact:split:distributions:distance}
Let $\p,\q\in\distribs{[n]}$, and $S$ a given multi-set of $[n]$. Then:
(i) We can simulate a sample from $\p_S$ or $\q_S$ by taking a single sample from $\p$ or $\q$, respectively.
(ii) It holds $\totalvardist{\p_S}{\q_S} = \totalvardist{\p}{\q}$.
\end{fact}

\noindent We will also require the analogue of~\cite[Lemma 2.6]{DK:16} (how flattening reduces the $\lp[2]$-norm of a distribution) 
for the non-Poissonized setting, i.e., when exactly $m$ samples are drawn (instead of $\poisson{m}$). 
The proof of this lemma is similar to that of~\cite[Lemma 2.6]{DK:16}, 
and we include it in~\cref{sec:deferred:prelims} for completeness:
\begin{lemma}\label{fact:split:distributions:l2norm:nonpoisson}
Let $\p\in\distribs{[n]}$. Then: 
(i) For any multi-sets $S\subseteq S'$ of $[n]$, $\normtwo{\p_{S'}} \leq \normtwo{\p_{S}}$, and
(ii) If $S$ is obtained by taking $m$ independent samples from $\p$, then $\expect{[\normtwo{\p_S}^2} \leq \frac{1}{m+1}$.
\end{lemma}

\begin{remark}\label{remark:split:l2:chi2}
Given $S$ and $(a_i)_{i\in[n]}$ as in~\cref{def:split:distribution}, it is immediate that for any $\p,\q\in\distribs{[n]}$
it holds $\normtwo{\p_S-q_S}^2 = \sum_{i=1}^n \frac{(\p_i-\q_i)^2}{a_i}$
so that an $\lp[2]^2$ statistic for $\p_S,\q_S$ can be seen as a particular rescaled $\lp[2]$ statistic for $\p,\q$. 
\end{remark}

\subsection{Technical Facts on Poisson Random Variables}\label{ssec:poisson:technical}

  We state below some technical result on moments of truncated Poisson random variables, 
  which we will use in various places of our analysis. 
  The proof of these claims are deferred to~\cref{sec:deferred:prelims}.
  \begin{claim}\label{claim:variance:truncated:poisson}
  There exists an absolute constant $C>0$ such that, for $N\sim\poisson{\lambda}$,
  \[
    \var[ N \indic{ N \geq 4}] \leq C\expect{N \indic{ N \geq 4}} \;.
  \]
  Moreover, one can take $C=4.22$.
  \end{claim}
  
  \begin{claim}\label{claim:variance:truncated:poisson:2}
  There exists an absolute constant $C>0$ such that, for $X\sim\poisson{\lambda}$ and $a,b\geq 0$,
  \[
  \var[ X \sqrt{\min(X,a)\min(X,b)}\indic{X\geq 4} ] \leq C\expect{X \sqrt{\min(X,a)\min(X,b)}\indic{X\geq 4} } \;.
  \]
  \end{claim}
  
  \begin{claim}\label{claim:expectation:truncated:poisson:squared:with:min}
  There exists an absolute constant $C>0$ such that, for $X\sim\poisson{\lambda}$ and integers $a,b\geq 2$,
  \[
      \expect{ X \sqrt{\min(X,a)\min(X,b)}\indic{X\geq 4}  } \geq C \min(\lambda \sqrt{\min(\lambda,a)\min(\lambda,b)}, \lambda^4) \;.
  \]
  \end{claim}

\section{Conditional Independence Tester: The Case of Constant \texorpdfstring{$|\domx|, |\domy|$}{|X|, |Y|}} \label{sec:alg-basic}
  
Let $\p\in\distribs{\domx\times\domy\times z}$.
In this section, we present and analyze our conditional independence tester for the case that 
$|\domx|, |\domy| = O(1)$. Specifically, we will present a tester for this regime 
whose sample complexity is optimal, up to constant factors.  
Our tester uses as a black-box an unbiased estimator for the $\ell_2^2$-distance
between a $2$-dimensional distribution and the product of its marginals.
Specifically, we assume that we have access to an estimator $\Phi$ with the following performance:
Given $N$ samples $s=(s_1,\dots,s_N)$ from a distribution $\p\in\distribs{\domx\times\domy}$, 
$\Phi$ satisfies:
\begin{align}
\expect{\Phi(s)} &= \normtwo{\p - \p_{\domx}\otimes\p_{\domy}}^2 \label{eq:estimator:expect}\\
\var[ \Phi(s) ] &\leq C\left(\frac{\expect{\Phi(s)}}{N}+\frac{1}{N^2}\right) \label{eq:estimator:variance} \;,
\end{align}
for some absolute constant $C>0$. Such an estimator follows as a special case of our generic
polynomial estimators of~\cref{sec:polynomial}.

\paragraph{Notation}
Let $\p\in\distribs{\domx\times\domy}$. We denote its marginal distributions by 
$\p_{\domx}$, $\p_{\domy}$. That is, we have that $\p_{\domx} \in \distribs{\domx}$ 
with $\p_{\domx}(x) \eqdef \probaDistrOf{(X,Y)\sim p}{ X = x}$, $x \in \domx$, 
and similarly for $\p_{\domy}$. 
Let $\p\in\distribs{\domx\times\domy\times \domz}$. 
For $z \in \domz$, we will denote by $\q_z$ the product distribution $\p_{z, X} \otimes \p_{x, Y}$.

\medskip

Let $M$ be a $\poisson{m}$ random variable representing the number of samples drawn from 
$\p\in\distribs{\domx\times\domy \times \domz}$. 
Given the multi-set $S$ of $M$ samples drawn from $\p$, let 
$S_z \eqdef \setOfSuchThat{(x,y)}{(x,y,z)\in S} $
denote the multi-set of pairs $(x, y) \in \domx\times\domy$ 
corresponding to samples $(x, y, z) \in S$, 
i.e., the multi-set of samples coming from the conditional distribution $\p_z$. 
For convenience, we will use the notation $\sigma_z \eqdef \abs{S_z}$.
Let  
\[
A_z \eqdef \sigma_z \cdot \Phi(S_z)\cdot \indic{\sigma_z\geq 4} \;,
\]
for all $z\in\domz$. Our final \new{statistic (that we will compare to a suitable threshold in the eventual test)} is 
\[
  A\eqdef \sum_{z\in\domz} A_z \;.
\]
We set $\eps' \eqdef \frac{\eps}{\sqrt{\abs{\domx}\abs{\domy}}} = \Theta(\eps)$, and choose
\begin{equation}\label{eq:m:choice}
m\geq  \beta\max\left(\sqrt{n}/{\eps'}^2,\min\mleft(n^{7/8}/{\eps'},n^{6/7}/{\eps'}^{8/7}\mright)\right) \;,
\end{equation}
for a sufficiently large absolute constant $\beta >0$.

Interestingly enough, there are three distinct regions for this expression, based on the relation between $n$ and $\eps$, as illustrated
in the following figure:
\begin{figure}[H]\centering
    \includegraphics[width = 1.00\textwidth]{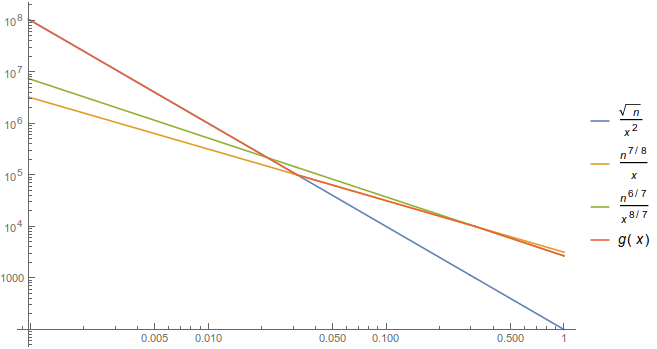}
          \caption{The three regimes of the sample complexity (for $n=10000$ and $\eps\in(0,1]$) (log-log plot).}
\end{figure}

Our conditional independence tester outputs ``\accept'' if $A \geq \tau$ and ``\reject'' otherwise,
where $\tau$ is selected to be \new{$\bigTheta{\sqrt{\min(n,m)}}$}. A detailed pseudo-code for the algorithm is given
in~\cref{algo:testing}. 
\begin{algorithm}[ht]
    \begin{algorithmic}[1]
      \Require Parameter $n\eqdef \abs{\domz}$, $\ell_1\eqdef \abs{\domx}$, $\ell_2\eqdef \abs{\domy}$, $\eps\in(0,1]$, 
                    and sample access to $\p\in\distribs{\domx\times\domy\times\domz}$.
      \State Set $m\gets \beta\max\left(\sqrt{n}/{\eps'}^2,\min\mleft(n^{7/8}/{\eps'},n^{6/7}/{\eps'}^{8/7}\mright)\right)$, 
                 where $\eps'\eqdef{\eps}/{\sqrt{\ell_1\ell_2}}$ \Comment{$\beta \geq 1$ is a sufficiently large absolute constant}
      \State Set \new{$\tau \gets \zeta\sqrt{\min(n,m)}$}, 
                 where \new{$\zeta>0$ is a sufficiently small constant}. \Comment{Threshold for accepting}
      \State Draw $M\sim \poisson{m}$ samples from $\p$ and let $S$ be the multi-set of samples.
      \ForAll{ $z\in\domz$ }
        \State Let $S_z\subseteq \domx\times\domy$ be the multi-set $S_z\eqdef \setOfSuchThat{(x,y)}{(x,y,z)\in S}$.
        \If{ $\abs{S_z} \geq 4$ } \Comment{Enough samples to call $\Phi$}
          \State Compute $\Phi(S_z)$.
          \State Set $A_z\gets \abs{S_z} \cdot \Phi(S_z)$.
        \Else
          \State Set $A_z\gets 0$.
        \EndIf
      \EndFor
      \If{ $A\eqdef \sum_{z\in \domz} A_z \leq \tau$ }
        \State \Return \accept
      \Else
        \State \Return \reject
      \EndIf
    \end{algorithmic}
    \caption{\textsc{TestCondIndependence}}\label{algo:testing}
  \end{algorithm}

\subsection{Proof of Correctness} \label{ssec:alg-analysis}
In this section, we prove correctness of~\cref{algo:testing}.
Specifically, we will show that: (1) If $\p \in \condindprop{\domx}{\domy}{\domz}$ (completeness),
 then~\cref{algo:testing} outputs ``\accept'' with probability at least $2/3$, and 
 (2) If $\dtv(\p, \condindprop{\domx}{\domy}{\domz} )>\eps$, then 
~\cref{algo:testing} outputs ``\reject'' with probability at least $2/3$. 
The proof proceeds by analyzing the expectation and variance of our statistic $A$ 
and using Chebyshev's inequality. We note that $\beta, \zeta$ are absolute constants defined in the algorithm pseudo-code.

\subsubsection{Analyzing the Expectation of \texorpdfstring{$A$}{A}} \label{ssec:exp-gap}

The main result of this subsection is the following proposition establishing 
the existence of a gap in the expected value of $A$ 
in the completeness and soundness cases:
  
\begin{proposition}\label{prop:exp-gap}
We have the following: (a) If $\p \in \condindprop{\domx}{\domy}{\domz}$, then $\expect{A}=0$. 
(b) If $\totalvardist{\p}{ \condindprop{\domx}{\domy}{\domz}}> \eps$, 
then $\expect{A} > \gamma \min\left( m{\eps'}^2, \frac{m^4{\eps'}^4}{8n^3}\right) 
\geq \frac{\beta \cdot \gamma}{8} \cdot \sqrt{\min(n,m)}$, for some absolute constant $\gamma>0$.
\end{proposition}
  
The rest of this subsection is devoted to the proof of~\cref{prop:exp-gap}.
We start by providing a convenient lower bound on the expectation of $A$.
We prove the following lemma:

\begin{lemma} \label{lemma:exp-lb}
For $z\in\domz$, let $\delta_z\eqdef \normtwo{\p_z - \q_z}$ and $\alpha_z\eqdef m \cdot \p_Z(z)$. 
Then, we have that:
\begin{equation}\label{eq:lowerbound:expectation}
\expect{A} \geq \gamma \cdot \sum_{z\in\domz} \delta_z^2 \min(\alpha_z, \alpha_z^4)\;.
\end{equation}
\end{lemma} 
\begin{proof}
Conditioned on $\sigma_z = |S_z|$, \cref{eq:estimator:expect} gives that 
$\expectCond{ A_z }{ \sigma_z } = \sigma_z \delta_z^2 \indic{\sigma_z\geq 4}$.
Therefore, for $\sigma\eqdef (\sigma_z)_{z\in\domz}$, we can write
$\expectCond{ A }{ \sigma } = \sum_{z\in\domz} \sigma_z \delta_z^2 \indic{\sigma_z\geq 4}$. 
Using the fact that the $\sigma_z$'s are independent and $\sigma_z\sim\poisson{\alpha_z}$, 
we obtain the following closed-form expression for the expectation:
\begin{equation} \label{eqn:exp-closed}
\expect{A} = \expect{\expectCond{ A }{ \sigma }}_{\sigma} 
= \sum_{z\in\domz} \delta_z^2 \expect{\sigma_z \indic{\sigma_z\geq 4}}
= \sum_{z\in\domz} \delta_z^2 \cdot f(\alpha_z) \;,
\end{equation}
where $f \colon \R_+ \to\R$ is the function 
$f(x) = e^{-x}\sum_{k=4}^\infty k\frac{x^k}{k!} = x - e^{-x}(x+x^2+\frac{x^3}{2}) \;.$ 
Let $g\colon \R_+ \to\R$ be defined by $g(x) = \min(x,x^4) \;.$
It is not hard to check that the function $f(x)/g(x)$ achieves its minimum at $x=1$, 
where it takes the value $\gamma\eqdef 1-\frac{5}{2e} > 0$. That is, $f(\alpha_z) \geq \gamma \cdot g(\alpha_z)$
and the lemma follows from~\eqref{eqn:exp-closed}.
\end{proof} 
  
Given~\eqref{eqn:exp-closed}, the first statement of~\cref{prop:exp-gap} is immediate.
Indeed, if $\p$ is conditionally independent, then all $\delta_z$'s are zero. To establish 
the second statement, we will require a number of intermediate lemmata.
We henceforth focus on the analysis of the soundness case, i.e., we will assume that 
$\totalvardist{\p}{ \condindprop{\domx}{\domy}{\domz}}> \eps$.
We require the following useful claim:

\begin{claim}\label{claim:sound-ineq}
If $\totalvardist{\p}{ \condindprop{\domx}{\domy}{\domz}}> \eps$, then 
$\sum_{z\in\domz} \delta_z \alpha_z > 2m\eps'$.
\end{claim}
\begin{proof}
We use the identity
$\totalvardist{\p}{\q}  = \sum_{z\in\domz} \p_Z(z)  \cdot \totalvardist{\p_z}{\q_z}$,
which follows from~\cref{lem:dtv-cond} noting that $\q_Z  = \p_Z$. By assumption,
we have that $\totalvardist{\p}{\q}>\eps$.
We can therefore write
\[ \sum_{z\in\domz} \delta_z \alpha_z = m\sum_{z\in\domz} \normtwo{\p_z - \q_z} \cdot \p_Z(z) 
\geq \frac{m}{\sqrt{\abs{\domx}\abs{\domy}}} \cdot \sum_{z\in\domz} \normone{\p_z - \q_z} \p_Z(z) 
>  \frac{2m}{\sqrt{\abs{\domx}\abs{\domy}}}\eps \;,\]
where the last inequality is Cauchy--Schwarz.
\end{proof}

\cref{lemma:exp-lb} suggests the existence of two distinct regimes:
the value of the expectation of our statistic is dominated by 
(1) the ``heavy'' elements $z \in \domz$ for which $\alpha_z > 1$, or 
(2) the ``light'' elements $z \in \domz$ for which $\alpha_z\leq 1$. 
Formally, let $\domz_H \eqdef \setOfSuchThat{z\in\domz}{\alpha_z> 1}$ and 
$\domz_L \eqdef \setOfSuchThat{z\in\domz}{\alpha_z\leq 1}$, 
so that
\begin{equation} \label{eqn:hl}
  \sum_{z\in\domz} \delta_z^2 \min(\alpha_z, \alpha_z^4) = \sum_{z\in\domz_H} \delta_z^2 \alpha_z + \sum_{z\in\domz_L} \delta_z^2 \alpha_z^4\;.
\end{equation}
By~\cref{claim:sound-ineq}, at least one of the following holds:
(1) $\sum_{z\in\domz_H} \delta_z \alpha_z > m\eps'$ 
or (2) $\sum_{z\in\domz_L} \delta_z \alpha_z > m\eps'$.
We analyze each case separately.
\begin{enumerate}[(1)]
\item Suppose that  $\sum_{z\in\domz_H} \delta_z \alpha_z > m\eps'$. 
We claim that $\expect{A} > \gamma \cdot m{\eps'}^2$.
Indeed, this follows from~\cref{lemma:exp-lb} and \eqref{eqn:hl} 
using the following chain of (in)-equalities:
\[
\sum_{z\in\domz_H} \delta_z^2 \alpha_z \geq \frac{\left(\sum_{z\in\domz_H} \delta_z\alpha_z\right)^2 }{\sum_{z\in\domz_H} \alpha_z}
> m{\eps'}^2 \;,
    \]
where the first inequality is Cauchy--Schwarz, and the second follows 
using that $\sum_{z\in\domz_H} \alpha_z \leq m$.

\item Suppose that $\sum_{z\in\domz_L} \delta_z \alpha_z > m\eps'$.
We claim that $\expect{A} > \gamma \cdot \frac{m^4{\eps'}^4}{8n^3}$.
Indeed, this follows from~\cref{lemma:exp-lb} and \eqref{eqn:hl} 
using the following chain of (in)-equalities:

\[
\sum_{z\in \domz_L} \delta_z^2 \alpha_z^4 
\geq \frac{1}{8n^3} \Big(\sum_{z\in \domz_L}\delta_z \alpha_z\Big)^4
> \frac{1}{8}\frac{m^4{\eps'}^4}{n^3} \;.
\]
The first inequality essentially follows by an application Jensen's inequality as follows: 
Let $\delta \eqdef \sum_{z\in \domz_L} \delta_z^{2/3}$.
By Jensen's inequality we have:
\[
\Big(\sum_{z\in \domz_L} (\delta_z^{2/3}/\delta) \cdot \delta_z^{1/3}\alpha_z\Big)^4 
\leq \sum_{z\in \domz_L} (\delta_z^{2/3}/\delta) \cdot \delta_z^{4/3}\alpha_z^4 \;,
\]
or 
\[
  \Big(\sum_{z\in \domz_L} \delta_z \alpha_z\Big)^4 \leq \delta^3 \sum_{z\in \domz_L} \delta_z^2 \alpha_z^4 \;.
\]
Since $\delta_z \leq 2$ for all $z \in \domz$, it follows that $\delta \leq 2n$, 
which completes the proof of the claim.
\end{enumerate}
We have thus far established that 
\[
  \expect{A} > \gamma \min\left( m{\eps'}^2, \frac{m^4{\eps'}^4}{8n^3}\right) \;,
\]
giving the first inequality of~\cref{prop:exp-gap} (b).
To complete the proof of the proposition, it suffices to show that
\[
  \min\left(m{\eps'}^2, \frac{m^4{\eps'}^4}{n^3}\right)\gg \sqrt{\min(n,m)} \;.
\]
We show this below by considering the following cases:
\begin{itemize}
\item If $n\geq \beta m$, we must be in the range $1/n^{1/8}\leq \eps' \leq 1$ where $\max(\sqrt{n}/{\eps'}^2,\min(n^{7/8}/{\eps'},n^{6/7}/{\eps'}^{8/7}))=n^{6/7}/{\eps'}^{8/7}$. We get $\frac{m^4{\eps'}^4}{n^3}\leq \beta^3\frac{m^4{\eps'}^4}{n^3} \leq m{\eps'}^4 \leq m{\eps'}^2$, 
and then since $\frac{m^4{\eps'}^4}{n^3} \geq \beta^{7/2}\sqrt{m}$ by our choice of $m$ in~\cref{eq:m:choice}, 
we get that
\[
  \min\mleft(m{\eps'}^2, \frac{m^4{\eps'}^4}{n^3}\mright)\geq \sqrt{\min(n,m)} \;,
\]
as desired assuming that $\beta \geq 1$.
\item If $\beta m\geq n$, we must be in the range $0< \eps' \leq 1/n^{1/8}$, and therefore $\min(n^{7/8}/{\eps'},n^{6/7}/{\eps'}^{8/7})=n^{7/8}/{\eps'}$. Since $m{\eps'}^2 \geq \beta\sqrt{n}$ and 
$\frac{m^4{\eps'}^4}{n^3} \geq \beta^{4}\sqrt{n}$ by our choice of $m$ in~\cref{eq:m:choice}, 
we get that
\[
  \min\mleft(m{\eps'}^2, \frac{m^4{\eps'}^4}{n^3}\mright)\geq \sqrt{\min(n,m)} \;,
\]
as desired assuming that $\beta \geq 1$.
\end{itemize}

This completes the proof of~\cref{prop:exp-gap}. \qed
  
\subsubsection{Analyzing the Variance of \texorpdfstring{$A$}{}} \label{ssec:var-bound}  

The main result of this subsection is the following proposition establishing 
an upper bound on the variance of $A$ as a function of its expectation: 
  
\begin{proposition}\label{prop:var-ub}
We have that
\begin{equation}\label{eqn:bound:variance}
\var [A] \leq C''\left(\min(n,m) + \expect{ A } \right) \;,
\end{equation}
for some absolute constant $C''$.
\end{proposition}

The rest of this subsection is devoted to the proof of~\cref{prop:var-ub}.
By the law of total variance, we have that:
\[
  \var A = \expect{ \var[ A \mid \sigma ] } + \var \expectCond{ A }{ \sigma } \;.
\]
We will proceed to bound each term from above, which will give the proof.
We start with the first term.
Conditioned on $\sigma_z = |S_z|$,~\cref{eq:estimator:variance} gives that
\[
\var[ A_z \mid \sigma_z ] \leq C\sigma_z^2 \left(\frac{\delta_z^2}{\sigma_z}+\frac{1}{\sigma_z^2}\right) \indic{\sigma_z\geq 4}
= C\left(1+\expectCond{ A_z }{ \sigma_z } \right) \indic{\sigma_z\geq 4} \;.
\]
Therefore, for $\sigma\eqdef (\sigma_z)_{z\in\domz}$, we can write
\begin{equation}\label{eq:cond:var}
\var[ A \mid \sigma ] \leq C\left(\min(n,M) + \expectCond{ A }{ \sigma } \right) \;,
\end{equation}
where we used the inequality 
$\sum_{z\in\domz} \indic{\sigma_z\geq 4} \leq \sum_{z\in\domz} \indic{\sigma_z\geq 1} \leq \min(n,M)$.
From~\cref{eq:cond:var}, we immediately get
\[
\expect{ \var[ A \mid \sigma ] }  \leq C\left(\min(n,m) + \expect{ A } \right) \;,
\]
as desired.

We now proceed to bound the second term.
As shown in~\cref{lemma:exp-lb}, 
$\expectCond{ A }{ \sigma } = \sum_{z\in\domz} \sigma_z \delta_z^2 \indic{\sigma_z\geq 4}$.
By the independence of the $\sigma_z$'s, we obtain that
\begin{equation} \label{eqn:ve}
\var \left[ \expectCond{ A }{ \sigma } \right] = \sum_{z\in\domz} \delta_z^4 \var[\sigma_z \indic{\sigma_z\geq 4}] \;.
\end{equation}
From \eqref{eqn:ve} and~\cref{claim:variance:truncated:poisson}, recalling that $\delta_z\leq 2$, $z \in \domz$, 
we get that 
\[
  \var \left[ \expectCond{ A }{\sigma} \right] \leq 4C'\sum_{z\in\domz} \delta_z^2 \expect{\sigma_z \indic{\sigma_z\geq 4}} = 4C'\expect{A} \;.
\]
This completes the proof of~\cref{prop:var-ub}.

\subsubsection{Completing the Proof} \label{ssec:cheb}
Recall that the threshold of the algorithm is defined to be  
\new{$\tau \eqdef \zeta\sqrt{\min(n,m)}$}. 

In the completeness case, by~\cref{prop:exp-gap} (a),
we have that $\expect{A}=0$.~\cref{prop:var-ub} then gives 
that $\var [A] \leq C'' \cdot \min(n,m)$. Therefore, by Chebyshev's inequality we obtain
\[
\probaOf{ A \geq \tau } \leq \frac{\var [A]}{\tau^2} \leq \frac{1}{\zeta^2}C''\frac{\min(n,m)}{\min(n,m)} \leq \frac{1}{3} \;,
\]
where the last inequality follows by choosing the \new{constant $\zeta$ to be sufficiently small (compared to $C''$).}

In the soundness case, by Chebyshev's inequality \new{and recalling the lower bound on $\expect{A}$ from~\cref{prop:exp-gap} (b) (which implies $\tau \leq 8\frac{\zeta}{\beta\gamma}\expect{A}\leq \expect{A}/2$ as long as $\zeta$ is chosen sufficiently small)} we get
\[
\probaOf{ A < \tau } \leq \probaOf{ \abs{A-\expect{A}} \geq \expect{A}/2 } \leq 4\frac{\var [A]}{\expect{A}^2}
\leq 4C''\left(\frac{\min(n,m)}{\expect{A}^2} + \frac{1}{\expect{A}} \right) \leq \frac{1}{3} \;,
\]
where the third inequality uses~\cref{prop:var-ub} and the fourth inequality uses~\cref{prop:exp-gap} (b), assuming $\beta$ is sufficiently large.
This completes the proof of correctness. \qed
 
\section{Estimating a Polynomial of a Discrete Distribution}\label{sec:polynomial}

In this section, we consider the following general problem: 
Given a degree-$d$ $n$-variate polynomial $Q\in\R_d[X_1,\dots,X_n]$ 
and access to i.i.d. samples from a distribution $\p\in\distribs{[n]}$, 
we want to estimate the quantity $Q(\p) = Q(\p_1,\dots,\p_n)$ to within 
an additive error $\eps$. In this section, we 
analyze an \emph{unbiased} estimator for $Q(\p)$ 
and provide quantitative bounds on its variance.

The structure of this section is as follows: In~\cref{ssec:unbiased-properties},
we describe the unbiased estimator and establish its basic properties.
In~\cref{ssec:var-bound:poly}, we bound from above the variance of the estimator.
Finally,~\cref{ssec:poly:l2} applies the aforementioned results to 
the setting that is relevant for our conditional independence tester.

\subsection{Unbiased Estimator and its Properties} \label{ssec:unbiased-properties}

We start by noting that the general case can be reduced to the case that the polynomial 
$Q$ is homogeneous.

\begin{remark}[Reduction to homogeneous polynomials]
  It is sufficient to consider, without loss of generality, the case where $Q\in\R_d[X_1,\dots,X_n]$ is a \emph{homogeneous} polynomial, 
  i.e., a sum of monomials of total degree exactly $d$. This is because otherwise one can multiply any monomial of total degree $d'<d$ by $\left(\sum_{i=1}^n X_i\right)^{d-d'}$: since $\sum_{i=1}^n \p_i = 1$, this does not affect the value of $Q(\p)$.
\end{remark}
We henceforth assume $Q$ is a homogeneous polynomial of degree $d$. 
Before stating our results, we will need to set some notation. 
Given a multi-set $S$ of independent samples from a distribution $p\in\distribs{[n]}$, 
we let $\Phi_S$ denote the \emph{fingerprint} of $S$, i.e., the vector $(\Phi_{S,1},\dots,\Phi_{S,n})\in\N^n$ of counts: 
$\sum_{i=1}^n \Phi_{S,i} = \abs{S}$, and $\Phi_{S,i}$ is the number of occurrences of $i$ in $S$. 
Moreover, for a vector $\vect{\alpha}=(\alpha_1,\dots,\alpha_n)\in\N^n$, 
we write $X^{\vect{\alpha}}$ for the monomial $X^{\vect{\alpha}}\eqdef \prod_{i=1}^n X_i^{\alpha_i}$, 
$\norm{\vect{\alpha}}$ for the $\lp[1]$-norm $\sum_{i=1}^n \alpha_i$, 
and $\binom{\norm{\vect{\alpha}}}{\vect{\alpha}}$ for the multinomial coefficient 
$\frac{\norm{\vect{\alpha}}!}{\alpha_1!\cdots\alpha_n!}$. Finally, for any integer $d\geq 0$, 
we denote by $\mathcal{H}_d\subseteq \R_d[X_1,\dots,X_n]$ 
the set of homogeneous degree-$d$ $n$-variate polynomials.

The estimators we consider are symmetric, that is only a function of the fingerprint $\Phi_S$. 
We first focus on the special case $N=d$.

\begin{lemma} \label{prop:unbiased:existence:Neqd}
  There exists an unbiased symmetric \emph{linear} estimator for $Q(\p)$, 
  i.e., a linear function $U^d_d\colon \R_d[X_1,\dots,X_n] \to \R_d[X_1,\dots,X_n]$ such that
  \[
      \expect{U^d_d Q(\Phi_S)} = Q(\p) \;,
  \]
  where $S$ is obtained by drawing $d$ independent samples from $\p$.
\end{lemma}
\begin{proof}
For any ordered $d$-tuple $T \in [n]^d$, by independence of the samples in $S$, 
we see that $\probaOf{S=T}=\prod_{i=1}^n p_{T_i} = \p^{\Phi_T}$. 
For any $\vect{\alpha} \in \N^n$ with $\norm{\vect{\alpha}}= d$, 
the number of $T \in [n]^d$ with fingerprint $\vect{\alpha}$ is $\binom{d}{\vect{\alpha}}$. 
Thus, we have that $\probaOf{\Phi_S=\vect{\alpha}}=\binom{d}{\vect{\alpha}} \p^{\vect{\alpha}}$. 
Noting that since $\norm{\vect{\alpha}}=\norm{\Phi_S}$, 
$\prod_{i=1}^n \binom{\Phi_{S,i}}{\alpha_i} = \delta_{\Phi_S, \alpha}$, 
we can define
\begin{equation} \label{eq:def:unbiased:Neqd}
	U^d_d X^{\vect{\alpha}}(\Phi_S) \eqdef \binom{d}{\vect{\alpha}}^{-1} \indic{\Phi_S=\vect{\alpha}} = \binom{d}{\vect{\alpha}}^{-1} \prod_{i=1}^n \binom{\Phi_{S,i}}{\alpha_i} \; .
\end{equation}
Then we have $\expect{U^d_d X^{\vect{\alpha}}(\Phi_S)} = \p^{\vect{\alpha}}$. 
We extend this linearly to all $Q \in \mathcal{H}_d$. 
By linearity of expectation, we obtain an unbiased estimator for any such $Q(\p)$.
\end{proof}

We can generalize this to $N \geq d$, by taking the average over all subsets of size $d$ of $S$ of the above estimator.
\begin{proposition}[Existence]\label{prop:unbiased:existence}
For  $N \geq d$ and $Q \in \mathcal{H}_d$ written in terms of monomials as 
$Q(X)= \sum_{\vect{\alpha}} c_{\vect{\alpha}} X^{\vect{\alpha}}$, the symmetric linear estimator
  \begin{equation}\label{eq:def:unbiased}
    U_N Q(\Phi_S) \eqdef \sum_{\vect{\alpha}} c_{\vect{\alpha}} \binom{N}{\vect{\alpha},N-\norm{\vect{\alpha}}}^{-1} \prod_{i=1}^n \binom{X_i}{\alpha_i}
  \end{equation}
  is an unbiased estimator for $Q(\p)$.
\end{proposition}
\begin{proof}
For the case $N=d$, this follows from~\cref{prop:unbiased:existence:Neqd,eq:def:unbiased:Neqd}. 
For any set of $d$ indices $I \subseteq [N], |I|=d$, the subset $S_I=\setOfSuchThat{S_i}{i \in I}$ 
is a set of $d$ independent samples from $\p$, thus $U^d_d Q(\Phi_{S_I})$ 
is an unbiased estimator for $Q(\p)$. To get a symmetric unbiased estimator (and to reduce the variance), 
we can take the average over all subsets $S_I$ of $S$ of size $d$. 
We claim that this estimator is $U_N$ as defined above, i.e., that
  \begin{equation} \label{eq:estimator-is-symmetrization}
    U_N  Q(\Phi_S) = \binom{N}{d}^{-1} \sum_{S' \subseteq S,|S'|=d} U^d_d Q(\Phi_{S'}) \; .
  \end{equation}
By linearity of expectation, the RHS is an unbiased estimator for $Q(\p)$, 
and so~\eqref{eq:estimator-is-symmetrization} suffices to show the proposition. 
By linearity of $U_N$ and $U^d_d$, we need to show~\eqref{eq:estimator-is-symmetrization} 
for each monomial $X^{\vect{\alpha}}$. Noting that the number of subsets $S'$ of $S$ of size $d$ 
that have fingerprint $\alpha$ is $\prod_{i=1}^n \binom{\Phi_{S,i}}{\alpha_i}$, we have
  \begin{align*}
    \binom{N}{d}^{-1} \sum_{S' \subseteq S,|S'|=d} U^d_d X^\vect{\alpha}(\Phi_{S'})  
    & = \binom{N}{d}^{-1} \sum_{S' \subseteq S,|S'|=d} \binom{d}{\vect{\alpha}}^{-1} \indic{\Phi_{S'}=\vect{\alpha}} \\
    &= \binom{N}{\vect{\alpha},N-\norm{\vect{\alpha}}}^{-1} \sum_{S' \subseteq S,|S'|=d} \indic{\Phi_{S'}=\vect{\alpha}} \\
    & = \binom{N}{\vect{\alpha},N-\norm{\vect{\alpha}}}^{-1} \prod_{i=1}^n \binom{\Phi_{S,i}}{\alpha_i} \\
    & = U_N X^{\vect{\alpha}}(\Phi_S) \;.
  \end{align*}
This completes the proof.
\end{proof}

\begin{proposition}[Uniqueness]\label{prop:poly:estim:uniqueness}
The unbiased estimator $U_N Q(\Phi_S)$ of \eqref{eq:def:unbiased} 
is unique among symmetric estimators. 
That is, for every $N\geq d$, for any symmetric estimator $V_N\colon [n]^N\to\R$ satisfying
$\expect{V_N(\Phi_S)} = Q(\p)$,
where $S$ is a multiset of $N$ samples drawn from $\p$, 
one must have $V_N(\Phi_S) = U_N Q(\Phi_S)$ for all $S$.
\end{proposition}
\begin{proof}
We first show that it is sufficient to establish uniqueness only for the case $d=N$, 
i.e., to show that $\operatorname{U}_d^d$ maps polynomials to singletons. 
To argue this is enough, suppose $N > d$, and with have two different $N$-sample 
estimators $V_N,W_N$ for a homogeneous degree-$d$ polynomial $Q$. 
Considering $R\eqdef \left(\sum_{i=1}^n X_i\right)^{N-d}Q$ 
which is homogeneous of degree $N$ and agrees with $Q$ 
on every probability distribution $\p$, we obtain two different 
$N$-sample estimators $V_N,W_N$ for a homogeneous degree-$N$ polynomial.
  
When $N=d$, we have a map $U^N_N$ from polynomials to estimators that gives an unbiased estimator for the polynomial.
By~\eqref{eq:def:unbiased:Neqd}, for $Q(X)=\sum_{\vect{\alpha}} c_{\vect{\alpha}} X^{\vect{\alpha}}$, this is given by
\[
U^d_d Q(\Phi_S) = \sum_{\vect{\alpha}} c_{\vect{\alpha}} \binom{d}{\vect{\alpha}}^{-1} \indic{\Phi_S=\vect{\alpha}} \;.
\]
Any symmetric estimator on $N$ samples can by written as a linear combination of $\indic{\Phi_S=\vect{\alpha}}$. 
Hence, given an estimator $V_N$, we can find a unique polynomial $Q_{V_n}$ 
with $U^d_d Q_{V_n}(\Phi_S) = V_n$ by choosing $c_{\vect{\alpha}}$ 
to match the coefficients in this linear combination, i.e., $U^d_d$ is a bijection between polynomials and symmetric estimators. 
Thus, if we have two different $N$-sample estimators $V_N,W_N$ 
for a homogeneous degree-$N$ polynomial $Q$, 
at least one of them is $U^d_d R$ for some homogeneous degree-$N$ polynomial $R$. 
  
Now we have an estimator $V_N$ that is unbiased for two different homogeneous 
degree-$N$ polynomials $Q$ and $R$.  So we get that for every $\p\in\distribs{[n]}$, 
$Q(\p)=\shortexpect_S[ V_N(\Phi_S) ] = R(\p)$. Hence, their difference $D\eqdef Q-R$ 
is a non-zero homogeneous degree-$N$ polynomial which vanishes on every point 
$(x_1,\dots,x_n)\in\N^n$ with $\sum_{i=1}^n x_i = 1$. 
By homogeneity, for every non-zero $\vect{x}=(x_1,\dots,x_n)\in\R_+^n$,
\[
D(\vect{x}) = \normone{\vect{x}}^d D\mleft(\frac{\vect{x}}{\normone{\vect{x}}}\mright) = \normone{\vect{x}}^N\cdot 0 = 0 \;,
\]
and therefore $D$ vanishes on the whole 
non-negative quadrant $\R_+^n = \setOfSuchThat{\vect{x}\in\R^n}{x_i \geq 0 \text{ for all } i}$. 
Being identically zero on an open set, $D$ must be the zero polynomial, leading to a contradiction.
\end{proof}

The above shows existence and uniqueness of an unbiased estimator, 
provided the number of samples $N$ is at least the degree $d$ 
of the polynomial (in $\p$) we are trying to estimate. 
The proposition below shows this is necessary: 
if $N<d$, there is no unbiased estimator in general.
\begin{proposition}
Let $Q\in\mathcal{H}_d$ be a homogeneous $n$-variate polynomial 
such that $\sum_{k=1}^n X$ does not divide $Q$. Then, there 
exists no unbiased estimator for $Q(\p)$ from $N$ samples unless $N\geq d$.
\end{proposition}
\begin{proof}
Suppose by contradiction that, for such a $Q\in\mathcal{H}_d$, 
there exists an unbiased estimator for $Q(\p)$ with $N<d$ samples. 
Then, since $\operatorname{U}_N^N$ (with the notation of the 
proof of~\cref{prop:poly:estim:uniqueness}) is invertible, 
this estimator is also an unbiased estimator for some homogeneous 
degree-$N$ polynomial $R\in\mathcal{H}_N$. Therefore, it is also 
an unbiased estimator for the degree-$d$ homogeneous polynomial 
$R'\eqdef R\cdot(\sum_{k=1}^n X_k)^{d-N}\in\mathcal{H}_d$. 
But by~\cref{prop:poly:estim:uniqueness}, one must then 
have $Q=R'$, which is impossible since $\sum_{k=1}^n X$ does not divide $Q$.
\end{proof}

\subsection{Bounding the Variance of the Unbiased Estimator} \label{ssec:var-bound:poly}

Having established existence and uniqueness of our unbiased estimator, 
it remains to bound its variance:
\begin{theorem}\label{theo:expect:squared:unbiased}
  Fix $N\geq d$, and let $U_N\colon \R_d[X_1,\dots,X_n] \to \R_d[X_1,\dots,X_n]$ be as above. Then, for every $Q\in\mathcal{H}_d$,
  \begin{equation}\label{eq:expect:squared:unbiased}
      \expect{(U_N Q(\Phi_S))^2} 
      = \sum_{\substack{ \vect{s}\in\N^n \\\norm{\vect{s}}\leq d}}  \p^{\vect{s}} \left(\frac{d^{\norm{\vect{s}}}Q(\p)}{dX^{\vect{s}}}\right)^2 \frac{(N-d)!^2}{N! (N-2d+\norm{\vect{s}})! \prod_{i=1}^n s_i!} \;,
  \end{equation}
  where the expectation is over $S$ obtained by drawing $N$ independent samples from $\p$.
\end{theorem}
\begin{proof}
  In order to establish the identity, we first consider monomials: for $\vect{\alpha},\vect{\beta}\in\N^n$, we will analyze $\expect{U_N X^{\vect{\alpha}}(\Phi_S))U_N X^{\vect{\beta}}(\Phi_S))}$, before extending it to $Q\in\mathcal{H}_d$, relying on the linearity of $U_N$. First, note that by definition of $U_N$ (in~\cref{eq:def:unbiased}),
  \begin{equation}\label{eq:UXalpha:UXbeta}
      U_N X^{\vect{\alpha}} U_N X^{\vect{\beta}} = \frac{1}{\binom{N}{\vect{\alpha}, N-\norm{\vect{\alpha}}}\binom{N}{\vect{\beta},N-\norm{\vect{\beta}}}} \prod_{i=1}^n \binom{X_i}{\alpha_i}\binom{X_i}{\beta_i}\,.
  \end{equation}
  We will use the following fact:
  \begin{claim}\label{claim:binomial:multinomial:identity}
      For $0\leq a,b\leq n$, we have 
      \[
        \binom{n}{a}\binom{n}{b} = \sum_{s=0}^{\min(a,b)} \binom{n}{a+b-s}\binom{a+b-s}{a-s,b-s,s}\,.
      \]
  \end{claim}
  \begin{proof}
    The left-hand-side $\binom{n}{a}\binom{n}{b}$ is the number of subsets $A$,$B$ of $[n]$ with $|A|=a$,$|B|=b$. We can group the set of such pairs of subsets by the size of their intersection $s=|A \cap B|$. Summing the size of these classes gives that
    \[
        \binom{n}{a}\binom{n}{b} = \sum_{s=0}^{\min(a,b)} \binom{n}{a-s,b-s,s,n-(a+b)+s} \;,
    \]
	which is easily seen to be equivalent to the claim by multiplying out the factorials.
  \end{proof}
  
  We can then rewrite, combining~\cref{eq:UXalpha:UXbeta,claim:binomial:multinomial:identity}, and setting $\pi_{\vect{\alpha},\vect{\beta}} \eqdef \frac{1}{\binom{N}{\vect{\alpha},N-\norm{\vect{\alpha}}}\binom{N}{\vect{\beta},N-\norm{\vect{\beta}}}}$ for convenience, that
  \begin{align*}
      U_N X^{\vect{\alpha}} U_N X^{\vect{\beta}} 
        &= \pi_{\vect{\alpha},\vect{\beta}} \prod_{i=1}^n \sum_{s=0}^{\min(\alpha_i,\beta_i)} \binom{X_i}{\alpha_i+\beta_i-s}\binom{\alpha_i+\beta_i-s}{\alpha_i-s,\beta_i-s,s} \\
        &= \pi_{\vect{\alpha},\vect{\beta}} \sum_{\substack{\vect{s}\in\N^n \\ \vect{s}\leq \min(\vect{\alpha},\vect{\beta})}}\prod_{i=1}^n \binom{X_i}{\alpha_i+\beta_i-s_i}\binom{\alpha_i+\beta_i-s_i}{\alpha_i-s_i,\beta_i-s_i,s_i} \,.
  \end{align*}
  Taking the expectation over an $N$-sample multiset $S$, we obtain 
  \begin{align*}
      \expect{ U_N X^{\vect{\alpha}}(\Phi_S) U_N X^{\vect{\beta}}(\Phi_S) } 
        &= \pi_{\vect{\alpha},\vect{\beta}} \sum_{\substack{\vect{s}\in\N^n \\ \vect{s}\leq \min(\vect{\alpha},\vect{\beta})}} 
        \prod_{i=1}^n \binom{\alpha_i+\beta_i-s_i}{\alpha_i-s_i,\beta_i-s_i,s_i} \expect{  \prod_{i=1}^n \binom{\Phi_{S,i}}{\alpha_i+\beta_i-s_i} } \;.
  \end{align*}
  Recalling the proof of~\cref{prop:unbiased:existence}, we have
\begin{align*}
\expect{  \prod_{i=1}^n \binom{\Phi_{S,i}}{\alpha_i+\beta_i-s_i} }
&= \binom{N}{\vect{\alpha}+\vect{\beta}-\vect{s}, N-(\norm{\vect{\alpha}}+\norm{\vect{\beta}}-\norm{\vect{s}})} \expect{U_N X^{\vect{\alpha}+\vect{\beta}-\vect{s}}(\Phi_S)} \\
&= \binom{N}{\vect{\alpha}+\vect{\beta}-\vect{s}, N-(\norm{\vect{\alpha}}+\norm{\vect{\beta}}-\norm{\vect{s}})} \p^{\vect{\alpha}+\vect{\beta}-\vect{s}} \;,
\end{align*}
  leading to
  \begin{align*}
      \expect{ U_N X^{\vect{\alpha}}(\Phi_S) \right.&\left. U_N X^{\vect{\beta}}(\Phi_S) } \\
        &= \pi_{\vect{\alpha},\vect{\beta}} \sum_{\substack{\vect{s}\in\N^n \\ \vect{s}\leq \min(\vect{\alpha},\vect{\beta})}} \p^{\vect{\alpha}+\vect{\beta}-\vect{s}} 
        \binom{N}{\vect{\alpha}+\vect{\beta}-\vect{s}, N-(\norm{\vect{\alpha}}+\norm{\vect{\beta}}-\norm{\vect{s}})} \prod_{i=1}^n \binom{\alpha_i+\beta_i-s_i}{\alpha_i-s_i,\beta_i-s_i,s_i}\\
        &= \pi_{\vect{\alpha},\vect{\beta}} \sum_{\substack{\vect{s}\in\N^n \\ \vect{s}\leq \min(\vect{\alpha},\vect{\beta})}} \p^{\vect{\alpha}+\vect{\beta}-\vect{s}} 
        \binom{N}{\vect{\alpha}-\vect{s},\vect{\beta}-\vect{s},\vect{s}, N-(\norm{\vect{\alpha}}+\norm{\vect{\beta}}-\norm{\vect{s}})} \;.
  \end{align*}
To get a better hold on this expression and extend the analysis 
to general homogeneous polynomials (instead of monomials), 
we first massage the expression above 
under the additional constraint that $\norm{\vect{\alpha}} = \norm{\vect{\beta}} = d$.
  \begin{align*}
      \expect{ U_N X^{\vect{\alpha}}(\Phi_S) \right.&\left. U_N X^{\vect{\beta}}(\Phi_S) } \\
        &= \frac{(N-d)!^2}{N!} \prod_{i=1}^n \alpha_i!\prod_{i=1}^n \beta_i!  \sum_{\substack{\vect{s}\in\N^n \\ \vect{s}\leq \min(\vect{\alpha},\vect{\beta})}} 
        \frac{\p^{\vect{\alpha}+\vect{\beta}-\vect{s}} }{(N-2d+\norm{\vect{s}})!\prod_{i=1}^n (\alpha_i-s_i)! (\beta_i-s_i)! s_i!} \\
        &= \frac{(N-d)!^2}{N!}
		\sum_{\substack{\vect{s}\in\N^n \\ \vect{s}\leq \min(\vect{\alpha},\vect{\beta})}} 
		\frac{\p^{\vect{\alpha}+\vect{\beta}-\vect{s}}}{(N-2d+\norm{\vect{s}})!\prod_{i=1}^n  s_i!}
		\prod_{i=1}^n \frac{\alpha_i!}{(\alpha_i-s_i)!} \prod_{i=1}^n \frac{\beta_i!}{(\beta_i-s_i)!} \;.
  \end{align*}
Recalling that $\frac{d^{\norm{\vect{s}}} X^{\vect{\alpha}}}{dX^{\vect{s}}} = \prod_{i=1}^n \frac{\alpha_i!}{(\alpha_i-s_i)!} X_i^{\alpha_i-s_i}$ for $\vect{s}\leq \vect{\alpha}$, we have
  \[
      \p^{\vect{s}}\frac{d^{\norm{\vect{s}}} \p^{\vect{\alpha}}}{dX^{\vect{s}}}\frac{d^{\norm{\vect{s}}} \p^{\vect{\beta}}}{dX^{\vect{s}}} = \p^{\vect{\alpha}+\vect{\beta}-\vect{s}} \prod_{i=1}^n \frac{\alpha_i!}{(\alpha_i-s_i)!} \prod_{i=1}^n \frac{\beta_i!}{(\beta_i-s_i)!}
  \]
  for $\vect{s}\leq \min(\vect{\alpha},\vect{\beta})$, from which
  \begin{align*}
      \expect{ U_N X^{\vect{\alpha}}(\Phi_S) U_N X^{\vect{\beta}}(\Phi_S) } 
        &= \frac{(N-d)!^2}{N!} \sum_{\substack{\vect{s}\in\N^n \\ \vect{s}\leq \min(\vect{\alpha},\vect{\beta})}} 
        \frac{\p^{\vect{s}}}{(N-2d+\norm{\vect{s}})!\prod_{i=1}^n  s_i!} \frac{d^{\norm{\vect{s}}} \p^{\vect{\alpha}}}{dX^{\vect{s}}}\frac{d^{\norm{\vect{s}}} \p^{\vect{\beta}}}{dX^{\vect{s}}} \,.
  \end{align*}
  By linearity of $U$ and differentiation, this implies that, for any $Q,R\in\mathcal{H}_d$,
  \[
    \expect{ U_N Q(\Phi_S) U_N R(\Phi_S) } = \frac{(N-d)!^2}{N!} \sum_{\substack{\vect{s}\in\N^n \\ \norm{\vect{s}}\leq d}} 
        \frac{\p^{\vect{s}}}{(N-2d+\norm{\vect{s}})!\prod_{i=1}^n  s_i!} \frac{d^{\norm{\vect{s}}} Q(\p)}{dX^{\vect{s}}}\frac{d^{\norm{\vect{s}}} R(\p)}{dX^{\vect{s}}}\,.
  \]
  Choosing $R=Q$ yields~\cref{eq:expect:squared:unbiased}.
\end{proof}

By the previous theorem, in order to analyze the variance 
$\var \left[ U_N Q(\Phi_S) \right] = \expect{(U_N Q(\Phi_S))^2} - \expect{U_N Q(\Phi_S)}^2$, 
one needs to bound the different terms of
\[
  \expect{(U_N Q(\Phi_S))^2} =
  \sum_{h=0}^d \sum_{\substack{ \vect{s}\in\N^n \\\norm{\vect{s}} = h}}  \p^{\vect{s}} \left(\frac{d^{h}Q(\p)}{dX^{\vect{s}}}\right)^2 \frac{(N-d)!^2}{N!(N-2d+h)! \prod_{i=1}^n s_i!} = \sum_{h=0}^d T_h(Q,\p,d,N) \;,
\]
letting $T_h(Q,\p,d,N)$ denote the inner sum 
for a given $0\leq h\leq d$. 
Next, we provide some useful bounds on some of these terms. 
We show that the first term will be mostly taken care of in the variance 
by the subtracted squared expectation, $\expect{U_N Q(\Phi_S)}^2 = Q(\p)^2$. 
This allows us to get a bound on the variance directly:
\begin{corollary} \label{cor:variance:unbiased}
For $h \geq 0$,
\[
  T_h(Q,\p,d,N) \leq \frac{(N-h)!}{N!} \sum_{\substack{ \vect{s}\in\N^n \\\norm{\vect{s}} = h}}  \p^{\vect{s}} \left(\frac{d^{h}Q(\p)}{dX^{\vect{s}}}\right)^2 \frac{1}{\prod_{i=1}^n s_i!}  \;,
\]
and so
\[
  \var{U_N Q(\Phi_S)} \leq
  \sum_{h=1}^d \frac{(N-h)!}{N!} \sum_{\substack{ \vect{s}\in\N^n \\\norm{\vect{s}} = h}}  \p^{\vect{s}} \left(\frac{d^{h}Q(\p)}{dX^{\vect{s}}}\right)^2 \frac{1}{\prod_{i=1}^n s_i!} \;.
\]
\end{corollary}
\begin{proof}
We have
\[
  \frac{(N-d)!^2}{N!(N-2d+h)!} = \prod_{i=1}^{d-h} \frac{N-2d+h+i}{N-d+i} \prod_{j=0}^{h-1} \frac{1}{N-j} \leq  \frac{(N-h)!}{N!}
\]
which gives the bound on $T_h(Q,\p,d,N)$. For $h=0$, this gives that $T_0(Q,\p,d,N) \leq Q(\p)^2=\expect{U_N Q(\Phi_S)}^2$ and so if we expand $\var{U_N Q(\Phi_S)}=\expect{U_N Q(\Phi_S)^2}- \expect{U_N Q(\Phi_S)}^2$, the $T_0(Q,\p,d,N)$ term is at least cancelled by the $- \expect{U_N Q(\Phi_S)}^2$.
\end{proof}

In view of bounding the rest of the terms, 
let $Q^+\in\mathcal{H}_d$ denote the polynomial obtained 
from $Q$ by making all its coefficients non-negative: 
that is, if $Q=\sum_{\substack{\norm{\vect{\alpha}}= d}} c_{\vect{\alpha}} X^{\vect{\alpha}}$, 
then $Q^+\eqdef\sum_{\substack{\norm{\vect{\alpha}}= d}} \abs{c_{\vect{\alpha}}} X^{\vect{\alpha}}$. 
Then, we show the following:

\begin{lemma}\label{lemma:contribution:higher}
Fix any $0\leq g\leq d$. Then,
\[
\sum_{h=g}^d T_h(Q,\p,d,N) = \bigO{\frac{1}{N^g}} 2^d Q^+(\p) \max_{\vect{s}: \norm{\vect{s}}\geq g} \abs{ \frac{d^{h}Q(\p)}{dX^{\vect{s}}} }\,.
\]
\end{lemma}
\begin{proof}
  For $g$ as above, we have
  \begin{align*}
     \sum_{h=g}^d T_h(Q,\p,d,N) 
     &= \sum_{h=g}^d \bigO{\frac{1}{N^h}} \sum_{\substack{ \vect{s}\in\N^n \\\norm{\vect{s}} = h}} \frac{1}{\prod_{i=1}^n s_i!} \p^{\vect{s}} \left(\frac{d^{h}Q(\p)}{dX^{\vect{s}}}\right)^2\\
     &= \bigO{\frac{1}{N^g}} \sum_{h=g}^d \sum_{\substack{ \vect{s}\in\N^n \\\norm{\vect{s}} = h}} \frac{1}{\prod_{i=1}^n s_i!} \p^{\vect{s}} \left(\frac{d^{h}Q(\p)}{dX^{\vect{s}}}\right)^2\,.
  \end{align*}
A useful observation is that since $Q$ is homogeneous of degree $d$, then so is $X^{\vect{s}}\frac{d^{h}Q}{dX^{\vect{s}}}$ for every $\vect{s}$. Consider a term $c_{\vect{\alpha}} X^{\vect{\alpha}}$ in $Q$; a term $X^{\vect{\alpha}}$ will appear in $X^{\vect{s}}\abs{\frac{d^{h}Q}{dX^{\vect{s}}}}$ if and only if $\vect{\alpha}\geq \vect{s}$, in which case this term will be 
\[
    \left(\prod_{i=1}^n\frac{\alpha_i!}{(\alpha_i-s_i)!}\right) \abs{c_{\vect{\alpha}}}X^{\vect{\alpha}}
    = \abs{c_{\vect{\alpha}}}X^{\vect{\alpha}}\prod_{i=1}^n s_i! \prod_{i=1}^n\binom{\alpha_i}{s_i}\,.
\]
Therefore, summing over all $\vect{s}$, we get
\[
\sum_{h=g}^d \sum_{\substack{ \vect{s}\in\N^n \\\norm{\vect{s}} = h}} \frac{1}{\prod_{i=1}^n s_i!} X^{\vect{s}} \abs{ \frac{d^{h}(c_{\vect{\alpha}} X^{\vect{\alpha}})}{dX^{\vect{s}}} }
=
\sum_{h=g}^d \sum_{\substack{ \vect{s}\leq \vect{\alpha}  \\\norm{\vect{s}} = h}} \abs{c_{\vect{\alpha}}}X^{\vect{\alpha}}\prod_{i=1}^n\binom{\alpha_i}{s_i}
\leq 
\sum_{\vect{s}\leq \vect{\alpha} } \abs{c_{\vect{\alpha}}}X^{\vect{\alpha}}\prod_{i=1}^n\binom{\alpha_i}{s_i}
= 2^d\abs{c_{\vect{\alpha}}}X^{\vect{\alpha}} \;,
\]
where the inequality is an abuse of notation, 
assuming $X$ is a non-negative vector. 
For the last equality, we relied on the facts that $\norm{\vect{\alpha}}=d$ and
\[
\sum_{\vect{s}\leq \vect{\alpha} } \prod_{i=1}^n\binom{\alpha_i}{s_i}
= \prod_{i=1}^n \sum_{ s : s \leq \alpha_i }\binom{\alpha_i}{s} = \prod_{i=1}^n 2^{\alpha_i} = 2^{\norm{\vect{\alpha}}}\,.
\]
By linearity and the definition of $Q^+$, this yields
  \begin{align*}
     \sum_{h=g}^d \sum_{\substack{ \vect{s}\in\N^n \\\norm{\vect{s}} = h}} \frac{1}{\prod_{i=1}^n s_i!} \p^{\vect{s}} \abs{ \frac{d^{h}Q(\p)}{dX^{\vect{s}}} }
     \leq 2^d Q^+(\p) \;,
  \end{align*}
and thus
  \begin{align*}
     \sum_{h=g}^d T_h(Q,\p,d,N) 
     &= \bigO{\frac{1}{N^g}} \sum_{h=g}^d \sum_{\substack{ \vect{s}\in\N^n \\\norm{\vect{s}} = h}} \frac{1}{\prod_{i=1}^n s_i!} \p^{\vect{s}} \left(\frac{d^{h}Q(\p)}{dX^{\vect{s}}}\right)^2\\
     &\leq \bigO{\frac{1}{N^g}} 2^d Q^+(\p) \max_{\vect{s}: \norm{\vect{s}}\geq g} \abs{ \frac{d^{h}Q(\p)}{dX^{\vect{s}}} } \;.
  \end{align*}
This completes the proof.
\end{proof}

\begin{remark}[Poissonized case]
Frequently, in distribution testing we analyze Poissonized statistics. 
That is, instead of $S$ being a set of $N$ samples, we consider a set $S$ of $\poisson{N}$ samples. 
In this case, $\Phi_{S,i}$ is independent for different $i$'s 
and $\expect{\prod_{i=1}^n \binom{S_i}{\alpha_i}} = \p^{\vect{\alpha}} \prod_{i=1}^n \frac{N}{\alpha_i!}$. 
Thus, we can define an unbiased estimator for $U'_N Q$ for a polynomial $Q(\p)$ 
by taking linear combinations of $U'_N X^{\vect{\alpha}}(\Phi_S) =N^{-\norm{\vect{\alpha}}} \prod_{i=1}^n \binom{S_i}{\alpha_i} \alpha_i!$.
The theory in the Poissonized setting is a little different: 
this estimator is not unique and is unbiased for any $N > 0$, including non-integral $N$ and $N < d$. 
However, the expression for $\expect{U'_N X^{\vect{\alpha}}(\Phi_S)^2}$ is very similar, 
and is obtained by an analogous proof. 
The difference is that we obtain a term $N^{-h}$ instead of $\frac{(N-d)!^2}{N!(N-2d+h)!}$. 
The bound on the variance in~\cref{cor:variance:unbiased} 
holds for the unbiased estimators in both the Poissonized and non-Poissonized cases.
\end{remark}

\subsection{Case of Interest: \texorpdfstring{$\lp[2]$}{l2}-Distance between \texorpdfstring{$\p$}{p} and \texorpdfstring{$\p_{\domx}\otimes \p_{\domy}$}{pX x pY}}\label{ssec:poly:l2}

We now instantiate the results of the previous subsections to a case of interest, 
the polynomial $Q$ corresponding to the $\lp[2]$ distance 
between a bivariate discrete distribution and the product of its marginals. 
In more detail, for any distribution $\p\in\distribs{\domx\times\domy}$, 
where $\abs{\domx}=\ell_1$, $\abs{\domy}=\ell_2$, 
we let $\p_\Pi\eqdef \p_{\domx}\otimes \p_{\domy}\in\distribs{\domx\times\domy}$ be the product of its marginals. Moreover, let $Q\in\R_4[X_{1,1}, X_{2,1},\dots,X_{\ell_1,1},X_{\ell_1,2},\dots,X_{\ell_1,\ell_2}]$ be the degree-$4$ $(\ell_1\ell_2)$-variate polynomial defined as
\begin{equation}\label{eq:polynomial:l2:product}
    Q(X_{1,1},\dots,X_{\ell_1,\ell_2}) \eqdef \sum_{i=1}^{\ell_1}\sum_{j=1}^{\ell_2} \left( X_{i,j} \sum_{i'\neq i}\sum_{j'\neq j} X_{i',j'} - \sum_{i'\neq i} X_{i',j}\sum_{j'\neq j} X_{i,j'} \right)^2\,.
\end{equation}
An explicit expression for its unbiased estimator $U_N Q(\Phi_S)$ will be given in~\cref{eq:explicit:formula:estimator:l2}. 
Specifically, we shall prove the following result:
\begin{proposition}\label{prop:variance:l2}
  Let $Q$ be as in~\cref{eq:polynomial:l2:product}, and suppose that $b \geq \max( \normtwo{\p}^2, \normtwo{\p_\Pi}^2 )$. Then, for $N\geq 4$,
  \[
      \var \left[ U_N Q(\Phi_S) \right] = \bigO{\frac{Q(\p)\sqrt{b}}{N} + \frac{b}{N^2} } \;.
  \]
\end{proposition}

\noindent For consistency of notation with the previous section, we let $n\eqdef \ell_1\ell_2$ in what follows.
\begin{claim}
  For any $\p$ over $\domx\times\domy$, we have  $Q(\p) = \normtwo{\p-\p_\Pi}^2$.
\end{claim}
\begin{proof}
  Unraveling the definitions, we can write
  \begin{align*}
    \normtwo{\p-\p_\Pi}^2 
    &= \sum_{i=1}^{\ell_1}\sum_{j=1}^{\ell_2} \left( \p(i,j) - \p_\Pi(i,j) \right)^2
     = \sum_{i=1}^{\ell_1}\sum_{j=1}^{\ell_2} \left( \p(i,j) - \sum_{j'=1}^{\ell_2} \p(i,j')\sum_{i'=1}^{\ell_1} \p(i',j) \right)^2\\
    &= \sum_{i=1}^{\ell_1}\sum_{j=1}^{\ell_2} \left( \p(i,j) - \left(\p(i,j)+\sum_{j'\neq j} \p(i,j')\right)\left(\p(i,j)+\sum_{i'\neq i} \p(i',j)\right) \right)^2\\
    &= \sum_{i=1}^{\ell_1}\sum_{j=1}^{\ell_2} \left( \p(i,j)\left(1-\p(i,j)-\sum_{i'\neq i} \p(i',j)-\sum_{j'\neq j} \p(i,j')\right) -\sum_{i'\neq i} \p(i',j)\sum_{j'\neq j} \p(i,j')\right)^2\\
&= \sum_{i=1}^{\ell_1}\sum_{j=1}^{\ell_2} \left( \p(i,j)\sum_{i'\neq i}\sum_{j'\neq j} \p(i',j')-\sum_{i'\neq i} \p(i',j)\sum_{j'\neq j} \p(i,j')\right)^2 = Q(\p) \;,
  \end{align*}
  as claimed.
\end{proof}

Firstly, we compute $U_N Q$ explicitly. 
By linearity of $U_N$, we can compute the unbiased estimator for each term separately, after writing
$Q(X) = \sum_{i=1}^{\ell_1}\sum_{j=1}^{\ell_2} \Delta_{ij}(X)^2$,
where $\Delta_{ij}(X) \eqdef X_{i,j} \sum_{i'\neq i}\sum_{j'\neq j} X_{i',j'} - \sum_{i'\neq i} X_{i',j}\sum_{j'\neq j} X_{i,j'}$. 
Now  $U_N Q = \new{\sum_{i=1}^{\ell_1}\sum_{j=1}^{\ell_2}}U_N \Delta_{ij}^2$ 
and we want to compute $U_N \Delta_{ij}^2$. Note that the sums in $\Delta_{ij}(X)$ are over disjoint sets of $X_{i,j}$'s 
whose union is every $X_{i,j}$. We can consider $\Delta_{ij}$ as a polynomial over the probabilities of a distribution with support of size $4$, which consists of the events given by whether the marginal $X$ is equal to $i$, and whether the marginal $Y$ is equal to $j$. By uniqueness of the unbiased estimator, $U_N \Delta_{ij}^2$ is the same on this distribution of support $4$ as on the original $\ell_1 \ell_2$-size support distribution. Formally, we will write 
\[
    \Delta_{ij}(X) \eqdef X_{i,j} X_{-i,-j} - X_{i,-j} X_{-i,j} \;,
\]
where  $X_{-i,-j}\eqdef \sum_{i'\neq i}\sum_{j'\neq j} X_{i',j'}$, 
$X_{-i,j}\eqdef\sum_{i'\neq i} X_{i',j}$, and $X_{i,-j} \eqdef \sum_{j'\neq j} X_{i,j'}$. 
Squaring gives
\[
    \Delta_{ij}(X)^2 = X_{i,j}^2 X_{-i,-j}^2 + X_{i,-j}^2 X_{-i,j}^2 - X_{i,j} X_{-i,-j}X_{i,-j} X_{-i,j} \;,
\]
and it remains to apply $U_N$ to each of these terms. We see that 
\[
    \frac{N!}{(N-4)!} U_N X_{i,j} X_{-i,-j}X_{i,-j} X_{-i,j} = \Phi_{S,i,j} \Phi_{S,-i,-j} \Phi_{S,i,-j} \Phi_{S,-i,j}\,,
\]
\[
    \frac{N!}{(N-4)!} U_N X_{i,j}^2 X_{-i,-j}^2 = \Phi_{S,i,j}(\Phi_{S,i,j}-1) \Phi_{S,-i,-j}(\Phi_{S,-i,-j} -1)\,,
\] and 
\[
    \frac{N!}{(N-4)!} U_N X_{-i,j}^2 X_{i,-j}^2 = \Phi_{S,-i,j}(\Phi_{S,-i,j}-1) \Phi_{S,i,-j}(\Phi_{S,i,-j} -1)\,.
\]
These counts are similarly summed so that, for example, $\Phi_{S,i,-j} = \sum_{j'\neq j} \Phi_{S,i,j'}$. 
Adding these together, we get that:
  \begin{align}
 \frac{N!}{(N-4)!} U_N Q(\Phi_S) 
 &=\frac{N!}{(N-4)!} \sum_{i=1}^{\ell_1}\sum_{j=1}^{\ell_2} U_N \Delta_{ij}(\Phi_S)^2 \notag\\
 &= \sum_{i=1}^{\ell_1}\sum_{j=1}^{\ell_2} \left( \Phi_{S,i,j}(\Phi_{S,i,j}-1) \Phi_{S,-i,-j}(\Phi_{S,-i,-j}-1)  \right.\notag\\
      &\qquad+ \left. \Phi_{S,-i,j}(\Phi_{S,-i,j}-1) \Phi_{S,i,-j}(\Phi_{S,i,-j} -1) - 2 \Phi_{S,i,j} \Phi_{S,-i,-j} \Phi_{S,i,-j} \Phi_{S,-i,j} \right) \notag\\
& = \sum_{i=1}^{\ell_1}\sum_{j=1}^{\ell_2} \left( (\Phi_{S,i,j} \Phi_{S,-i,-j} - \Phi_{S,-i,j}  \Phi_{S,-i,-j})^2 \right. \notag\\
      &\qquad+ \left. \Phi_{S,i,j} \Phi_{S,-i,-j} (1-\Phi_{S,i,j} - \Phi_{S,-i,-j}) + \Phi_{S,-i,j} \Phi_{S, i,-j} (1-\Phi_{S,-i,j} - \Phi_{S,i,-j}) \right) \;,\label{eq:explicit:formula:estimator:l2}
\end{align}
where $\Phi_{S,-i,-j} \eqdef \sum_{i'\neq i}\sum_{j'\neq j} \Phi_{S,i',j'}$, 
$\Phi_{S,-i,j}\eqdef\sum_{i'\neq i} \Phi_{S,i',j}$, 
and $\Phi_{S,i,-j} \eqdef \sum_{j'\neq j} \Phi_{S,i,j'}$. 
This yields the explicit formula for our unbiased estimator of $Q(\p)$. 

\medskip

\noindent We then turn to bounding its variance. 
From~\cref{theo:expect:squared:unbiased}, we then have that, for $N\geq 4$,
\begin{equation}
      \expect{(U_N Q(\Phi_S))^2} =
      \sum_{h=0}^4 \sum_{\substack{ \vect{s}\in\N^n \\\norm{\vect{s}}=h}} \binom{h}{\vect{s}} \p^{\vect{s}} \left(\frac{d^{h}Q(\p)}{dX^{\vect{s}}}\right)^2 \binom{ N-4 }{ 4-h } \binom{N}{h,4-h,N-4}^{-1} \frac{1}{h!^2} \;.
\end{equation}
The rest of this section is devoted to bounding this quantity. 
For $h\in\{0,\dots,4\}$, we let $T_h(N)$ be the inner sum corresponding to $h$, so that 
$\expect{(U_N Q(\Phi_S))^2} = \sum_{h=0}^4 T_h(N)$.

For clarity, we (re-)introduce some notation: that is, we write
$Q(X) = \sum_{i=1}^{\ell_1}\sum_{j=1}^{\ell_2} \Delta_{ij}(X)^2$,
where $\Delta_{ij}(X) \eqdef X_{i,j} \sum_{i'\neq i}\sum_{j'\neq j} X_{i',j'} - \sum_{i'\neq i} X_{i',j}\sum_{j'\neq j} X_{i,j'}$ as before. 
Each $\Delta_{ij}$ is a degree-$2$ polynomial, with partial derivatives
\[
    \frac{\partial \Delta_{ij}}{ \partial X_{k,\ell} } =
    \begin{cases}
        X_{i,j} &\text{ if } k\neq i, \ell\neq j\\
        \sum_{i'\neq i}\sum_{j'\neq j} X_{i',j'} &\text{ if } k=i, \ell = j\\
        -\sum_{i'\neq i} X_{i',j} &\text{ if } k=i, \ell \neq j\\
        -\sum_{j'\neq j} X_{i,j'} &\text{ if } k \neq i, \ell = j\\
    \end{cases}
\]
and
\[
    \frac{\partial^2 \Delta_{ij}}{ \partial X_{k,\ell}\partial X_{k',\ell'} } =
    (\delta_{ik}-\delta_{ik'})(\delta_{j\ell}-\delta_{j\ell'})\,.
\]

\begin{itemize}
  \item The first contribution, for $h=0$, is $\bigO{Q(\p)^2/N}$ 
  by~\cref{cor:variance:unbiased}, so we have $T_0$ under control. Indeed,
  \[
      Q(\p) \leq 2\sqrt{b}
  \]
  by the triangle inequality and the definition of $b$. So, $T_0(N)-Q(\p)^2 = \bigO{ Q(\p)\sqrt{b}/N}$.
  \item The second, $h=1$, contributes
\[
    T_1(N) =  \sum_{\substack{ \vect{s}\in\N^n \\\norm{\vect{s}}=1}} \p^{\vect{s}} \left(\frac{dQ(\p)}{dX^{\vect{s}}}\right)^2 \binom{ N-4 }{ 3 } \binom{N}{1,3,N-4}^{-1}
    =  4\frac{\binom{ N-4 }{ 3 }}{\binom{N}{4}}\sum_{\substack{ \vect{s}\in\N^n \\\norm{\vect{s}}=1}} \p^{\vect{s}} \left(\frac{dQ(\p)}{dX^{\vect{s}}}\right)^2
 \]
 Since ${\binom{ N-4 }{ 3 }}/{\binom{N}{4}} = \bigO{1/N}$, it is enough to consider the other factor,
 \[
    \sum_{\substack{ \vect{s}\in\N^n \\\norm{\vect{s}}=1}} \p^{\vect{s}} \left(\frac{dQ(\p)}{dX^{\vect{s}}}\right)^2
    = \sum_{k,\ell} \p_{k,\ell} \left(\frac{dQ(\p)}{dX_{k,\ell}}\right)^2\,.
 \]
 Recalling the expression of the derivatives of $\Delta_{ij}$, we have that
 \begin{align*}
    \frac{1}{2}\frac{dQ}{dX_{k,\ell}} &= \frac{1}{2}\sum_{i,j} 2\Delta_{ij}\frac{d\Delta_{ij}}{dX_{k,\ell}}
    \\
    &= \sum_{i\neq k}\sum_{j\neq \ell} X_{i,j}\Delta_{ij}(X)
    + \Delta_{k\ell}(X) \sum_{i\neq k}\sum_{j\neq \ell} X_{i,j}
    - \sum_{j\neq \ell} \Delta_{kj}(X) \sum_{i\neq k} X_{i,j}
    - \sum_{i\neq k} \Delta_{i\ell}(X) \sum_{j\neq \ell} X_{i,j} \;.
 \end{align*}
 Having this sum of four terms $A_1,A_2,A_3,A_4$ for $\frac{dQ}{dX_{k,\ell}}$, 
 by Cauchy--Schwarz it holds that
 $\left(\frac{dQ}{dX_{k,\ell}}\right)^2 \leq 4(A_1^2+A_2^2+A_3^2+A_4^2)$, 
 and so we can bound each of the square of these terms separately, ignoring cross factors.
 
 \begin{itemize}
 \item For the first, we have (again by Cauchy--Schwarz)
 \[
 \left(\sum_{i\neq k}\sum_{j\neq \ell} \p_{i,j}\Delta_{ij}(\p)\right)^2
 \leq  \left(\sum_{i,j} \p_{i,j}\Delta_{ij}(\p)\right)^2
 \leq \left(\sum_{i,j} \p_{i,j}^2\right)\left(\sum_{i,j} \Delta_{ij}(\p)^2\right)
 \leq b Q(\p) \leq \sqrt{b}Q(\p) \;,
 \]
 so $\sum_{k,\ell}  \p_{k,\ell} \left(\sum_{i\neq k}\sum_{j\neq \ell} \p_{i,j}\Delta_{ij}(\p)\right)^2 \leq b Q(\p)$.
 
 \item For the second, since $\left(\Delta_{k\ell}(\p) \sum_{i\neq k}\sum_{j\neq \ell} \p_{i,j}\right)^2
 \leq  \Delta_{k\ell}(\p)^2$, we have
 \[
 \sum_{k,\ell} \p_{k,\ell}\left(\Delta_{k\ell}(\p) \sum_{i\neq k}\sum_{j\neq \ell} \p_{i,j}\right)^2
 \leq  \sum_{k,\ell} \p_{k,\ell} \Delta_{k\ell}(\p)^2
 \leq \sqrt{\sum_{k,\ell} \p_{k,\ell}^2}\sqrt{\sum_{k,\ell} \Delta_{k\ell}(\p)^4}
  \leq \sqrt{b}\sqrt{\left(\sum_{k,\ell} \Delta_{k\ell}(\p)^2\right)^2} \;,
 \]
 which is equal to $\sqrt{b}Q(\p)$.
 
 \item For the third and fourth term (similarly handled by symmetry),
 \begin{align*}
 \sum_{k,\ell} \p_{k,\ell}\left(\sum_{j\neq \ell} \Delta_{kj}(\p) \sum_{i\neq k} \p_{i,j}\right)^2
 &\leq \sum_{k,\ell} \p_{k,\ell}\left(\sum_{j} \Delta_{kj}(\p) \sum_{i\neq k} \p_{i,j}\right)^2\\
 &= \sum_{k}\left(\sum_{j} \Delta_{kj}(\p) \sum_{i\neq k} \p_{i,j}\right)^2\sum_\ell \p_{k,\ell} \\
 &\leq \sum_{k}\left(\sum_{j} \Delta_{kj}(\p)^2 \sum_{j}\left(\sum_{i\neq k} \p_{i,j}\right)^2\right)\sum_\ell \p_{k,\ell} \tag{Cauchy--Schwarz}\\
 &\leq \sum_{k}\left(\sum_{j} \Delta_{kj}(\p)^2 \sum_{j}\left(\sum_{i} \p_{i,j}\right)^2\right)\sum_\ell \p_{k,\ell}\\
 &=\sum_{j}\left(\sum_{i} \p_{i,j}\right)^2\cdot \sum_{k}\left(\sum_{j} \Delta_{kj}(\p)^2\right)\sum_\ell \p_{k,\ell}\\
 &\leq \sum_{j}\Big(\sum_{i} \p_{i,j}\Big)^2\sqrt{ \sum_{k}\Big(\sum_{j} \Delta_{kj}(\p)^2\Big)^2\sum_k\Big(\sum_{\ell} \p_{k,\ell}\Big)^2 } \tag{Cauchy--Schwarz}\\
  &\leq \sqrt{ \sum_{j}\Big(\sum_{i} \p_{i,j}\Big)^2 \sum_k\Big(\sum_{\ell} \p_{k,\ell}\Big)^2 }\sqrt{ \sum_{k}\Big(\sum_{j} \Delta_{kj}(\p)^2\Big)^2} \;,
 \end{align*}
 where the last step relies on $\sum_{j}\big(\sum_{i} \p_{i,j}\big)^2\leq 1$ 
 (since it is the  squared $\lp[2]$-norm of a probability distribution, 
 that of the first marginal of $\p$) to write $\sum_{j}\big(\sum_{i} \p_{i,j}\big)^2\leq \sqrt{\sum_{j}\big(\sum_{i} \p_{i,j}\big)^2}$. 
 Continuing from there, and using monotonicity of $\lp[p]$ norms to write $\sum_i v_i^2 \leq \big(\sum_i \abs{v_i}\big)^2$,
  \begin{align*}
 \sum_{k,\ell} \p_{k,\ell}\left(\sum_{j\neq \ell} \Delta_{kj}(\p) \sum_{i\neq k} \p_{i,j}\right)^2
  &\leq \sqrt{ \sum_{j}\Big(\sum_{i} \p_{i,j}\Big)^2 \sum_k\Big(\sum_{\ell} \p_{k,\ell}\Big)^2 }\sum_{k}\sum_{j} \Delta_{kj}(\p)^2\\
  &=  \sqrt{ \sum_{j}\p_\domy(k)^2 \sum_k\p_\domx(j)^2 }Q(\p)
  =  \sqrt{ \sum_{k,j}\p_\Pi(k,j)^2 }Q(\p)\\
  &\leq \sqrt{b} Q(\p) \;,
 \end{align*}
 \end{itemize}
 and so $T_1(N) = \bigO{Q(\p)\sqrt{b}/N}$.
 
 Gathering these four terms, and by the above discussion, we obtain
 \[
    T_1(N) = 4\frac{\binom{ N-4 }{ 3 }}{\binom{N}{4}}\sum_{k,\ell} \p_{k,\ell} \left(\frac{dQ(\p)}{dX_{k,\ell}}\right)^2 \leq 4\frac{\binom{ N-4 }{ 3 }}{\binom{N}{4}}\cdot 8\cdot4\sqrt{b} Q(\p) = \bigO{\frac{\sqrt{b} Q(\p)}{N}}\,.
 \]
 
 \item Finally, for the rest of the contributions ($h\geq 2$), 
 we invoke~\cref{lemma:contribution:higher}. Specifically, we first observe that, for any distribution $\p\in\distribs{\domx\times\domy}$,
 \begin{align*}
 Q^+(\p) &= \sum_{i=1}^{\ell_1}\sum_{j=1}^{\ell_2} \left( \p_{i,j} \sum_{i'\neq i}\sum_{j'\neq j} \p_{i',j'} + \sum_{i'\neq i} \p_{i',j}\sum_{j'\neq j} \p_{i,j'} \right)^2
 \leq \sum_{i,j} \left( \p_{i,j} + \sum_{i'=1}^{\ell_1} \p_{i',j}\sum_{j'=1}^{\ell_2} \p_{i,j'} \right)^2 \\
 &\leq 2\sum_{i,j} \left( \p_{i,j}^2 + \left(\sum_{i'=1}^{\ell_1} \p_{i',j}\right)^2\left(\sum_{j'=1}^{\ell_2} \p_{i,j'} \right)^2  \right) \leq
 2\left(\normtwo{\p}^2+\normtwo{\p_\Pi}^2\right) \leq  4b\,.
 \end{align*}
 Next, we need to bound from above the high-order derivatives of $Q$. 
 By Leibniz's rule, for $h\geq 2$ and $\norm{\vect{s}}=h$, we can write:
 \begin{align*}
    \frac{d^{h}Q}{dX^{\vect{s}}} 
    &= \sum_{i,j} \frac{d^{h} \Delta_{ij}^2}{dX^{\vect{s}}}
    = \sum_{i,j} \sum_{\vect{s'}\leq \vect{s}} \prod_{\ell=1}^n \binom{s_\ell}{s'_\ell} \frac{d^{\norm{\vect{s'}}} \Delta_{ij}}{dX^{\vect{s'}}}\frac{d^{\norm{\vect{s}}-\norm{\vect{s'}}} \Delta_{ij}}{dX^{\vect{s}-\vect{s'}}} \\
    &\leq \sum_{\vect{s'}\leq \vect{s}} \prod_{i=\ell}^n \binom{s_\ell}{s'_\ell} \sqrt{ \sum_{i,j} \left(\frac{d^{\norm{\vect{s'}}} \Delta_{ij}}{dX^{\vect{s'}}}\right)^2 \sum_{i,j} \left(\frac{d^{\norm{\vect{s}-\vect{s'}}} \Delta_{ij}}{dX^{\vect{s}-\vect{s'}}}\right)^2 } \tag{Cauchy--Schwarz}\\
    &\leq \max_{ \vect{s'}\leq \vect{s} } \sum_{i,j} \left(\frac{d^{\norm{\vect{s'}}} \Delta_{ij}}{dX^{\vect{s'}}}\right)^2 \sum_{\vect{s'}\leq \vect{s}} \prod_{i=\ell}^n \binom{s_\ell}{s'_\ell} = 2^h \max_{ \vect{s'}\leq \vect{s} } \sum_{i,j} \left(\frac{d^{\norm{\vect{s'}}} \Delta_{ij}}{dX^{\vect{s'}}}\right)^2\,.
 \end{align*}
 Since $\Delta_{ij}$ has degree $2$, to bound this maximum 
 we have to consider three cases: first, $\sum_{i,j} \left(\frac{d^{0} \Delta_{ij}(\p)}{dX^{\vect{0}}}\right)^2 = Q(\p) \leq 4$. Second,
 recalling the partial derivatives of $\Delta_{ij}$ we computed earlier,
 \[
    \sum_{i,j} \left(\frac{d \Delta_{ij}(\p)}{dX_{k,\ell}}\right)^2
    = \sum_{i\neq k}\sum_{j\neq \ell} \p_{k,\ell}^2 + \left(\sum_{i'\neq k}\sum_{j'\neq \ell} \p_{i',j'}\right)^2
      + \sum_{i'\neq k}\p_{i',\ell}^2
      + \sum_{j'\neq \ell}\p_{k,j'}^2 \leq 4\,.
 \]
 Third,
 \[
    \sum_{i,j} \left(\frac{d^2 \Delta_{ij}(\p)}{dX_{k,\ell} dX_{k',\ell'}}\right)^2
    = \sum_{i,j} (\delta_{ik}-\delta_{ik'})^2(\delta_{j\ell}-\delta_{j\ell'})^2 \leq 4\,.
 \]

Combining all of the above cases results in 
$\abs{\frac{d^{h}Q}{dX^{\vect{s}}}}\leq 2^4\cdot 4$ for any $h\geq 2$ and $\norm{\vect{s}}=h$, 
and from there
\[
\sum_{h=2}^4 T_h(N) = \bigO{\frac{1}{N^2}} \cdot 2^4 \cdot 4\cdot 4b = \bigO{\frac{b}{N^2}}\,.
\] 
\end{itemize}
Accounting for all the terms, we can thus bound the variance as
\[
    \var U_N Q(\Phi_S) = (T_0(N)-Q(\p)^2) + T_1(N) + \sum_{h=2}^4 T_h(N)
    = \bigO{\frac{Q(\p)\sqrt{b}}{N} + \frac{b}{N^2} } \;,
\]
concluding the proof of~\cref{prop:variance:l2}.

\begin{remark}[Estimating a Polynomial under Poisson Sampling]\label{ssec:poly:poisson}
We observe that analogues of our theorems hold under \emph{Poisson} 
sampling (instead of multinomial sampling as treated in~\cref{sec:polynomial}). 
We defer these results, which follow from a straightforward (yet slightly cumbersome) 
adaptation of the proofs of this section, to an updated version of this paper.
\end{remark}

\section{The General Conditional Independence Tester} \label{sec:alg:flattened}

In this section, we present and analyze our general algorithm for testing conditional independence. 
The structure of this section is as follows:
In~\cref{sec:alg:flattened:flattening}, we begin by describing how we flatten the marginals of the distribution $\p_z$,
for each bin $z$ for which we receive enough samples.
After this flattening is performed, in~\cref{sec:alg:flattened:algo} we explain how 
we use the remaining samples for each such bin $z$ to compute a statistic $A$ as an appropriate weighted 
sum of bin-wise statistics $A_z$. \newer{Before going further, we discuss in~\cref{ssec:sample:complexity:discussion:flattened} 
the eventual result our analysis yields and comment on the sample complexity bound of our algorithm. 
In~\cref{sec:alg:flattened:randomness}, we explain the three different sources of randomness involved in our estimator, 
in order to clarify what will follow~--~as we will crucially later condition on part of this randomness 
to obtain bounds on some of its conditional expectations and variances. 
\cref{sec:alg:flattened:analysis:A} then details how the analysis of our statistic $A$ is performed 
(\cref{sec:alg:flattened:expectation,sec:alg:flattened:variance} respectively contain 
the analysis of the expectation and variance of $A$, conditioned on \emph{some} of the randomness at play).} 
Finally,~\cref{sec:alg:flattened:conclusion} puts everything together
and derives the correctness guarantee of our overall algorithm. 

\subsection{Flattening \texorpdfstring{$\domx$, $\domy$}{X,Y} for any Given Bin \texorpdfstring{$z$}{z}}\label{sec:alg:flattened:flattening}
Given a multiset $S$ of $N\geq 4$ independent samples from $\p\in\distribs{\domx\times\domy}$, 
where $\abs{\domx}=\ell_1$, $\abs{\domy}=\ell_2$, we perform the following. 
Losing at most three samples, we can assume $N=4+4t$ for some integer $t$. 
Let $t_1\eqdef \min(t,\ell_1)$  and $t_2\eqdef \min(t,\ell_2)$. 
We divide $S$ into two disjoint multi-sets $S_{\mathcal{F}},S_{\mathcal{T}}$ of size $t_1+t_2$ and $2t+4$ respectively, 
where the subscripts $\mathcal{F}$ and $\mathcal{T}$ stand for \emph{Flatten} and \emph{Test}.
\begin{itemize}
  \item We use $S_{\mathcal{F}}$ to flatten $\domx\times\domy$, as per~\cref{def:split:distribution}. 
  Namely, first we partition it into two multi-sets $S^1_{\mathcal{F}}$, $S^2_{\mathcal{F}}$ of size $t_1,t_2$. 
  Looking at the projections $\pi_{\domx}S^1_{\mathcal{F}}$, $\pi_{\domy}S^2_{\mathcal{F}}$ of $S^1_{\mathcal{F}}, S^2_{\mathcal{F}}$ 
  onto $\domx$ and $\domy$ respectively, we have two multi-sets of $t_1$ and $t_2$ elements. 
  We then let $T\subseteq\domx\times\domy$ obtained by, for each $(x,y)\in\domx\times\domy$, 
  adding in $T$ $a_{x,y}$ copies of $(x,y)$, where
  \begin{equation}\label{eq:flattening:axy}
      1+ a_{x,y} \eqdef (1+ a_x)(1+ a'_y)
  \end{equation}
  with $a_x$ (resp. $a'_y$) being the number of occurrences of $x$ in $\pi_{\domx}S^1_{\mathcal{F}}$ (resp. of $y$ in $\pi_{\domy}S^2_{\mathcal{F}}$). 
  Note that
  $
      \abs{T} + \ell_1\ell_2 = (\abs{\domx}+t_1)(\abs{\domy}+t_2)
  $, and that for all $(x,y)\in\domx\times\domy$, by a similar proof as that of~\cref{fact:split:distributions:l2norm:nonpoisson} 
  (using the fact that $a_x$ and $a'_y$ are independent),
  \[
    \expect{\frac{1}{1+ a_{x,y}}}=\expect{\frac{1}{1+ a_x}}\expect{\frac{1}{1+ a_y}} \leq \frac{1}{(1+t_1)(1+t_2)\p_{\domx}(x)\p_{\domy}(x)} \;,
  \]
  and so, letting $\q_T$ denote the product of the marginals of $\p_T$,
  \begin{equation}\label{eq:expected:l2norm:product:marginals:split}
    \expect{\normtwo{\q_T}^2} \leq \frac{1}{(1+t_1)(1+t_2)}\,.
  \end{equation}
  \item Next, we use the $2t+4\geq 4$ samples from $S_{\mathcal{T}}$ to estimate the squared $\lp[2]$-distance 
  between $\p_T$ and $\q_T$, as per~\cref{sec:polynomial}. Here,~\cref{remark:split:l2:chi2} will come in handy, 
  as it allows us to do it implicitly without having to actually map $\p$ to $\p_T$. 
  Indeed, recalling that the polynomial $Q$ for which we wish to estimate $Q(\p_T)$ is of the form 
  \[
      Q(X) = \sum_{(i,j)\in\domx\times\domy} \Delta_{ij}(X)^2 \;,
  \]
  we will instead estimate $R_T(\p)$, where $R$ is defined as
  \[
      R_T(X) \eqdef  \sum_{(i,j)\in\domx\times\domy} c_{i,j}\Delta_{ij}(X)^2
  \]
  with $c_{i,j} \eqdef \frac{1}{1+a_{i,j}}$ for all $(i,j)\in\domx\times\domy$. 
  From~\cref{remark:split:l2:chi2}, it is immediate that $R_T(\p) = Q(\p_T) = \normtwo{\p_T-\q_T}^2$, 
  and further by inspection of the proof of~\cref{prop:variance:l2} it is not hard to see 
  that the variance of our estimator $U_N R_T$ on $\p$ is the same as that of $U_NQ$ on $\p_T$.
  
  Let $B\eqdef \normtwo{\q_T}^2$. Note that $B$ is a random variable, determined by the choice of $S_{\mathcal{F}}$. 
  The first observation is that, while the statement of~\cref{prop:variance:l2} would be with regard to the 
  maximum of $\normtwo{\p_T}^2$, $\normtwo{\q_T}^2$, we would like to relate it to $B$. To do so, observe that
  \[
    \normtwo{\p_T}^2 \leq \left( \normtwo{\q_T} + \normtwo{\p_T - \q_T} \right)^2
    \leq 2\left( \normtwo{\q_T}^2 + \normtwo{\p_T - \q_T}^2 \right) 
    = 2\left( B+ Q(\p_T) \right)
  \]
  so we can use $B' \eqdef 2 B+ 2Q(\p_T)$ instead of our original bound $B$.  
  
  Therefore, our bound $B$ can be used in the statement of~\cref{prop:variance:l2}, leading to a variance for our estimator of
  \begin{equation}\label{eq:variance:estimator:flattened}
      \var \left[ U_N R_T \right] = \bigO{\frac{Q(\p_T)\sqrt{B'}}{N} + \frac{B'}{N^2} } = \bigO{\frac{Q(\p_T)\sqrt{B}}{N} + \frac{Q(\p_T)^{3/2}}{N} + \frac{B}{N^2} } \,.
  \end{equation}
  Now, recall that by~\cref{fact:split:distributions:l2norm:nonpoisson} (more precisely,~\cref{eq:expected:l2norm:product:marginals:split}), 
  we only have a handle on the \emph{expectation} of $B$. We could try to first obtain instead a high-probability bound on its value 
  by proving sufficiently strong concentration followed by a union bound over all estimators that we may run (i.e., all $n$ bins in $\domz$). 
  \new{However, this would lead to a rather unwieldy argument. Instead, as outlined in~\cref{sec:alg:flattened:analysis:A}, 
  we will analyze our estimators by carefully conditioning on some of the randomness 
  (the one underlying the flattening we perform for each bin), 
  and only convert the bounds obtained into high-probability statements at the end, 
  by a combination of Markov's and Chebyshev's inequalities.}
  \end{itemize}

\subsection{From Flattening to an Algorithm}\label{sec:alg:flattened:algo}
We now explain how the guarantees established above are sufficient to use in our algorithm. We will use the same notations as above, but now specifying the bin $z\in\domz$: that is, we will write $\p_z,\q_z,T_z,\p_{z,T_z},\q_{z,T_z}$ instead of $\p,\q,T,\p_{T},\q_{T}$ to make the dependence on the bin we condition on explicit. In what follows, we write $\sigma=(\sigma_z)_{z\in\domz}$, $T=(T_z\mid\sigma_z)_{z\in\domz}$.

\noindent We let  
\[
  A_z \eqdef \sigma_z\cdot \omega_z \cdot \Phi(S_z)\cdot \indic{\sigma_z\geq 4} \;,
\]
for all $z\in\domz$, where $\omega_z \eqdef \sqrt{\min\left(\sigma_z,\ell_1\right)\min\left(\sigma_z,\ell_2\right)}$. Our final statistic is 
\[
  A\eqdef \sum_{z\in\domz} A_z \;.
\]
That is, compared to algorithm of Section~\ref{sec:alg-basic}, we now re-weight 
the statistics by $\sigma_z\omega_z$ instead of $\sigma_z$ 
(since, intuitively, the flattening is done with ``$t_{1,z},t_{2,z}$'' samples for which 
$\sqrt{t_{1,z}t_{2,z}}=\Theta(\omega_z)$ samples, we multiply the weight by the ``flattening amount'').

\noindent Recalling that $\ell_1\geq \ell_2$ without loss of generality, we set
\begin{align}\label{eq:m:choice:flattened}
m\geq \zeta\max\!\Big(
    &\min\mleft( \frac{n^{7/8}\ell_1^{1/4}\ell_2^{1/4}}{\eps},\frac{n^{6/7}\ell_1^{2/7}\ell_2^{2/7}}{\eps^{8/7}},\frac{n\ell_1^{1/2}\ell_2^{1/2}}{\eps} \mright),
     \min\mleft( \frac{n^{3/4}\ell_1^{1/2}\ell_2^{1/2}}{\eps},\frac{\ell_1^{2}\ell_2^{2}}{\eps^4},\frac{n\ell_1^{1/2}\ell_2^{3/2}}{\eps} \mright), \notag\\
    &\min\mleft( \frac{n^{2/3}\ell_1^{2/3}\ell_2^{1/3}}{\eps^{4/3}},\frac{\ell_1\ell_2}{\eps^4}, \frac{\sqrt{n}\ell_1\sqrt{\ell_2}}{\eps^2}, \frac{n\ell_1^{3/2}\ell_2^{1/2}}{\eps} \mright),
     \min\mleft( \frac{\sqrt{n\ell_1\ell_2}}{\eps^2},\frac{\ell_1\ell_2}{\eps^4}\mright)
\Big) \;,
\end{align}
for some sufficiently big absolute constant $\zeta \geq 1$. The resulting pseudo-code is given in~\cref{algo:general:testing}.

\begin{algorithm}[h!t]
    \begin{algorithmic}[1]
      \Require Parameter $n\eqdef \abs{\domz}$, $\ell_1\eqdef \abs{\domx}$, $\ell_2\eqdef \abs{\domy}$, $\eps\in(0,1]$, 
                    and sample access to $\p\in\distribs{\domx\times\domy\times\domz}$.
      \State Set $m$ as in~\cref{eq:m:choice:flattened} \Comment{$\zeta \geq 1$ is an absolute constant}
      \State Set $\tau \gets \newer{\zeta^{1/4} \sqrt{\min(n,m)}}$. \Comment{Threshold for accepting}
      \State Draw $M\sim \poisson{m}$ samples from $\p$ and let $S$ be the multi-set of samples.
      \ForAll{ $z\in\domz$ }
        \State Let $S_z\subseteq \domx\times\domy$ be the multi-set $S_z\eqdef \setOfSuchThat{(x,y)}{(x,y,z)\in S}$.
        \If{ $\abs{S_z} \geq 4$ } \Comment{Enough samples to call $\Phi$}
          \State Set $N_z\gets 4\flr{(\abs{S_z}-4)/4}$, and let $S'_z$ be the multi-set of the first $N_z$ elements of $S_z$. \Comment{$N_z=4+4t_z$ for some integer $t_z$.}
          \State Set $t_{1,z} \gets \min(t_z,\ell_1)$, $t_{2,z} \gets \min(t_z,\ell_2)$, and divide $S'_z$ into disjoint $S'_{\mathcal{F},z}$, $S'_{\mathcal{T},z}$ of size $t_{1,z}+t_{2,z}$ and $\sigma_z\eqdef 2t_z+4$, respectively.
          \State $(a^{(z)}_{x,y})_{(x,y)\in \domx\times\domy} \gets \textsc{ImplicitFlattening}(S'_{\mathcal{F},z})$ \Comment{Flatten $\domx\times\domy$ using $S'_{\mathcal{F},z}$ as explained in the first bullet of~\cref{sec:alg:flattened:flattening}, by calling~\cref{algo:general:implicit:flattening}}
          \State $\Phi_z \gets \textsc{UnbiasedEstimator}((a^{(z)}_{x,y})_{(x,y)\in \domx\times\domy}, S'_{\mathcal{F},z})$  \Comment{Compute $\Phi(S'_{\mathcal{F},z})$, the unbiased estimator of $Q$ as defined in the second bullet of~\cref{sec:alg:flattened:flattening}, by calling~\cref{algo:general:implicit:estimator}} 
          \State Set $A_z\gets \sigma_z \omega_z \cdot \Phi_z$, where $\omega_z \gets \sqrt{\min\left(\sigma_z,\ell_1\right)\min\left(\sigma_z,\ell_2\right)}$.
        \Else
          \State Set $A_z\gets 0$.
        \EndIf
      \EndFor
      \If{ $A\eqdef \sum_{z\in \domz} A_z \leq \tau$ }
        \State \Return \accept
      \Else
        \State \Return \reject
      \EndIf
    \end{algorithmic}
    \caption{\textsc{TestCondIndependenceGeneral}}\label{algo:general:testing}
\end{algorithm}

\begin{algorithm}[h!t]
    \begin{algorithmic}[1]
      \Require Multi-set $S\subseteq \domx\times\domy$.
      \State\Comment{This simulates the construction of the ``flattening set'' as per~\cref{sec:alg:flattened:flattening}; by~\cref{remark:split:l2:chi2}, it is actually sufficient to compute the corresponding normalization coefficients $a_{x,y}$, which we perform below.}
      \State \Comment{All $b_x$ and $c_y$ are initialized to $0$}
      \ForAll{$(x,y)\in S$}
        \State $b_x \gets b_x + 1$
        \State $c_y \gets c_y + 1$
      \EndFor
      \State \Comment{Note that the step below can be done more efficiently by only looping through elements $(x,y)$ for which either $b_x$ or $c_y$ is positive}
      \ForAll{$(x,y)\in \domx\times\domy$}
        \State $a_{x,y} \gets (1+b_x)(1+c_y)-1$ \Comment{Implement~\cref{eq:flattening:axy}}
      \EndFor
      \State\Return $(a_{x,y})_{(x,y)\in \domx\times\domy}$
    \end{algorithmic}
    \caption{\textsc{ImplicitFlattening}}\label{algo:general:implicit:flattening}
\end{algorithm}

\begin{algorithm}[h!t]
    \begin{algorithmic}[1]
      \Require Set of coefficients $(a_{x,y})_{(x,y)\in \domx\times\domy}$, multi-set of samples $S\subseteq\domx\times\domy$.
      \State\Comment{This computes the unbiased estimator $U_N R_T$ for $Q(\p_T)=R_T(\p)$  from the samples in $S$, as explained in~\cref{sec:alg:flattened:flattening}: where \[
      R_T(X) = \sum_{(x,y)\in \domx\times\domy} \frac{1}{1+a_{x,y}}\Delta_{x,y}(X)^2
      \]}
      \State Let $N\gets \abs{S}$.
      \State \Comment{Recall that $\Phi_{S,x,y}$ denotes the count of occurrences of $(x,y)$ in the multi-set $S$}
      \ForAll{$(x,y)\in\domx\times\domy$} \Comment{Compute for $U_N \Delta_{x,y}(\Phi_S)^2$, from~\cref{eq:explicit:formula:estimator:l2}}
        \State $\Phi_{S,-x,-y} \gets \sum_{x'\neq x}\sum_{y'\neq y} \Phi_{S,x',y'}$
        \State $\Phi_{S,-x,y} \gets\sum_{x'\neq x} \Phi_{S,x',y}$
        \State $\Phi_{S,x,-y} \gets  \sum_{y'\neq y} \Phi_{S,x,y'}$
        \State $C_{i,j} \gets (\Phi_{S,i,j} \Phi_{S,-i,-j} - \Phi_{S,-i,j}  \Phi_{S,-i,-j})^2 +  \Phi_{S,i,j} \Phi_{S,-i,-j} (1-\Phi_{S,i,j} - \Phi_{S,-i,-j}) + \Phi_{S,-i,j} \Phi_{S, i,-j} (1-\Phi_{S,-i,j} - \Phi_{S,i,-j})$
      \EndFor
      \State\Return $\frac{(N-4)!}{N!}\sum_{(x,y)\in \domx\times\domy} \frac{1}{1+a_{x,y}}C_{i,j}$
    \end{algorithmic}
    \caption{\textsc{UnbiasedEstimator}}\label{algo:general:implicit:estimator}
\end{algorithm}

\subsection{Discussion of the Sample Complexity}\label{ssec:sample:complexity:discussion:flattened}

The expression of our sample complexity in~\cref{eq:m:choice:flattened} may seem rather complicated. 
We argue here that it captures at least some of the regimes of our four parameters in a tight way:
\begin{itemize}
  \item For $\ell_1=\ell_2=2$, we fall back to the case $\domx=\domy=\{0,1\}$, for which we had proven a tight bound in~\cref{sec:alg-basic}. Note that in this case the expression of $m$ in~\cref{eq:m:choice:flattened} reduces to 
  \[
  \bigO{\max\mleft(\min(n^{7/8}/\eps,n^{6/7}/\eps^{8/7}),\sqrt{n}/\eps^2\mright)}\;,
  \] matching the bounds of~\cref{sec:alg-basic}.
  
  \item For $n=1$ (and $\ell_1\geq\ell_2$ as before) we fall back to the independence testing problem~\cite{BFFKRW:01,LRR11, ADK15, DK:16}, 
  for which the tight sample complexity is known to be 
  $\bigTheta{\max\mleft( \ell_1^{2/3}\ell_2^{1/3}/\eps^{4/3},\sqrt{\ell_1\ell_2}/\eps^2 \mright)}$~\cite{DK:16}. 
  It is easy to see that, with these parameters,~\cref{eq:m:choice:flattened} 
  reduces to $\bigO{\max\mleft( \ell_1^{2/3}\ell_2^{1/3}/\eps^{4/3},\sqrt{\ell_1\ell_2}/\eps^2 \mright)}$ as well.
  
  \item For $\ell_1=\ell_2=n$ (and $\eps$ not too small), the choice of $m$ reduces to $O(n^{7/4}/\eps)$. 
  This matches the $\bigOmega{n^{7/4}}$ lower bound of~\cref{sec:lb:nnn} for $\eps=1/20$.
\end{itemize}

\begin{remark}
  We further note that the expression of~\cref{eq:m:choice:flattened}, which emerges from the analysis, 
  can be simplified by a careful accounting of the regimes of the parameters. Namely, one can show that it is equivalent to
  \begin{equation}\label{eq:m:choice:flattened:simplified}
m\geq \beta\max\mleft(
    \min\mleft( \frac{n^{7/8}\ell_1^{1/4}\ell_2^{1/4}}{\eps},\frac{n^{6/7}\ell_1^{2/7}\ell_2^{2/7}}{\eps^{8/7}} \mright),
     \frac{n^{3/4}\ell_1^{1/2}\ell_2^{1/2}}{\eps},
    \frac{n^{2/3}\ell_1^{2/3}\ell_2^{1/3}}{\eps^{4/3}},
     \frac{n^{1/2}\ell_1^{1/2}\ell_2^{1/2}}{\eps^2}
\mright) \;.
\end{equation}
\end{remark}

\subsection{The Different Sources of Randomness}\label{sec:alg:flattened:randomness}

As the argument will heavily rely on conditioning on \emph{some} of the randomness at play and analyzing the resulting conditional expectations and variances, it is important to clearly state upfront what the different sources of randomness are and how we refer to them.

In what follows, we will use the following notations: for each bin $z\in\domz$,
\begin{itemize}
  \item $\sigma_z$ is the number of samples from $\p$ we obtain with the $Z$ coordinate falling in bin $z$;
  \item $T_z$ is the randomness corresponding to the flattening of $\domx$, $\domy$ for the corresponding bin $z$ (as described in~\cref{sec:alg:flattened:flattening});
  \item $R_z$ is the randomness of the estimator $\Phi_{S_z}$ on bin $z$.
\end{itemize}
Accordingly, we will write $\sigma=(\sigma_z)_{z\in\domz}$, $T=(T_z)_{z\in\domz}$, and $R=(R_z)_{z\in\domz}$ for the three sources of randomness (over all bins).

\subsection{Analyzing \texorpdfstring{$A$}{A}}\label{sec:alg:flattened:analysis:A}

\newer{
The goal of this subsection is to show that, with high probability over $\sigma,T$, the following holds:
\begin{itemize}
  \item If $\p$ is indeed conditionally independent, then $\expectCond{A}{\sigma,T} = 0$ and $\var[A\mid \sigma,T] = O(\min(n,m))$.
  \item If $\p$ is far from conditionally independent, then $\expectCond{A}{\sigma,T} = \Omega(\sqrt{\min(n,m)})$ and $\var[A\mid \sigma,T]$ 
  is ``not too big''compared to $\min(n,m)$ and $\expectCond{A}{\sigma,T}$.
\end{itemize}
This high-probability guarantee will allow us to use Chebyshev's inequality in~\cref{sec:alg:flattened:conclusion} to conclude that, 
by comparing $A$ to a suitably chosen threshold, we can distinguish between the two cases with high probability (both over $\sigma,T$ and $R$).

The reason for which we only obtain the above guarantees ``with high probability over $\sigma,T$'' is, roughly speaking, 
that we need to handle the complicated dependencies between $A$ and $T$, 
which prevent us from analyzing $\expect{A}$ and $\var[A]$ directly. 
To do so, we introduce an intermediate statistic, $D$ (which itself only depends on $\sigma$ and $R$, but not on the flattening randomness $T$), 
and relate it to $\expectCond{A}{\sigma,T}$. This enables us to analyze the expectation and variance of $D$
instead of $\expectCond{A}{\sigma,T}$, before concluding by Markov's and (another application of) 
Chebyshev's inequality that these bounds carry over to $\expectCond{A}{\sigma,T}$ \emph{with high probability over $\sigma,T$}.
} 

\subsubsection{The Expectation of \texorpdfstring{$A$}{A}}\label{sec:alg:flattened:expectation}
We have that
\begin{equation}\label{eq:condexpect:az:tzsz}
      \expectCond{ A_z }{ \sigma_z, T_z } = \sigma_z\omega_z\normtwo{\p_{T_z}-\q_{T_z}}^2\indic{\sigma_z\geq 4}
\end{equation}
but $Q(\p_{T_z})=\normtwo{\p_{T_z}-\q_{T_z}}^2$ depends on $T_z$. To get around this, we will start by analyzing $D\eqdef \sum_{z\in\domz} D_z$, 
where 
\[
      D_z \eqdef \sigma_z\omega_z  \frac{\eps_z^2}{\ell_1\ell_2} \indic{\sigma_z\geq 4} \;,
\]
and $\eps_z \eqdef \totalvardist{\p_z}{\q_z}$. Note that now $D$ only depends on $R$ and $\sigma$ (and no longer on $T$).
For simplicity, we will often write $\eps'_z = \frac{\eps_z}{\sqrt{\ell_1\ell_2}}$.
 Next we show that whatever flattenings $T=(T_z)_{z\in\domz}$ we use given $\sigma=(\sigma_z)_{z\in\domz}$, $D$ is a lower bound for the conditional expectation of $A$: 
\begin{lemma}\label{lemma:D:lowerbound:AsigmaT:flattened}
    $\expectCond{ A }{ \sigma, T } \geq D$.
\end{lemma}
\begin{proof}
Using that $\normtwo{\p_{T_z}-\q_{T_z}} \geq \frac{2\eps_z}{\sqrt{(\ell_1+t_{1,z})(\ell_2+t_{2,z})}} \geq \frac{\eps_Z}{\sqrt{\ell_1\ell_2}}$, we have that
\[
D = \sum_{z\in\domz} \sigma_z\omega_z  \frac{\eps_z^2}{\sqrt{\ell_1\ell_2}} \indic{\sigma_z\geq 4} \leq \sum_{z\in\domz} \sigma_z\omega_z \normtwo{\p_{T_z}-\q_{T_z}}^2 \indic{\sigma_z\geq 4} = A \; .
\]
\end{proof}

We will require the following analogue of~\cref{lemma:exp-lb} for $D$:  
  \begin{lemma} \label{lemma:exp-lb:flattened}
    For $z\in\domz$, let  $\alpha_z\eqdef m \cdot \p_Z(z)$. 
    Then, we have that:
    \begin{equation}\label{eq:lowerbound:expectation:flattened}
      \expect{D} \geq \gamma \cdot \sum_{z\in\domz} \eps_z^2 \min(\alpha_z \beta_z, \alpha_z^4)
    \end{equation}
    for some absolute constant $\gamma>0$, 
    where $\beta_z \eqdef \sqrt{\min(\alpha_z,\ell_1)\min(\alpha_z,\ell_2)}$.
  \end{lemma} 
\begin{proof}
From the definition of $D$, we obtain that its expectation is:
  \[
      \expect{D} = \sum_z \shortexpect_\sigma\!\left[ \sigma_z\omega_z \indic{\sigma_z\geq 4} \eps'^2_z \right]\,.
  \]
  Now
  \[
      \expect{D} = \sum_z \eps'^2_z \shortexpect_\sigma\!\left[ \sigma_z\omega_z \indic{\sigma_z\geq 4}  \right]
      = \Omega(1) \sum_z \eps'^2_z \min(\alpha_z \beta_z,\alpha_z^4) \;,
  \]
  using the fact (\cref{claim:expectation:truncated:poisson:squared:with:min}) that, for a Poisson random variable $X$ with parameter $\lambda$, $\expect{ X \sqrt{\min(X,a)\min(X,b)}\indic{X\geq 4}  } \geq \gamma \min(\lambda \sqrt{\min(\lambda,a)\min(\lambda,b)}, \lambda^4)$, for some absolute constant $\gamma>0$.
  \end{proof}

We will leverage this lemma to show the following lower bound on the expectation of $D$:
\begin{proposition}\label{prop:exp-gap:flattened}
If $\totalvardist{\p}{ \condindprop{\domx}{\domy}{\domz}}> \eps$, 
then $\expect{D} = \bigOmega{\zeta \sqrt{\min(n,m)}}$ (where $\zeta$ the constant in the definition of $m$).
\end{proposition}
\begin{proof}
Since $\totalvardist{\p}{ \condindprop{\domx}{\domy}{\domz}}> \eps$, we have that 
  $\sum_{z\in\domz} \eps'_z \alpha_z \geq \frac{1}{2\sqrt{\ell_1\ell_2}}\sum_{z\in\domz} \eps_z \alpha_z > \frac{m\eps}{\sqrt{\ell_1\ell_2}}$.

 We once again divide $\domz$ into heavy and light bins, $\domz_{H}\eqdef \setOfSuchThat{z}{\alpha_z^3 \geq \beta_z }$ and $\domz_{L}\eqdef \domz\setminus \domz_{H}$. By the above, we must have $\sum_{z\in\domz_{H}} \eps_z \alpha_z > m\eps$ or $\sum_{z\in\domz_{L}} \eps_z \alpha_z > m\eps$. We proceed as in the proof of~\cref{prop:exp-gap} to handle these two cases.
  \begin{itemize}
    \item In the first case, we want to lower bound $\sum_{z\in\domz_{H}} \eps'^2_z \alpha_z\beta_z$. We consider three sub-cases, partitioning $\domz_{H}$ in $3$: (1) $\domz_{H,1} \eqdef \setOfSuchThat{ z\in\domz }{ \ell_2\leq \ell_1 < \alpha_z }$, (2) $\domz_{H,2} \eqdef \setOfSuchThat{ z\in\domz }{ \ell_2 \leq \alpha_z \leq \ell_1  }$, and (3) $\domz_{H,3} \eqdef \setOfSuchThat{ z\in\domz }{ \alpha_z < \ell_2 \leq \ell_1  }$. By a similar argument, at least one of these sets is such that $\sum_{z\in\domz_{H,i}} \eps_z\alpha_z > \frac{1}{3}m\eps'$.
      \begin{itemize}
        \item In the first sub-case:
          \[
              \sum_{z\in\domz_{H,1}} \eps'^2_z \alpha_z\beta_z
              = \sqrt{\ell_1\ell_2}\sum_{z\in\domz_{H,1}} \eps'^2_z \alpha_z
              \geq \sqrt{\ell_1\ell_2}\frac{\left(\sum_{z\in\domz_{H,1}} \eps'_z \alpha_z\right)^2}{\sum_{z\in\domz_{H,1}} \alpha_z}
              \geq \sqrt{\ell_1\ell_2}\frac{\left(\sum_{z\in\domz_{H,1}} \eps'_Z \alpha_z\right)^2}{m}
          \]
          by Cauchy--Schwarz and recalling that $\sum_{z\in\domz_{H,1}} \alpha_z\leq \sum_{z\in\domz} \alpha_z = m$; 
          and again by Jensen's inequality after taking expectations on both sides,
          \begin{equation}\label{eq:flattened:expect:case1:subcase:1}
              \sum_{z\in\domz_{H,1}} \eps'^2_z \alpha_z\beta_z 
                \geq \sqrt{\ell_1\ell_2}\frac{\left(\sum_{z\in\domz_{H,1}} \eps'_z \alpha_z\right)^2}{m} 
                > \sqrt{\ell_1\ell_2}\frac{1}{36}m{\eps'}^2 = \frac{1}{36}\frac{m\eps^2}{\sqrt{\ell_1\ell_2}}\,.
          \end{equation}
        \item In the second sub-case:
          \[
              \sum_{z\in\domz_{H,2}}\eps'^2_z \alpha_z\beta_z
              = \sqrt{\ell_2}\sum_{z\in\domz_{H,2}} \eps'^2_z \alpha_z^{3/2}
              \geq \sqrt{\ell_2}\frac{\left(\sum_{z\in\domz_{H,2}} \eps'_z \alpha_z\right)^2}{\sum_{z\in\domz_{H,2}} \sqrt{\alpha_z}}
              \geq \sqrt{\ell_2}\frac{\left(\sum_{z\in\domz_{H,2}} \eps'_z \alpha_z\right)^2}{\min(\sqrt{mn},m/\sqrt{\ell_1})}
          \]
          by Jensen's inequality and then using that $\sum_{z\in\domz_{H,2}} \sqrt{\alpha_z} = \sqrt{m}\sum_{z\in\domz_{H,2}} \sqrt{\p_Z(z)} \leq \sqrt{mn}$, and also that by definition of $\domz_{H,2}$  we have $\sum_{z\in\domz_{H,2}} \sqrt{\alpha_z} \leq \sqrt{m} \sum_{z\in\domz_{H,2}} \sqrt{\p_Z(z)}\leq \sqrt{m}\frac{m}{\ell_1} \sqrt{\frac{\ell_1}{m}} = \frac{m}{\sqrt{\ell_1}}$. Again by Jensen's inequality after taking expectations on both sides,
          \begin{align}
              \sum_{z\in\domz_{H,2}} \eps'^2_z \alpha_z\beta_z 
                &\geq \sqrt{\ell_2}\frac{\left(\sum_{z\in\domz_{H,2}} \eps'_z \alpha_z\right)^2}{\min(\sqrt{mn},m/\sqrt{\ell_1})} 
                > \sqrt{\ell_2}\frac{1}{36}\frac{m^2{\eps'}^2}{{\min(\sqrt{mn},m/\sqrt{\ell_1})}} \notag\\
                &= \frac{1}{36}\frac{m^{3/2}\eps^2}{\ell_1\sqrt{\ell_2}}\max\mleft(\frac{1}{\sqrt{n}},\sqrt{\frac{\ell_1}{m}}\mright)\,. \label{eq:flattened:expect:case1:subcase:2}
          \end{align}
        However, note that since
        $\sum_{z\in\domz_{H,2}} \eps'_z \alpha_z \leq \sqrt{2} \sum_{z\in\domz_{H,2}} \alpha_z \leq \sqrt{2}\abs{\domz_{H,2}}\ell_1 \leq \sqrt{2}n\ell_1$
      (as $\alpha_z \leq \ell_1$ for $z\in\domz_{H,2}$), the second sub-case cannot happen if $m\eps' \geq 2\sqrt{2}n\ell_1$.
        \item In the third sub-case:
          \[
              \sum_{z\in\domz_{H,3}} \eps'^2_z \alpha_z\beta_z
              = \sum_{z\in\domz_{H,3}} \eps'^2_z \alpha_z^{2}
              \geq \frac{\left(\sum_{z\in\domz_{H,3}} \eps'_z \alpha_z\right)^2}{\sum_{z\in\domz_{H,3}} 1}
              \geq \frac{\left(\sum_{z\in\domz_{H,3}} \eps'_z \alpha_z\right)^2}{\min(n,m)}
          \]
          by Jensen's inequality and recalling that $\abs{\domz_{H,3}}\leq \min(n,m)$; and again by Jensen's inequality after taking expectations on both sides,
          \begin{equation}\label{eq:flattened:expect:case1:subcase:3}
              \sum_{z\in\domz_{H,3}} \eps'^2_z \alpha_z 
              \geq \frac{\left(\sum_{z\in\domz_{H,3}} \eps'_z \alpha_z\right)^2}{\min(n,m)} 
              > \frac{1}{36}\frac{m^2}{\min(n,m)}{\eps'}^2 = \frac{1}{36}\max\mleft(\frac{m^2}{n},m\mright)\frac{\eps^2}{\ell_1\ell_2}\,. 
          \end{equation}
      However, note that since
        $\sum_{z\in\domz_{H,3}} \delta_z \alpha_z \leq \sqrt{2} \sum_{z\in\domz_{H,3}} \alpha_z \leq \sqrt{2}\abs{\domz_{H,3}}\ell_2 \leq \sqrt{2}n\ell_2$
      (as $\alpha_z < \ell_2$ for $z\in\domz_{H,3}$), the third sub-case cannot happen if $m\eps' \geq 2\sqrt{2}n\ell_2$.
      \end{itemize}
    \item In the second case, we want to lower bound 
    $
        \sum_{z\in \domz_{L}} \eps'^2_z \alpha_z^4
    $. We then use the same chain of (in)-equalities as in the second case of~\cref{prop:exp-gap}, to obtain
    \[
        \sum_{z\in \domz_{L}} \eps'^2_z \alpha_z^4 \geq \frac{\left(\sum_{z\in \domz_{L}}\eps'_z \alpha_z\right)^4}{\left(\sum_{z\in \domz_{L}} \eps'^{2/3}_z \right)^3} \;,
    \]
    and recall that $\eps'_z = \frac{\eps_z}{2\sqrt{\ell_1\ell_2}} \leq \frac{1}{\sqrt{\ell_1\ell_2}}$ to conclude
    \begin{equation}\label{eq:flattened:expect:case2}
    \sum_{z\in \domz_{L}} \eps'^2_z \alpha_z^4 
    \geq \frac{\ell_1\ell_2}{4n^3} \Big(\sum_{z\in \domz_{L}} \eps'_z \alpha_z\Big)^4
    = \frac{\ell_1\ell_2}{4n^3} \Big(\frac{1}{2\sqrt{\ell_1\ell_2}}\sum_{z\in \domz_{L}} \eps_z \alpha_z\Big)^4
    > \frac{1}{8}\frac{m^4\eps^4}{n^3\ell_1\ell_2} \;.
    \end{equation}
    However, note that since
      $\sum_{z\in\domz_{L}} \delta_z \alpha_z \leq \sqrt{2} \sum_{z\in\domz_{L}} \alpha_z \leq \sqrt{2}\abs{\domz_{L}} \leq \sqrt{2}n$
    (as $\alpha_z\leq 1$ for $z\in\domz_{L}$), the second case cannot happen if $m\eps' \geq 2\sqrt{2}n$.
  \end{itemize}
  
It remains to use~\cref{eq:flattened:expect:case1:subcase:1,eq:flattened:expect:case1:subcase:2,eq:flattened:expect:case1:subcase:3,eq:flattened:expect:case2} and our setting of $m$ to show that $\expect{D} \geq C\sqrt{\min(n,m)}$ (where the constant $C>0$ depends on the choice of the constant in the definition of $m$).
\begin{itemize}
\item From~\cref{eq:flattened:expect:case1:subcase:1} and the fact that $m \geq \zeta \min(\ell_1\ell_2/\eps^4, \sqrt{n\ell_1\ell_2}/\eps^2)$, we get
  \[
      \sum_{z\in\domz_{H,1}} \eps'^2_z \alpha_z\beta_z \gg \sqrt{\zeta \min(n,m)}
  \]
  in the first sub-case of the first case.
\item From~\cref{eq:flattened:expect:case1:subcase:2} and the fact that $m \geq \zeta \min(n^{2/3}\ell_1^{2/3}\ell_2^{1/3}/\eps^{4/3},\ell_1\ell_2/\eps^4,\sqrt{n}\ell_1\sqrt{\ell_2}/\eps^2,n\ell_1^{3/2}\ell_2^{1/2}/\eps)$, we get
  \[
      \sum_{z\in\domz_{H,2}} \eps'^2_z \alpha_z\beta_z \gg \sqrt{\zeta \min(n,m)}
  \]
  in the second sub-case of the first case (depending on whether $\min(n,m)\min\mleft(n,\frac{m}{\ell_1}\mright)$ is equal to $n^2$, $m^2/\ell_1$, or $mn$). (The last term in the $\min$ enforcing the condition that this sub-case can only happen whenever $m \eps' = O(n \ell_1)$.)
\item From~\cref{eq:flattened:expect:case1:subcase:3} and the fact that $m \geq \zeta \min(n^{3/4}\ell_1^{1/2}\ell_2^{1/2}/\eps,\ell_1^{2}\ell_2^{2}/\eps^4,n\ell_1^{1/2}\ell_2^{3/2}/\eps)$, we get
  \[
      \sum_{z\in\domz_{H,3}} \eps'^2_z \alpha_z\beta_z \gg \sqrt{\zeta \min(n,m)}
  \]
  in the third sub-case of the first case (depending on whether $\sqrt{\min(n,m)}\min\mleft(\frac{1}{m},\frac{n}{m^2}\mright)$ is equal to $n^{3/2}/m^2$ or $1/m^{1/2}$). (The last term in the $\min$ enforcing the condition that this sub-case can only happen whenever $m \eps' = O(n \ell_2)$.)
\item From~\cref{eq:flattened:expect:case2} and the fact that $m \geq \zeta \min(n^{7/8}\ell_1^{1/4}\ell_2^{1/4}/\eps,n^{6/7}\ell_1^{2/7}\ell_2^{2/7}/\eps^{8/7},n\ell_1^{1/2}\ell_2^{1/2}/\eps)$, we get
  \[
      \sum_{z\in\domz_{L}} \eps'^2_z \alpha_z\beta_z \gg \zeta^2 \sqrt{\min(n,m)}
  \]
  in the second case (depending on whether $\min(n,m)$ is equal to $n$ or $m$). 
  (The last term in the $\min$ enforcing the condition that this sub-case can only happen whenever $m \eps' = O(n)$.)
\end{itemize}
This completes the proof of Proposition~\ref{prop:exp-gap:flattened}.
\end{proof}

\subsubsection{Variances of \texorpdfstring{$D$}{D} and \texorpdfstring{$A$}{A}}\label{sec:alg:flattened:variance}
First we bound the variance of $D$: 
\begin{lemma}\label{lemma:D:variance:expectation:flattened}
\[
  \var[D] \leq O(\expect{D})\,.
\]
\end{lemma}
\begin{proof}
 Recall that $D= \sum_{z\in\domz} \sigma_z\omega_z  \eps'^2_z \indic{\sigma_z\geq 4}$. Since the $\sigma_z$'s are independent and $D_z$ is a function of $\sigma_z$, the $D_z$'s are independent as well and so
  \[
      \var[D] = \sum_{z\in\domz} \var[ \sigma_z\omega_z \eps'^2_z \indic{\sigma_z\geq 4} ]
      = \sum_{z\in\domz} \eps'^4_z \var_\sigma[ \sigma_z\omega_z\indic{\sigma_z\geq 4} ] \;.
  \]
  
  As $\sigma_z$ is distributed as $\poisson{\alpha_z}$, we can use \cref{claim:variance:truncated:poisson:2} to bound this
  \begin{align*}
      \var[D] &\leq C'\sum_{z\in\domz} \eps'^4_z\shortexpect_\sigma[ \sigma_z\omega_z\indic{\sigma_z\geq 4} ] \\
      & \leq C'\sum_{z\in\domz} \eps'^2_z \shortexpect_\sigma[ \sigma_z\omega_z\indic{\sigma_z\geq 4} ] \\
	  &= C' \expect{D} \;,
  \end{align*}
  for some absolute constant $C'>0$.  
\end{proof}

Since our statistic $A$ is a linear combination of the $\Phi_{S_z}$'s and all the $S_z$'s are independent by Poissonization, we get the analogue of~\cref{prop:var-ub}:
\begin{proposition}\label{prop:var-ub:in:expectation}
Let $E\eqdef \sum_{z\in\domz} \omega_z^2 B_{T_z}\indic{\sigma_z\geq 4}$.  Then,
\begin{equation}\label{eqn:bound:variance:in:expectation}
\var[A \mid \sigma,T] \leq C\left( E + E^{1/2}\expectCond{ A }{ \sigma, T } +  \expectCond{ A }{ \sigma, T }^{3/2} \right)  \;,
\end{equation}
where $\expectCond{E}{\sigma} = O( \min(n,M) )$ and $C> 0$ is some absolute constant.
\end{proposition}
\begin{proof}
  Since $\var[ A_z \mid \sigma_z, T_z ] = \sigma_z^2\omega_z^2 \indic{\sigma_z\geq 4} \var[ \Phi(S_z) \mid \sigma_z, T_z ]$, 
  we have by~\cref{eq:variance:estimator:flattened} that, for some absolute constant $C>0$,
  \begin{align*}
       \var[ A_z \mid \sigma_z, T_z ] 
      &\leq  C \left( \sigma_z^2\omega_z^2\left( \frac{\normtwo{\p_{T_z}-\q_{T_z}}^2 \sqrt{B_{T_z}}}{ \sigma_z} + \frac{B_{T_z}}{ \sigma_z^2} + \frac{\normtwo{\p_{T_z}-\q_{T_z}}^3}{\sigma_z} \right)  \indic{\sigma_z\geq 4} \right) \;.
  \end{align*} 
  We will handle the three terms of the RHS separately.  First, by Cauchy--Schwarz and monotonicity of $\lp[p]$ norms we get that
  \begin{align}
    \sum_{z\in\domz} \sigma_z \omega_z^2 \normtwo{\p_{T_z}-\q_{T_z}}^2 \sqrt{B_{T_z}}\indic{\sigma_z\geq 4} 
    &\leq \mleft(\sum_{z\in\domz} \omega_z^2 B_{T_z}\indic{\sigma_z\geq 4} \mright)^{1/2}
          \mleft(\sum_{z\in\domz} \mleft(\sigma_z\omega_z \normtwo{\p_{T_z}-\q_{T_z}}^2\mright)^2\indic{\sigma_z\geq 4} \mright)^{1/2} \notag\\
    &\leq \mleft(\sum_{z\in\domz} \omega_z^2 B_{T_z}\indic{\sigma_z\geq 4} \mright)^{1/2} \sum_{z\in\domz} \sigma_z\omega_z \normtwo{\p_{T_z}-\q_{T_z}}^2\indic{\sigma_z\geq 4} \notag\\
    &= E^{1/2}\expectCond{ A }{ \sigma, T }  \label{prop:var-ub:in:expectation:term:1} \;,
  \end{align}
  the last equality from~\cref{eq:condexpect:az:tzsz}.  
  
  Moreover, for the second term $\sigma_z^2 \omega_z^2 \frac{B_{T_z}}{ \sigma_z^2 }\indic{\sigma_z\geq 4}= \omega_z^2 B_{T_z}\indic{\sigma_z\geq 4}$, it is immediate that summing over all bins we get $\sum_{z\in\domz }\omega_z^2 B_{T_z}\indic{\sigma_z\geq 4} = E$. 

Let us now turn to the last term of our upper bound on the variance. We can write
\begin{align*}
\sigma_z^2\omega_z^2\frac{\normtwo{\p_{T_z}-\q_{T_z}}^3}{\sigma_z}\indic{\sigma_z\geq 4} &= \sigma_z\omega_z^2 \normtwo{\p_{T_z}-\q_{T_z}}^3 \indic{\sigma_z\geq 4}
= \sqrt{\frac{\omega_z}{\sigma_z}} \sigma_z^{3/2}\omega_z^{3/2} \normtwo{\p_{T_z}-\q_{T_z}}^3 \indic{\sigma_z\geq 4} \\
&\leq \left( \sigma_z\omega_z \normtwo{\p_{T_z}-\q_{T_z}}^2 \indic{\sigma_z\geq 4} \right)^{3/2}
= \expectCond{ A_z }{ \sigma, T }^{3/2}
\end{align*}
recalling that $\omega_z\leq \sigma_z$ by definition. We may use the inequality between the $\lp[1]$ and $\lp[3/2]$ norms to conclude that $\sum_{z \in \domz} \expectCond{ A_z }{ \sigma, T }^{3/2} \leq \expectCond{ A }{ \sigma, T }^{3/2}$, which leads by the above to
\begin{equation}\label{prop:var-ub:in:expectation:term:3}
  \expectCond{ \sum_{z\in\domz}\sigma_z\omega_z^2 \normtwo{\p_{T_z}-\q_{T_z}}^3 \indic{\sigma_z\geq 4} }{ \sigma, T }  \;.
  \leq \expectCond{ A }{ \sigma, T }^{3/2}
\end{equation}
Since the $A_z$'s are independent conditioned on $\sigma_z$ and $T_z$, we have
\[
    \var[A \mid \sigma, T ] = \sum_{z\in\domz} \var[A \mid \sigma_z, T_z ] \;,
\]
and therefore by~\cref{prop:var-ub:in:expectation:term:1,prop:var-ub:in:expectation:term:3} and the definition of $E$ we obtain
\begin{equation}
    \var[A \mid \sigma, T ]  \leq O\mleft( E^{1/2}\expectCond{ A }{ \sigma, T } + E + \expectCond{ A }{ \sigma, T }^{3/2} \mright)\,.
\end{equation}
It remains to establish the further guarantee that
$\expectCond{ E }{\sigma} = O(\min(n,M))$. To do so, observe that we can write, as $\omega_z$ only depends on the randomness $\sigma$,
  \begin{align*}
    \expectCond{ \omega_z^2 B_{T_z} \indic{\sigma_z\geq 4}}{\sigma}
    = \omega_z^2  \expectCond{ B_{T_z}}{\sigma} \indic{\sigma_z\geq 4}
    \leq \frac{\omega_z^2}{(1+t_{1,z})(1+t_{2,z})} \indic{\sigma_z\geq 4}
  \end{align*}
  by~\cref{eq:expected:l2norm:product:marginals:split}. Recalling that $t_{i,z} = \min\mleft( ({\sigma_z-4})/{4}, \ell_i \mright)$ by the definition of the flattening (\cref{sec:alg:flattened:flattening}) and that $\omega_z^2 = \min\mleft(\sigma_z, \ell_1 \mright)\min\mleft(\sigma_z, \ell_2 \mright)$, this leads to
  \begin{align*}
    \expectCond{ \omega_z^2 B_{T_z} \indic{\sigma_z\geq 4}}{\sigma}
    \leq O(1)\cdot\indic{\sigma_z\geq 4}\,.
  \end{align*}
  In particular, by summing over all bins $z$ this implies that
  \begin{equation}\label{prop:var-ub:in:expectation:term:2}
      \expectCond{ E }{\sigma} \leq O(1)\sum_{z\in\domz}\indic{\sigma_z\geq 4} \leq O(1)\sum_{z\in\domz}\indic{\sigma_z\geq 1} = O(\min(n,M))\;,
  \end{equation}
as claimed.

\begin{lemma}[Soundness] \label{lemma:AsigmaT:variance:expectation:soundness:flattened}
If $\totalvardist{\p}{ \condindprop{\domx}{\domy}{\domz}}> \eps$, then with probability at least $99/100$ over $\sigma,T$ we have simultaneously 
  $\expectCond{ A }{\sigma, T} = \bigOmega{\sqrt{\newerest{\zeta} \min(n,m)}}$ 
and 
  \[
    \var[A \mid \sigma, T] \leq \bigO{ \min(n,m) + \sqrt{\min(n,m)}\expectCond{ A }{ \sigma, T } +  \expectCond{ A }{ \sigma, T }^{3/2} }.
  \]
\end{lemma}
\begin{proof}
  By~\cref{lemma:D:lowerbound:AsigmaT:flattened}, we have that $D$ is a lower bound on $\expectCond{ A }{\sigma, T}$ for all $\sigma,T$. Since~\cref{prop:exp-gap:flattened,lemma:D:variance:expectation:flattened} further ensures that $\expect{D} \geq \Omega(\sqrt{\newerest{\zeta}\min(n,m)})$ and $\var[D] \leq O(\expect{D})$, applying Chebyshev's inequality on $D$ results in
  \begin{align*}
      \probaDistrOf{ \sigma,T }{ \expectCond{ A }{\sigma, T} < \kappa\sqrt{\newerest{\zeta} \min(n,m)} }
      &\leq \probaDistrOf{ \sigma,T }{ D < O(\expect{D}) }
      = \bigO{\frac{\var[D]}{\expect{D}^2}}\\
      &= \bigO{\frac{1}{\expect{D}}}
      = \bigO{\frac{1}{\sqrt{\newerest{\zeta} \min(n,m)}}}
      \leq \frac{1}{200}
  \end{align*}
  for some absolute constant $\kappa>0$. 
  This gives the first statement. For the second, we start from~\cref{eqn:bound:variance:in:expectation} 
  and we apply Markov's inequality to $E$: as $\expectCond{E}{\sigma} = O( \min(n,M) )$, 
  with probability at least $399/400$ we have $E \leq  400\expectCond{E}{\sigma} = O( \min(n,M) )$. 
  Moreover, recalling that $M=\sum_{z\in\domz}\sigma_z$ is a Poisson random variable 
  with parameter $m$, we have $\probaOf{M > 2m} \leq  399/400$. Therefore, by a union bound 
  \[
      \probaDistrOf{ \sigma,T }{ \var[A \mid \sigma,T] \geq  \kappa'\left( \min(n,m) + \sqrt{\min(n,m)}\expectCond{ A }{ \sigma, T } +  \expectCond{ A }{ \sigma, T }^{3/2} \right)  }
      \leq \frac{1}{400}+\frac{1}{400}=\frac{1}{200}
  \]
  again for some absolute constant $\kappa'>0$. 
  This gives the second statement. 
  A union bound over both events concludes the proof.
\end{proof}

\begin{lemma}[Completeness] \label{lemma:AsigmaT:variance:expectation:completeness:flattened}
If $\p \in \condindprop{\domx}{\domy}{\domz}$, then with probability at least $99/100$ over $\sigma,T$ we have simultaneously 
  $\expectCond{ A }{\sigma, T} = 0$ 
and 
  $\var[A \mid \sigma, T] \leq \bigO{ \min(n,m) }$.
\end{lemma}
\begin{proof}
  The first statement is obvious by the definition of $A$ as sum of the $A_z$'s, since $\eps_z=0$ for all $z\in\domz$. For the second, the proof is identical as that of~\cref{lemma:AsigmaT:variance:expectation:soundness:flattened}, but having only to deal with the term $E$ in the bound on the variance (as the others are zero).
\end{proof}
\end{proof}

\subsection{Completing the Proof}\label{sec:alg:flattened:conclusion}

Let our threshold $\tau$ be set to \new{$\zeta^{1/4} \sqrt{\min(n,m)}$}. Gathering the above pieces we obtain the following:
\begin{lemma}[Soundness]\label{lemma:AsigmaT:soundness:flattened}
  If $\p$ is $\eps$-far from conditionally independent, then $\probaOf{A < \tau } \leq \frac{1}{50}$.
\end{lemma}
\begin{proof}
We apply Chebyshev's inequality once more, this time to $A'\eqdef(A \mid \sigma, T)$ and relying on the bounds on its expectation and variance established in~\cref{lemma:AsigmaT:variance:expectation:soundness:flattened}. Specifically, let $\mathcal{E}$ denote the event that both bounds of~\cref{lemma:AsigmaT:variance:expectation:soundness:flattened} hold simultaneously; then
  \begin{align*}
      \probaOf{A \leq \tau } = \probaOf{ A' \leq \tau  } &\leq \probaCond{ A' \leq \tau }{ \mathcal{E} } + \proba[ \overline{\mathcal{E}} ] \\
      &\leq \probaCond{ \abs{ A' - \expectCond{ A }{\sigma, T} } \geq \frac{1}{2}\expectCond{ A }{\sigma, T} }{ \mathcal{E} } + \frac{1}{100}
  \end{align*}
  where the second line is because, conditioned on $\mathcal{E}$, $\expectCond{ A }{\sigma, T} \geq \newerest{\bigOmega{\sqrt{\zeta \min(n,m)}} \geq \zeta^{1/4} \sqrt{\min(n,m)} =}  2\tau$. It only remains to bound the first term:
  \begin{align*}
    \probaCond{ \abs{ A' - \expectCond{ A }{\sigma, T} } \geq \expectCond{ A }{\sigma, T} }{ \mathcal{E} }
    &\leq \bigO{\frac{ \min(n,m) + \sqrt{\min(n,m)}\expectCond{ A }{ \sigma, T } +  \expectCond{ A }{ \sigma, T }^{3/2} }{\expectCond{ A }{\sigma, T}^2}} \\
    &= \bigO{ \frac{\min(n,m)}{\expectCond{ A }{\sigma, T}^2} + \frac{\sqrt{\min(n,m)}}{\expectCond{ A }{\sigma, T}}+ \frac{1}{\expectCond{ A }{\sigma, T}^{1/2}}  } \\
    &\newerest{\leq \bigO{\frac{1}{\zeta^{1/4}}}} \leq  \frac{1}{100}
  \end{align*}
  for the choice of a sufficiently large constant $\zeta$ in the definition of $m$.
\end{proof}

\begin{lemma}[Completeness]\label{lemma:AsigmaT:completeness:flattened}
  If $\p$ is conditionally independent, then $\probaOf{A \geq \tau } \leq \frac{1}{50}$.
\end{lemma}
\begin{proof}
Analogously to the proof in the soundness case, we apply Chebyshev's inequality to $A'\eqdef(A \mid \sigma, T)$ and relying on~\cref{lemma:AsigmaT:variance:expectation:completeness:flattened}. Specifically, let $\mathcal{E}]$ denote the event that the bound of~\cref{lemma:AsigmaT:variance:expectation:completeness:flattened} holds; then
  \[
      \probaOf{A \geq \tau } = \probaOf{ A' \geq \tau  } \leq \probaCond{ A' \geq \tau }{ \mathcal{E}] } + \proba[ \overline{\mathcal{E}'} ]\,.
  \]
  To conclude, we bound the first term:
  \begin{align*}
    \probaCond{ A' \geq \tau }{ \mathcal{E}'] }
    \leq \frac{ \var[A'\mid \mathcal{E}'] }{\tau^2}
    &\leq \bigO{ \frac{\min(n,m)}{\tau^2} } \\
    &\newerest{\leq \bigO{\frac{1}{\zeta}} }\leq \frac{1}{100}
  \end{align*}
  again for the choice of a sufficiently large constant $\zeta$ in the definition of $m$.
\end{proof}
 
\section{Sample Complexity Lower Bounds: The Case of Constant \texorpdfstring{$|\domx|, |\domy|$}{|X|,|Y|}} \label{sec:lb:appendix}
   
In this section, we prove our tight sample complexity lower bound of
\[
  \bigOmega{\max(\min(n^{6/7}/\eps^{8/7},n^{7/8}/\eps),\sqrt{n}/\eps^2)}
\]
for testing conditional independence in the regime that 
$\domx = \domy=\{0,1\}$ and $\domz=[n]$. This matches the sample complexity of our algorithm in~\cref{sec:alg-basic},
up to constant factors. In the main body of this section, we prove each lower bound separately.

The following expression for the total variation distance will be useful in the analysis of the lower bound constructions:

\begin{fact}
For any $\p\in\distribs{\domx\times\domy\times\domz}$ 
for $\domx = \domy=\{0,1\}$ and $\domz=[n]$. 
we have that:
\begin{equation}\label{eq:relation:to:covariance}
\totalvardist{\p_z}{\q_z} =  2 \abs{\cov[\conddistr{X}{Z=z}, \conddistr{Y}{Z=z}]} = \normtwo{\p_z-\q_z} \;.
\end{equation}
\end{fact}
\begin{proof}
We have the following:
\begin{align*}
2\totalvardist{\p_z}{\q_z} 
&= \abs{\p_z(1,1)-(\p_z(1,0)+\p_z(1,1))\cdot (\p_z(0,1)+\p_z(1,1))}\\
&\quad + \abs{\p_z(1,0)-(\p_z(1,0)+\p_z(1,1))\cdot(\p_z(1,0)+\p_z(0,0))} \\
&\quad + \abs{\p_z(0,1)-(\p_z(0,1)+\p_z(0,0))\cdot (\p_z(0,1)+\p_z(1,1))} \\
&\quad + \abs{\p_z(0,0)-(\p_z(0,1)+\p_z(0,0))\cdot (\p_z(1,0)+\p_z(0,0))}\\
&= 4 \abs{\p_z(0,0) \cdot \p_z(1,1)-\p_z(0,1) \cdot \p_z(1,0)} \\
&= 4\abs{\cov[\conddistr{X}{Z=z}, \conddistr{Y}{Z=z}]} \;.
\end{align*}
\end{proof}

\subsection{First Lower Bound Regime\texorpdfstring{: $\bigOmega{n^{6/7}/\eps^{8/7}}$ for $\eps > \frac{1}{n^{1/8}}$}{}} \label{ssec:lb-large-eps:appendix}

Assume that we are in the regime where
\[
  \max(\min(n^{6/7}/\eps^{8/7},n^{7/8}/\eps),\sqrt{n}/\eps^2) = n^{6/7}/\eps^{8/7} \;,
\]
i.e., $\eps > {1}/{n^{1/8}}$.
Suppose there is an algorithm for $\eps$-testing conditional independence 
drawing $m \leq c n^{6/7}/\eps^{8/7}$ samples from $\p$, 
for some sufficiently small universal constant $c>0$. 
Note that in this regime, we have $m \ll n$, 
i.e., we can assume that $m< c'n$ for some small constant $c'>0$.
  
\paragraph{The $\yes$-instance}
A pseudo-distribution $\p$ is drawn from the \yes-instances as follows: 
Independently for each value $z\in[n]$, we set:
\begin{itemize}
\item With probability $\frac{m}{n}$, we select
$\p_Z(z) = \frac{1}{m}$ and $\p_z(i, j) = \frac{1}{4}$ for all $i, j \in\{0,1\}$. In other words, we select 
uniform marginals for the conditional distributions $\p_z(i, j)$.          
\item With probability $1-\frac{m}{n}$, we select $\p_Z(z) = \frac{\eps}{n}$ and 
$\p_z(i, j)$ us defined by the $2\times 2$ matrix:
\begin{itemize}
\item With probability $1/2$,
$
\frac{1}{100}\begin{pmatrix}
16 & 24\\
24 & 36 
\end{pmatrix}  \eqdef Y_1 \;,
$
\item With probability $1/2$,
$
\frac{1}{100}\begin{pmatrix}
36 & 24\\
24 & 16 
\end{pmatrix}  \eqdef Y_2 \;$
\end{itemize}
\end{itemize}
It is easy to see that the resulting distribution $\p$ satisfies
\[
\expect{\sum_{z=1}^n \p_Z(z)} = n \left( \frac{m}{n}\cdot \frac{1}{m} + \left(1-\frac{m}{n}\right) \cdot \frac{\eps}{n} \right) \in [1,1+\eps] \;,
\]
i.e., the marginal for $Z$ has mass roughly $1$ in expectation, 
and that $\p\in\condindprop{\{0,1\}}{\{0,1\}}{[n]}$.
      
\paragraph{The $\no$-instance}
A pseudo-distribution $\p$ is drawn from the \no-instances as follows: 
Independently for each value $z\in[n]$, we set
\begin{itemize}
\item With probability $\frac{m}{n}$, we set 
$\p_Z(z) = \frac{1}{m}$ and $\p_z(i, j) = \frac{1}{4}$ for all $i,j \in\{0,1\}$.
\item With probability $1-\frac{m}{n}$, we set $\p_Z(z) = \frac{\eps}{n}$ and 
$\p_z(i, j)$ be defined by the $2\times 2$ matrix:
\begin{itemize}
\item With probability $1/8$,
$
\frac{1}{100}\begin{pmatrix}
6 & 24\\
24 & 46 
\end{pmatrix}  \eqdef N_1 \;,
$
\item With probability $1/8$,
$
\frac{1}{100}\begin{pmatrix}
46 & 24\\
24 & 6 
\end{pmatrix}  \eqdef N_2 \;,
$
\item With probability $3/4$,
$
\frac{1}{100}\begin{pmatrix}
26 & 24\\
24 & 26 
\end{pmatrix}  \eqdef N_3 \;.
$
\end{itemize}
\end{itemize}
Similarly, we have that 
$\expect{ \sum_{z=1}^n \p_Z(z)} \in[1,1+\eps]$. Furthermore, the expected total variation distance 
between such a $\p$ and the corresponding $\q \eqdef \sum_{z=1}^n \p_Z(z) \q_z \in \condindprop{\{0,1\}}{\{0,1\}}{[n]}$ is
\[
\expect{\totalvardist{\p}{\q}}
= \frac{1}{2}\left( n\left( \frac{m}{n}\cdot\frac{1}{m}\cdot 0+ \left(1-\frac{m}{n}\right) \frac{\eps}{n} \left(  \frac{1}{8}\cdot\frac{12}{100} + \frac{1}{8}\cdot\frac{12}{100} + \frac{3}{4}\cdot \frac{4}{100}  \right) \right)\right) = \frac{3}{100}\left(1-\frac{m}{n}\right)\eps > \frac{\eps}{100} \;,
\]
      where
      \[
          \frac{12}{100} = \abs{ \frac{1}{100}\begin{pmatrix}
                       46 & 24\\
                       24 & 6 
                    \end{pmatrix} - \frac{1}{100}\begin{pmatrix}
                       7\\
                       3 
                    \end{pmatrix}\begin{pmatrix}
                       7&
                       3 
                    \end{pmatrix} }, \qquad 
                    \frac{4}{100} = \abs{ \frac{1}{100}\begin{pmatrix}
                       26 & 24\\
                       24 & 26 
                    \end{pmatrix} - \frac{1}{100}\begin{pmatrix}
                       5\\
                       5 
                    \end{pmatrix}\begin{pmatrix}
                       5&
                       5 
                    \end{pmatrix} }.
      \]
Thus, $\expect{\totalvardist{\p}{\q}} = \bigOmega{\eps}$, 
which by~\cref{lemma:distance:product} implies that 
$\expect{\totalvardist{\p}{\condindprop{\{0,1\}}{\{0,1\}}{[n]}}} = \bigOmega{\eps}$.

\medskip

The next claim shows that, for each $z \in [n]$, 
the first three norms of the conditional distribution $\p_z(i, j)$ match, hence
do not provide any information towards distinguishing 
between the \yes- and \no-cases. Therefore, we need to get at least $4$ samples $(X,Y,Z)$ 
with the same value of $Z$ --- that is a $4$-collision with regard to $Z$ ---  
in order to have useful information. 

\medskip

\noindent {\bf Notation.} 
Given a $4$-variable function $R = R[X_1,X_2,X_3,X_4]$ 
and a real $2 \times 2$ matrix $M\in\mathcal{M}_2(\R)$, 
we will denote $R(M) \eqdef R(M_{1,1},M_{1,2},M_{2,1},M_{2,2})$. 

\medskip

We have the following:

\begin{claim}\label{claim:lb:lowmoments:appendix}
Let $Y_1, Y_2, N_1, N_2, N_3$ the probability matrices in the definition of the $\yes$ and $\no$-instances.
For every $4$-variable polynomial $R\in\R[X_1,X_2,X_3,X_4]$ of degree at most $3$, 
the following holds:
\[
\frac{1}{8}R(N_1)+\frac{1}{8}R(N_2)+\frac{3}{4}R(N_3) = \frac{1}{2}R(Y_1)+\frac{1}{2}R(Y_2) \;.
\]
\end{claim}
\begin{proof}
The first crucial observation is that the associated matrices can be expressed in the form 
\[
  (N_1,Y_1,N_3,Y_2,N_2) = (A+kB)_{0\leq k \leq 4} \;,
\]
where $A, B$ are the following matrices:
\[
A = \frac{1}{100}
\begin{pmatrix}
6 & 24\\
24 & 46 
\end{pmatrix},
\qquad B = \frac{1}{100}
\begin{pmatrix}
10 & 0\\
0 & -10 
\end{pmatrix} \;.
\]
Therefore, for any function $4$-variable function $R$ (not necessarily a polynomial), we have
\begin{align*}
\left( \frac{1}{8}R(N_1)+\frac{1}{8}R(N_2)+\frac{3}{4}R(N_3) \right) &- \left( \frac{1}{2}R(Y_1)+\frac{1}{2}R(Y_2) \right)\\
&= \frac{1}{8}R(N_1)-\frac{1}{2}R(Y_1)+\frac{3}{4}R(N_3)-\frac{1}{2}R(Y_2)+\frac{1}{8}R(N_2) \\
&= \frac{1}{8}\sum_{k=0}^4 (-1)^k \binom{4}{k}R(A+kB)= \frac{1}{8}\sum_{k=0}^4 (-1)^{4-k} \binom{4}{k}R(A+kB) \;,
\end{align*}
which is the $4$\textsuperscript{th}-order forward difference of $R$ at $A$ (more precisely, the fourth finite
difference of $f(k) = R(A+kB)$). 
Using the fact that the $(d+1)$\textsuperscript{th}-order forward difference of a degree-$d$ polynomial is zero, 
we get that the above RHS is zero for every degree-$3$ polynomial $R$.
\end{proof}
      
For the sake of simplicity and without loss of generality, 
we can use the Poissonization trick for the analysis of our lower bound construction 
(cf.~\cite{DK:16,CDKS17}). Specifically, instead of drawing $m$ independent samples from $\p$, 
we assume that our algorithm is provided with $m_z$ samples from the conditional distribution $\p_z$ 
(i.e., conditioned on $Z=z$), where the $(m_z)$'s are independent Poisson random variables 
with $m_z \sim \poisson{ m \p_Z(z) }$.
      
Consider the following process: we let $U\sim\bernoulli{\frac{1}{2}}$ be a uniformly random bit, 
and choose $\p$ to be selected as follows: 
(i) If $U=0$, then $p$ is drawn from the $\yes$-instances, 
(ii) If $U=1$, then $p$ is drawn from the $\no$-instances.
For every $z \in [n]$, let $a_z=(a^{00}_z,a^{01}_z,a^{10}_z,a^{11}_z)$ 
be the $4$-tuple of counts of $(i,j)_{i,j\in\{0,1\}}$ 
among the $m_z$ samples $(X,Y)\sim \p_z$. 
Accordingly, we will denote $A=(a_z)_{z\in[n]}$.
      
Following the mutual information method used in~\cite{DK:16}, 
to show the desired sample complexity lower bound of $\bigOmega{n^{6/7}/\eps^{8/7}}$, 
it suffices to show that
$\mutualinfo{U}{A} = o(1)$, 
unless $m=\bigOmega{n^{6/7}/\eps^{8/7}}$. 
Since the $(a_z)_{z\in[n]}$'s are independent conditioned on $U$, 
we have that $\mutualinfo{U}{A} \leq \sum_{z=1}^n \mutualinfo{U}{a_z}$, 
and therefore it suffices to bound from above separately
$\mutualinfo{U}{a_z}$ for every $z$. We proceed to establish such a bound
in the following lemma:
\begin{lemma}\label{lemma:lb:case3:main:appendix}
For any $z\in[n]$, we have $\mutualinfo{U}{a_z} = \bigO{ \frac{\eps^8 m^7}{n^7} }$.
\end{lemma}
Before proving the lemma, we show that it implies the desired lower bound. 
Indeed, assuming~\cref{lemma:lb:case3:main:appendix}, we get that
\[
\mutualinfo{U}{A} \leq \sum_{z=1}^n \mutualinfo{U}{a_z} = \sum_{z=1}^n \bigO{ \frac{\eps^8 m^7}{n^7} } = \bigO{ \frac{\eps^8 m^7}{n^6}} \;,
\]
which is $o(1)$ unless $m = \bigOmega{n^{6/7}/\eps^{8/7}}$. It remains to prove the lemma.

\begin{proofof}{\cref{lemma:lb:case3:main:appendix}}
By symmetry, it is sufficient to show the claim for $z=1$. 
To simplify the notation, let $a \eqdef a_1$. 
We first bound $\mutualinfo{U}{a}$ from above using~\cite[Fact 4.12]{CDKS17} (see also~\cite{DK:16}) as follows:
\[
\mutualinfo{U}{a} \leq \sum_{\alpha\in\N^4} \probaOf{ a=\alpha } 
\left( 1 - \frac{\probaCond{ a =\alpha }{ U=1 }}{ \probaCond{ a =\alpha }{ U=0 }} \right)^2 \eqdef \Phi(n,m,\eps) \;.
\]
Our next step is to get a hold on the conditional probabilities 
$\probaCond{ a =\alpha }{U=0}$ and $\probaCond{ a =\alpha }{U=1}$. 
For notational convenience, we set
\[
p_1 \eqdef \frac{16}{100}, \quad p_2 \eqdef \frac{24}{100}, \quad p_3 \eqdef \frac{36}{100},
\quad q_1 \eqdef \frac{6}{100},q_2 \eqdef \frac{26}{100}, \quad q_3 \eqdef \frac{46}{100} \;,
\]   
and let $\Xi$ denote the event that the bin is ``heavy'', i.e., 
that $\p_Z$ puts probability mass $\frac{1}{m}$ on it. Note that by construction this event happens 
with probability $\frac{m}{n}$ and that $\Xi$ is independent of $U$.
        
We start with the $\yes$-case. Recall that with probability $m/n$, $\Xi$ holds: 
the probability $\p_Z(1)$ equals $1/m$, in which case we draw $\poisson{m\cdot \frac{1}{m}}$ samples 
from $Z=1$ and each sample is uniformly random on $\{0,1\}\times\{0,1\}$. That is, 
each of the four outcomes is an independent $\poisson{m\cdot \frac{1}{m}\cdot \frac{1}{4}}$ random variable. 
With probability $1-m/n$, $\bar{\Xi}$ holds: $\p_Z(1)$ equals $\eps/n$ and we draw $\poisson{m\cdot \frac{\eps}{n}}$ samples from $Z=1$. 
With probability $1/2$, all samples follow the first case, and with probability $1/2$ all samples follow the second case. 

For any $\alpha=(\alpha_1,\alpha_2,\alpha_3,\alpha_4)\in\N^4$, we can explicitly calculate the associated probabilities.
Specifically, we can write:
\begin{equation}\label{eq:lb:yescase:appendix}
            \probaCond{ a = \alpha }{ U=0,\bar{\Xi} } = \frac{e^{-\frac{\eps m}{n}}}{\alpha_1!\alpha_2!\alpha_3!\alpha_4!} \left( \frac{1}{2}p_1^{\alpha_1}p_2^{\alpha_2}p_2^{\alpha_3}p_3^{\alpha_4} + \frac{1}{2}p_3^{\alpha_1}p_2^{\alpha_2}p_2^{\alpha_3}p_1^{\alpha_4} \right) \;,
\end{equation}  
and
\begin{align*}
\probaCond{ a = \alpha }{ U=0 } &=
\probaCond{ a = \alpha }{ U=0, \Xi }\cdot \probaOf{\Xi} + \probaCond{ a = \alpha }{ U=0,\bar{\Xi} }\cdot \probaOf{\bar{\Xi}} \\
&=
\frac{m}{n}\cdot e^{-\cdot\frac{1}{4}\cdot 4} \frac{4^{-\sum_{\ell=1}^4 \alpha_\ell}}{\alpha_1!\alpha_2!\alpha_3!\alpha_4!} 
+ \left(1-\frac{m}{n}\right)\cdot \Big(\\
&\qquad \frac{1}{2} \frac{e^{-\frac{\eps m}{n}}}{\alpha_1!\alpha_2!\alpha_3!\alpha_4!} p_1^{\alpha_1}p_2^{\alpha_2}p_2^{\alpha_3}p_3^{\alpha_4} +\\
&\qquad \frac{1}{2} \frac{e^{-\frac{\eps m}{n}}}{\alpha_1!\alpha_2!\alpha_3!\alpha_4!} p_3^{\alpha_1}p_2^{\alpha_2}p_2^{\alpha_3}p_1^{\alpha_4}
\Big) \;.
\end{align*}  
Similarly, for the $\no$-case, we have
\begin{equation}\label{eq:lb:nocase:appendix}
            \probaCond{ a = \alpha }{ U=1,\bar{\Xi} } = \frac{e^{-\frac{\eps m}{n}}}{\alpha_1!\alpha_2!\alpha_3!\alpha_4!} \left( 
                \frac{1}{8} q_1^{\alpha_1}p_2^{\alpha_2}p_2^{\alpha_3}q_3^{\alpha_4} + 
                \frac{1}{8} q_3^{\alpha_1}p_2^{\alpha_2}p_2^{\alpha_3}q_1^{\alpha_4} +
                \frac{3}{4} q_2^{\alpha_1}p_2^{\alpha_2}p_2^{\alpha_3}q_2^{\alpha_4}
                \right) \;,
\end{equation}  
$$\probaCond{ a = \alpha }{ U=0, \Xi }=\probaCond{ a = \alpha }{ U=1, \Xi } \;,$$ 
and
\begin{align*}
          \probaCond{ a = \alpha }{ U=1 } 
          &= \probaCond{ a = \alpha }{ U=1, \Xi }\cdot \probaOf{\Xi} + \probaCond{ a = \alpha }{ U=1,\bar{\Xi} }\cdot \probaOf{\bar{\Xi}} \\
          &= \frac{m}{n}\cdot e^{-\frac{1}{4}\cdot 4} \frac{4^{-\sum_{\ell=1}^4 \alpha_\ell}}{\alpha_1!\alpha_2!\alpha_3!\alpha_4!} 
            + \left(1-\frac{m}{n}\right)\cdot \\ &\Big(
            \frac{1}{8} \frac{e^{-\frac{\eps m}{n}}}{\alpha_1!\alpha_2!\alpha_3!\alpha_4!} q_1^{\alpha_1}p_2^{\alpha_2}p_2^{\alpha_3}q_3^{\alpha_4} +
            \frac{1}{8} \frac{e^{-\frac{\eps m}{n}}}{\alpha_1!\alpha_2!\alpha_3!\alpha_4!} q_3^{\alpha_1}p_2^{\alpha_2}p_2^{\alpha_3}q_1^{\alpha_4} +
            \frac{3}{4} \frac{e^{-\frac{\eps m}{n}}}{\alpha_1!\alpha_2!\alpha_3!\alpha_4!} q_2^{\alpha_1}p_2^{\alpha_2}p_2^{\alpha_3}q_2^{\alpha_4}
            \Big) \;.
\end{align*}
\noindent With these formulas in hand, we can write
\begin{align*}
\Phi(n,m,\eps) 
& \eqdef \sum_{\alpha\in\N^4} \probaOf{ a=\alpha } \left(\frac{ \probaCond{ a =\alpha }{ U=0 } - \probaCond{ a =\alpha }{ U=1 }}{ \probaCond{ a =\alpha }{ U=0 }} \right)^2 \\
          &= \left(1-\frac{m}{n}\right)^2\sum_{\alpha\in\N^4} \probaOf{ a=\alpha } \left(\frac{ \probaCond{ a =\alpha }{ U=0, \bar{\Xi} } - \probaCond{ a =\alpha }{ U=1, \bar{\Xi} }}{ \probaCond{ a =\alpha }{ U=0 }} \right)^2 \\
          &\leq \sum_{\alpha\in\N^4} \probaOf{ a=\alpha } \left(\frac{ \probaCond{ a =\alpha }{ U=0, \bar{\Xi} } - \probaCond{ a =\alpha }{ U=1, \bar{\Xi} }}{ \probaCond{ a =\alpha }{ U=0 }} \right)^2 \;.
\end{align*}
By~\cref{eq:lb:yescase:appendix,eq:lb:nocase:appendix} and~\cref{claim:lb:lowmoments:appendix}, we observe that 
the difference $\probaCond{ a =\alpha }{ U=0, \bar{\Xi} } - \probaCond{ a =\alpha }{ U=1, \bar{\Xi} }$ 
is zero for any $\abs{\alpha} \eqdef \sum_i \alpha_i \leq 3$. We thus obtain
\begin{align*}
\Phi(n,m,\eps) 
&\leq \sum_{\substack{\alpha\in\N^4\\ \abs{\alpha} \geq 4}} \probaOf{ a=\alpha } 
\left(\frac{ \probaCond{ a =\alpha }{ U=0, \bar{\Xi} } - \probaCond{ a =\alpha }{ U=1, \bar{\Xi} }}{ \probaCond{ a =\alpha }{ U=0 }} \right)^2 \\
&= \sum_{k=4}^\infty\sum_{\substack{\alpha\in\N^4\\ \abs{\alpha} =k}} \probaOf{ a=\alpha } \cdot \probaCond{\abs{a}=k}{\bar{\Xi}}^2 \cdot \tag{$\dagger$}\\
&\qquad\left(\frac{ \probaCond{ a =\alpha }{ U=0, \bar{\Xi}, \abs{a}=k } - \probaCond{ a =\alpha }{ U=1, \bar{\Xi}, \abs{a}=k }}{ \probaCond{ a =\alpha }{ U=0 }} \right)^2 \\
&= \sum_{k=4}^\infty  \frac{e^{-\frac{2\eps m}{n}}}{k!^2}\left(\frac{\eps m}{n}\right)^{2k} \sum_{\substack{\alpha\in\N^4\\ \abs{\alpha} =k}} \probaOf{ a=\alpha } \cdot \\
&\qquad\left(\frac{ \probaCond{ a =\alpha }{ U=0, \bar{\Xi}, \abs{a}=k } - \probaCond{ a =\alpha }{ U=1, \bar{\Xi}, \abs{a}=k }}{ \probaCond{ a =\alpha }{ U=0 }} \right)^2 \\
&\leq \sum_{k=4}^\infty  \frac{e^{-\frac{2\eps m}{n}}}{k!^2}\left(\frac{\eps m}{n}\right)^{2k} \sum_{\substack{\alpha\in\N^4\\ \abs{\alpha} =k}} \probaOf{ a=\alpha } \cdot \left(\frac{2}{ \probaCond{ a =\alpha }{ U=0 }} \right)^2 \;,
\end{align*}
where for $(\dagger)$ we used the fact that, $\abs{a}$ is independent of $U$ to write 
$$\probaCond{\abs{a}=k}{U=1, \bar{\Xi}}=\probaCond{\abs{a}=k}{U=0, \bar{\Xi}}
=\probaCond{\abs{a}=k}{\bar{\Xi}} = \frac{e^{-\frac{\eps m}{n}}}{k!}\left(\frac{\eps m}{n}\right)^{k} \;.$$
To conclude the proof, we will handle the denominator using the bound
\[
\probaCond{ a =\alpha }{ U=0 } \geq \probaCond{ a =\alpha }{ \Xi,U=0 }\cdot \probaCond{\Xi}{U=0} = \frac{m}{n} e^{-1} \frac{4^{-k}}{k!} \;,
\]
and rewrite $\probaOf{ a=\alpha } = \probaCond{ a=\alpha }{ \abs{a}=\abs{\alpha} }\cdot\probaOf{ \abs{a}=\abs{\alpha} }$. 
Using $x\eqdef \frac{\eps m}{n}$ in the following expressions for conciseness, we now get:
\begin{align*}
\Phi(n,m,\eps) 
&\leq 4e^2\frac{n^2}{m^2}e^{-2x}\sum_{k=4}^\infty \left(4x\right)^{2k} \probaOf{ \abs{a}=k } \sum_{\substack{\alpha\in\N^4\\ \abs{\alpha} =k}} \probaCond{ a=\alpha }{ \abs{a}=k }  
= 4e^2\frac{n^2}{m^2}e^{-2x}\sum_{k=4}^\infty \left(4x\right)^{2k} \probaOf{ \abs{a}=k } \\
&= 4e^2\frac{n^2}{m^2}e^{-2x}\sum_{k=4}^\infty \left(4x\right)^{2k} \left( \frac{m}{n} \frac{e^{-1}}{k!} + \left(1-\frac{m}{n}\right) e^{-x} \frac{x^k}{k!} \right)\\
&\leq 40\frac{n^2}{m^2}\sum_{k=4}^\infty \left(4x\right)^{2k} \left( \frac{m}{n} \frac{1}{k!} + \frac{x^k}{k!}  \right)
= 40\frac{n}{m}\sum_{k=4}^\infty \frac{1}{k!}\left(4x\right)^{2k} + 40\frac{n^2}{m^2}\sum_{k=4}^\infty \frac{1}{k!} \left( 4^{2/3}x\right)^{3k} \\
&= 40\frac{n}{m}\frac{(4x)^{8}}{24} + o\left(\frac{n}{m}x^{8} \right) + 40\frac{n^2}{m^2} \frac{1}{24} \left( 4^{2/3}x\right)^{12}  + o\left(\frac{n^2}{m^2}x^{12}\right) \tag{$\ddagger$} \\
&= 2^{16}\frac{5}{3}\cdot  \frac{\eps^8 m^7}{n^7} + o\left(\frac{\eps^8 m^7}{n^7}\right) \;,
\end{align*}
where for $(\ddagger)$ we relied on the Taylor series expansion of $\exp$, recalling that $\frac{\eps m}{n} \ll 1$.
This completes the proof of~\cref{lemma:lb:case3:main:appendix}.
\end{proofof}
      
\subsection{Second Lower Bound Regime\texorpdfstring{: $\bigOmega{n^{7/8}/\eps}$ for $\frac{1}{n^{3/8}} \leq \eps \leq \frac{1}{n^{1/8}}$}{}} \label{ssec:lb-medium-eps:appendix}
Assume we are in the regime where
\[
\max(\min(n^{6/7}/\eps^{8/7},n^{7/8}/\eps),\sqrt{n}/\eps^2) = n^{7/8}/\eps \;,
\]
i.e., $\frac{1}{n^{3/8}} \leq \eps \leq \frac{1}{n^{1/8}}$. 
Suppose there is a testing algorithm for conditional independence using $m \leq c n^{7/8}/\eps$ samples, 
for a sufficiently small universal constant $c>0$. In this regime, we also have $m \ll n$.
  
Our construction of $\yes$- and $\no$- instances in this case is similar to 
those of the previous lower bound, although some specifics about how $\p_Z$ is generated will change.
  
\paragraph{The $\yes$-instance.}
A pseudo-distribution $\p$ is drawn from the $\yes$-instances as follows: 
Independently for each value $1\leq z\leq n-1$, we set:
\begin{itemize}
\item With probability $\frac{1}{2}$, $\p_Z(z) = \frac{1}{m}$ and $\p_z(i, j)= \frac{1}{4}$ for all $i, j\in\{0,1\}$. 
That is, we select uniform marginals for $\p_z(i, j)$.

\item With probability $\frac{1}{2}$, $\p_Z(z) = \frac{\eps}{n}$ and 
$\p_z(i, j)$ is defined by the same $2\times 2$ matrices as in the previous case:
\begin{itemize}
\item With probability $1/2$, $Y_1$,
\item With probability $1/2$, $Y_2$.
\end{itemize}
\end{itemize}
Furthermore, we set $\p_Z(n) = 1$, and $\p_n(i, j)= \frac{1}{4}$ for all $i, j \in\{0,1\}$. 
The last condition ensures that $\normone{\p_Z}=\bigTheta{1}$.  
It is clear that the resulting pseudo-distribution $\p$ satisfies $\normone{\p_Z}=\bigTheta{1}$ 
and that $\p \in\condindprop{\{0,1\}}{\{0,1\}}{[n]}$.
      
\paragraph{The $\no$-instance.}
A pseudo-distribution $\p$ is drawn from the $\no$-instances as follows: 
Independently for each value $z\in[n]$, we set:
\begin{itemize}
\item With probability $\frac{m}{n}$, $\p_Z(z) = \frac{1}{m}$ and $\p_z(i, j)= \frac{1}{4}$ for all $i,j\in\{0,1\}$, as before.
\item With probability $1-\frac{m}{n}$, $\p_Z(z) = \frac{\eps}{n}$, and $\p_z(i, j)$ is defined by the same $2\times 2$ matrix as before:
\begin{itemize}
\item With probability $1/8$, $N_1$,
\item With probability $1/8$, $N_2$,
\item With probability $3/4$, $N_3$.
\end{itemize}
\end{itemize}
Furthermore, we set as before $\p_Z(n) = 1$, 
and $\p_n(i, j) = \frac{1}{4}$ for all $i, j \in\{0,1\}$. 
This construction ensures that 
$\expect{\totalvardist{\p}{\condindprop{\{0,1\}}{\{0,1\}}{[n]}}} = \bigOmega{\eps}$.
      
The bulk of the proof of the $\bigOmega{n^{6/7}/\eps^{8/7}}$ remains the same. 
In particular, setting $a\eqdef a_1$, we can bound from above as before the mutual information $\mutualinfo{U}{a}$ by
\[
  \mutualinfo{U}{a} \leq \sum_{k=4}^\infty  \frac{e^{-\frac{2\eps m}{n}}}{k!^2}\left(\frac{\eps m}{n}\right)^{2k} \sum_{\substack{\alpha\in\N^4\\ \abs{\alpha} =k}} \probaOf{ a=\alpha } \cdot \left(\frac{2}{ \probaCond{ a =\alpha }{ U=0 }} \right)^2 \;.
\]
In the current setting, we use the fact that 
\[
\probaCond{ a =\alpha }{ U=0 } \geq \probaCond{ a =\alpha }{ \Xi,U=0 }\cdot \probaCond{\Xi}{U=0} = \frac{1}{2} e^{-1} \frac{4^{-k}}{k!} \;,
\]
to obtain that
\[
  \mutualinfo{U}{a} \leq 16e^2 e^{-\frac{2\eps m}{n}} \sum_{k=4}^\infty\left(\frac{4\eps m}{n}\right)^{2k} 
\sum_{\substack{\alpha\in\N^4\\ \abs{\alpha} =k}} \probaOf{ a=\alpha }
= 16e^2 e^{-\frac{2\eps m}{n}} \sum_{k=4}^\infty\left(\frac{4\eps m}{n}\right)^{2k} \probaOf{ \abs{a}=k } \;.
\]
Recalling that $\probaOf{ \abs{a}=k } = \frac{1}{2} e^{-1} \frac{1}{k!} +  \frac{1}{2} e^{-\frac{\eps m}{n}} \frac{1}{k!}\left(\frac{\eps m}{n}\right)^{2k}$ 
and that $\frac{\eps m}{n} = \bigTheta{\frac{1}{n^{1/8}}} \ll 1$, with a Taylor series expansion of the first term of the sum, 
we finally get that
\[
          \mutualinfo{U}{a} = \bigO{\frac{\eps^8m^8}{n^8}} \;.
\]
Therefore, $\mutualinfo{U}{A} \leq \sum_{z=1}^{n-1} \bigO{\frac{\eps^8m^8}{n^8}} =  \bigO{\frac{\eps^8m^8}{n^7}}$, 
which is $o(1)$ unless $m = \bigOmega{n^{7/8}/\eps}$. This completes the proof of this branch of the lower bound.

\subsection{Third Lower Bound Regime\texorpdfstring{: $\bigOmega{\frac{\sqrt{n}}{\eps^2}}$ for $\eps < n^{-3/8}$}{}} \label{ssec:lb-small-eps:appendix}
Finally, assume we are in the regime where
\[ 
\max(\min(n^{6/7}/\eps^{8/7},n^{7/8}/\eps),\sqrt{n}/\eps^2) = \sqrt{n}/\eps^2 \;,
\]
i.e., $\eps < n^{-3/8}$. In this case, we can show the desired lower bound by a simple 
reduction from the known hard instances for uniformity testing. 
Let $N$ be an even positive integer.
It is shown in~\cite{Paninski:08} that $\bigOmega{\sqrt{N}/\eps^2}$ samples are required to distinguish 
between (a) the uniform distribution on $[N]$, 
and (b) a distribution selected at random by pairing consecutive elements $2i, 2i+1$ 
and setting the probability mass of each pair to be either
$\left(\frac{1+2\eps}{N}, \frac{1-2\eps}{N}\right)$ or $\left(\frac{1-2\eps}{N},\frac{1+2\eps}{N}\right)$ 
independently and uniformly at random.
  
We map these instances to our conditional independence setting as follows: 
Let $N=4n$. We map $[N]$ to the set $\{0,1\}\times \{0,1\}\times [n]$
via the mapping $\Phi : [N] \to \{0,1\}\times \{0,1\}\times [n]$ defined as follows:
$\Phi(2i) = (0,0,i)$, $\Phi(2i+1) = (0,1,i)$, $\Phi(2i+2) = (1,0,i)$, and $\Phi(2i+3) = (1,1,i)$. 

For a distribution $\p\in\distribs{[N]}$ adversarially selected as described above, the following 
conditions are satisfied: (1) In case (a), $\Phi(\p)$ is the uniform distribution on $\{0,1\}\times \{0,1\}\times [n]$
and therefore $\Phi(\p)\in\condindprop{\{0,1\}}{\{0,1\}}{[n]}$. 
(2) In case (b), it is easy to see that for each fixed value of the third coordinate, 
the conditional distribution on the first two coordinates is one of the following:
  \[
      \frac{1}{4}\begin{pmatrix}
          1+2\eps & 1-2\eps\\
          1+2\eps & 1-2\eps\\
      \end{pmatrix}
      , \quad 
      \frac{1}{4}\begin{pmatrix}
          1+2\eps & 1-2\eps\\
          1-2\eps & 1+2\eps\\
      \end{pmatrix}
      , \quad 
      \frac{1}{4}\begin{pmatrix}
          1-2\eps & 1+2\eps\\
          1-2\eps & 1+2\eps\\
      \end{pmatrix}, \text{ or }\quad
      \frac{1}{4}\begin{pmatrix}
          1-2\eps & 1+2\eps\\
          1+2\eps & 1-2\eps\\
      \end{pmatrix} \;.
  \]
It is clear that each of these four distributions is $\bigOmega{\eps}$-far from independent.
Therefore, by~\cref{lemma:distance:product} it follows that 
$\totalvardist{ \Phi(\p) }{ \condindprop{\{0,1\}}{\{0,1\}}{[n]} } = \bigOmega{\eps}$.
  
In conclusion, we have established that any testing algorithm 
for $\condindprop{\{0,1\}}{\{0,1\}}{[n]}$ can be used (with parameter $\eps$) 
to test uniformity over $[N]$ (with parameter $\eps'=O(\eps)$). 
This implies a lower bound of $\bigOmega{\sqrt{N}/{\eps'}^2} = \bigOmega{\sqrt{n}/\eps^2}$ 
on the sample complexity of $\eps$-testing conditional independence.
 
\section{Sample Complexity Lower Bound for \texorpdfstring{$\domx=\domy=\domz=[n]$}{X=Y=Z=[n]}} \label{sec:lb:nnn}
   
In this section, we outline the proof of the (tight) sample complexity lower bound of
$
  \bigOmega{n^{7/4}}
$
for testing conditional independence in the regime that 
$\domx=\domy=\domz=[n]$, and $\eps=\bigOmega{1}$.

Specifically, we will show that, when $\abs{\domx}=\abs{\domy}=\abs{\domz}=n$ and $\eps= 1/20$, it is impossible to distinguish between conditional independence and $\eps$-far from conditional independence with $s=o(n^{7/4})$ samples. To do this, we begin by producing an adversarial ensemble of distributions. The adversarial distribution will be designed to match the cases where the upper bound construction will be tight. In particular, each conditional marginal distribution will have about $n^{3/4}$ heavy bins and the rest of the bins light. The difference between our distribution and the product of the marginals (if it exists) will be uniformly distributed about the light bins.

First we will come up with an ensemble where $X$ and $Y$ are conditionally independent and then we will tweak it slightly. For the first distribution, we let the distribution over $Z$ be uniform. For each value of $Z=z$, we pick random subsets $A_z, B_z \subset [n]$ of size $n^{3/4}$ (which we assume to be an integer), which will be the heavy bins of the conditional distribution. We then let the conditional probabilities be defined by
\[
  \probaCond{X=j}{Z=z} = \begin{cases} n^{-3/4}/2 & \textrm{if }j\in A_z\\ 1/(2(n-n^{3/4})) & \textrm{else} \end{cases}
\]
and
\[
  \probaCond{Y=j}{Z=z} = \begin{cases} n^{-3/4}/2 & \textrm{if }j\in B_z\\ 1/(2(n-n^{3/4})) & \textrm{else}.\end{cases}
\]
We then let $X$ and $Y$ be conditionally independent on $Z$, defining the distribution over $(X,Y,Z)$.

Finally, we introduce a new one bit variable $W$. In ensemble $\mathcal{D}_0$, $W$ is an independent random bit. In ensemble $\mathcal{D}_1$, the conditional distribution on $W$ given $(X,Y,Z)$ is uniform random if $X\in A_Z$ or $Y\in B_Z$, but otherwise is given by a uniform random function $f\colon[n]\times [n]\times [n] \to \{0,1\}$. In particular, in this case, $W$ is determined by the values of $X,Y,Z$, though different elements of $\mathcal{D}_1$ will give different functions. Note that elements of $\mathcal{D}_0$ have $XW$ and $Y$ conditionally independent on $Z$, whereas elements of $\mathcal{D}_1$ are $\eps$-far from any such distribution. We show that no algorithm that takes $s$ samples from a random distribution from one of these families can reliably distinguish which family the samples came from.

In particular, let $F$ be a uniform random bit. Let $S$ be a sequence of $s$ quadruples $(W,X,Y,Z)$ obtained by picking a random element $\p$ from $\mathcal{D}_F$ and taking $s$ independent samples from $\p$. It suffices to show that one cannot reliably recover $F$ from $S$. Note that with high probability $S$ contains at most $t\eqdef 2s/n+\log n=o(n^{3/4})$ samples for each value of $Z$. Therefore of we let $T_z$ be a sequence of $t$ independent samples from $\p$ conditioned on $Z=z$ for each $z$, it suffices to show that $F$ cannot be reliably recovered from $T_1,\ldots,T_n$. For this it suffices to show that $\mutualinfo{F}{T_1,\ldots, T_n}=o(1)$. Since the $T_z$ are conditionally independent on $F$, we have that $\mutualinfo{F}{T_1,\ldots, T_n} \leq \sum_{z=1}^n \mutualinfo{F}{T_z}$ and thus it suffices to show that $\mutualinfo{F}{T_z}=o(1/n)$ for every $z$. Since this shared information is clearly the same for each $z$, we will suppress the subscript.

We say that two distinct elements of $T$ \emph{collide} if they have the same values of $X$ and $Y$. We note that if we condition on the values of $X$ and $Y$ in the elements of $T$, that the values of $W$ for the elements that do not collide are uniform random bits independent of the other values of $W$ and of $F$. Therefore, no information can be gleaned from these $W$'s. This means that all of the information comes from $W$'s associated with collisions. Unfortunately, most collisions (as we will see) come from heavy values of either $X$ or $Y$ (or both), and these cases will also provide no extra information.

More formally, note that
\[
\mutualinfo{F}{T} = \shortexpect_{M\sim T}\left[O\left(\min\left(1,\left(1-\frac{\probaCond{T=M}{F=0}}{\probaCond{T=M}{F=1}} \right)^2 \right) \right) \right].
\]
We claim that if $M$ has $C$ pairs that collide, then
\[
\left(1-\frac{\probaCond{T=M}{F=0}}{\probaCond{T=M}{F=1}} \right)^2 = O(C^2/n).
\]
Our result will then follow from the observation that $\expect{C^2} = O(t^2/n^{3/2}) = o(1)$.

Given $M$, call a value of $X$ (resp. $Y$) \emph{extraneous} if it
\begin{itemize}
\item occurs as an $X$- (resp. $Y$-) coordinate of an element of $M$; and
\item does not occur as an $X$- (resp. $Y$-) coordinate of an element of $M$ involved in a collision.
\end{itemize}
We claim that our bound
\begin{equation}\label{eq:rel:prob:bound}
  \left(1-\frac{\probaCond{T=M}{F=0}}{\probaCond{T=M}{F=1}} \right)^2 = O(C^2/n).
\end{equation}
holds even after conditioning on which extraneous $X$ are in $A$ and which extraneous $Y$ are in $B$. This will be sufficient since which $X$ and $Y$ are extraneous, and which of them are in $A$ or $B$, are both independent of $F$. Let $M_E$ be the $X$ and $Y$ values of $M$ along with the information on which extraneous $X$ and $Y$ are in $A$ or $B$. It is enough to bound
\[
\left(1-\frac{\probaCond{T=M}{F=0}}{\probaCond{T=M}{F=1}} \right)^2 = \left(\frac{\probaCond{T=M}{F=1}-\probaCond{T=M}{F=0}}{\probaCond{T=M}{F=1}} \right)^2.
\]
It thus suffices to bound
\[
\left(\totalvardist{ (T\mid F=0,M_E) }{ (T\mid F=1,M_E) } \right)^2.
\]
Say that a pair of colliding elements of $M$ is \emph{light} if neither the corresponding $X$ nor the corresponding $Y$ are in $A$ or $B$. Observe that if we condition on no collision in $M$ being light, the conditional distributions of $T$ on $F=0$ and on $F=1$ are the same. Therefore, the expression above is bounded by
\[
\probaOf{ \textrm{There exists a light collision in }M}^2.
\]
Thus, it suffices to bound the probability that $M$ contains a light collision given its values of $X$ and $Y$. If $M$ (ignoring the values of $W$ and $Z$) is $((x_1,y_1),\ldots,(x_t,y_t))$ the probability of seeing this $M$ conditioned on the values of $A$ and $B$ is $\Theta(n^{-2t+\abs{\setOfSuchThat{i\in[t]}{ x_i\in A }}/4+\abs{\setOfSuchThat{i\in[t]}{y_i\in B}}/4})$. Now given some set of $a$ of values of $X$ appearing in $M$, the prior probability that those are the values of $X$ in $M$ appearing in $A$ is $n^{-a/4}\cdot \phi$, where $\phi$ is some quantity that changes by only a $1+o(1)$ factor if a single element is added or removed from the set of $X$'s. A similar relation holds for $Y$'s. Given this, and conditioning on whether the extraneous $X$'s and $Y$'s are in $A$ and $B$, we note that each non-extraneous $X$ or $Y$ that is in $A$ or $B$ contributes a factor of roughly $n^{-1/4}$ to the prior probability of having that configuration of elements in $A$ or $B$, but contributes a factor of at least $n^{1/2}$ to the conditional probability of seeing the $M$ that we saw given those values of $A$ and $B$. Therefore, even conditioned on the previously determined lightness of other collisions, each collision has only a $O(n^{-1/2})$ probability of being light. Therefore the probability that there is a light collision is $O(C/n^{1/2})$. This implies~\cref{eq:rel:prob:bound} and completes the proof of our lower bound.

\bibliographystyle{alpha}
\bibliography{allrefs}

\newcommand{\etalchar}[1]{$^{#1}$}
\begin{thebibliography}{dMASdBP14}

\bibitem[ADK15]{ADK15}
J.~Acharya, C.~Daskalakis, and G.~Kamath.
\newblock Optimal testing for properties of distributions.
\newblock In {\em Proceedings of NIPS'15}, 2015.

\bibitem[Agr92]{Agresti1992}
A.~Agresti.
\newblock A survey of exact inference for contingency tables.
\newblock {\em Statist. Sci.}, 7(1):131--153, 02 1992.

\bibitem[BFF{\etalchar{+}}01]{BFFKRW:01}
T.~Batu, E.~Fischer, L.~Fortnow, R.~Kumar, R.~Rubinfeld, and P.~White.
\newblock Testing random variables for independence and identity.
\newblock In {\em Proc. 42nd IEEE Symposium on Foundations of Computer
  Science}, pages 442--451, 2001.

\bibitem[BFR{\etalchar{+}}00]{BFR+:00}
T.~Batu, L.~Fortnow, R.~Rubinfeld, W.~D. Smith, and P.~White.
\newblock Testing that distributions are close.
\newblock In {\em {IEEE} Symposium on Foundations of Computer Science}, pages
  259--269, 2000.

\bibitem[BH07]{BlundHor07}
R.~Blundell and J.~L. Horowitz.
\newblock A non-parametric test of exogeneity.
\newblock {\em The Review of Economic Studies}, 74(4):1035--1058, 2007.

\bibitem[BKR04]{BKR:04}
T.~Batu, R.~Kumar, and R.~Rubinfeld.
\newblock Sublinear algorithms for testing monotone and unimodal distributions.
\newblock In {\em {ACM} Symposium on Theory of Computing}, pages 381--390,
  2004.

\bibitem[BT14]{BouezTaam14}
T.~Bouezmarni and A.~Taamouti.
\newblock Nonparametric tests for conditional independence using conditional
  distributions.
\newblock {\em Journal of Nonparametric Statistics}, 26(4):697--719, 2014.

\bibitem[Can15]{Canonne15}
C.~L. Canonne.
\newblock A survey on distribution testing: Your data is big. but is it blue?
\newblock {\em Electronic Colloquium on Computational Complexity {(ECCC)}},
  22:63, 2015.

\bibitem[CDGR16]{CDGR16}
C.~L. Canonne, I.~Diakonikolas, T.~Gouleakis, and R.~Rubinfeld.
\newblock Testing shape restrictions of discrete distributions.
\newblock In {\em 33rd Symposium on Theoretical Aspects of Computer Science,
  {STACS} 2016}, pages 25:1--25:14, 2016.
\newblock See also \cite{CDGR:17:journal} (full version).

\bibitem[CDGR17]{CDGR:17:journal}
C.~L. Canonne, I.~Diakonikolas, T.~Gouleakis, and R.~Rubinfeld.
\newblock Testing shape restrictions of discrete distributions.
\newblock {\em Theory of Computing Systems}, pages 1--59, 2017.

\bibitem[CDKS17]{CDKS17}
C.~L. Canonne, I.~Diakonikolas, D.~M. Kane, and A.~Stewart.
\newblock Testing {B}ayesian networks.
\newblock In {\em Proceedings of the 30th Conference on Learning Theory, {COLT}
  2017}, pages 370--448, 2017.

\bibitem[CDS17]{CDS17}
C.~L. Canonne, I.~Diakonikolas, and A.~Stewart.
\newblock Fourier-based testing for families of distributions.
\newblock {\em CoRR}, abs/1706.05738, 2017.

\bibitem[CDVV14]{CDVV14}
S.~Chan, I.~Diakonikolas, P.~Valiant, and G.~Valiant.
\newblock Optimal algorithms for testing closeness of discrete distributions.
\newblock In {\em SODA}, pages 1193--1203, 2014.

\bibitem[Coc54]{Coch54}
W.~G. Cochran.
\newblock Some methods for strengthening the common $\chi^2$ tests.
\newblock {\em Biometrics}, 10(4):417--451, 1954.

\bibitem[Daw79]{Dawid79}
A.~P. Dawid.
\newblock Conditional independence in statistical theory.
\newblock {\em Journal of the Royal Statistical Society. Series B
  (Methodological)}, 41(1):1--31, 1979.

\bibitem[DDK18]{DaskalakisDK16}
C.~Daskalakis, N.~Dikkala, and G.~Kamath.
\newblock Testing {I}sing models.
\newblock In {\em SODA}, 2018.
\newblock To appear.

\bibitem[DDS{\etalchar{+}}13]{DDSVV13}
C.~Daskalakis, I.~Diakonikolas, R.~Servedio, G.~Valiant, and P.~Valiant.
\newblock Testing $k$-modal distributions: Optimal algorithms via reductions.
\newblock In {\em SODA}, pages 1833--1852, 2013.

\bibitem[DGPP16]{DiakonikolasGPP16}
I.~Diakonikolas, T.~Gouleakis, J.~Peebles, and E.~Price.
\newblock Collision-based testers are optimal for uniformity and closeness.
\newblock {\em Electronic Colloquium on Computational Complexity {(ECCC)}},
  23:178, 2016.

\bibitem[DGPP17]{DGPP17}
I.~Diakonikolas, T.~Gouleakis, J.~Peebles, and E.~Price.
\newblock Sample-optimal identity testing with high probability.
\newblock {\em CoRR}, abs/1708.02728, 2017.

\bibitem[DK16]{DK:16}
I.~Diakonikolas and D.~M. Kane.
\newblock A new approach for testing properties of discrete distributions.
\newblock In {\em FOCS}, pages 685--694, 2016.
\newblock Full version available at abs/1601.05557.

\bibitem[DKN15a]{DKN:15:FOCS}
I.~Diakonikolas, D.~M. Kane, and V.~Nikishkin.
\newblock Optimal algorithms and lower bounds for testing closeness of
  structured distributions.
\newblock In {\em 56th Annual {IEEE} Symposium on Foundations of Computer
  Science, {FOCS} 2015}, 2015.

\bibitem[DKN15b]{DKN:15}
I.~Diakonikolas, D.~M. Kane, and V.~Nikishkin.
\newblock Testing identity of structured distributions.
\newblock In {\em Proceedings of the Twenty-Sixth Annual {ACM-SIAM} Symposium
  on Discrete Algorithms, {SODA} 2015, San Diego, CA, USA, January 4-6, 2015},
  2015.

\bibitem[DKN17]{DKN17}
I.~Diakonikolas, D.~M. Kane, and V.~Nikishkin.
\newblock Near-optimal closeness testing of discrete histogram distributions.
\newblock In {\em 44th International Colloquium on Automata, Languages, and
  Programming, {ICALP} 2017}, pages 8:1--8:15, 2017.

\bibitem[DM01]{DelgMant01}
M.~A. Delgado and W.~G. Manteiga.
\newblock Significance testing in nonparametric regression based on the
  bootstrap.
\newblock {\em The Annals of Statistics}, 29(5):1469--1507, 2001.

\bibitem[dMASdBP14]{dM14}
P.~de~Morais~Andrade, J.~M. Stern, and C.~A. de~Braganca~Pereira.
\newblock Bayesian test of significance for conditional independence: The
  multinomial model.
\newblock {\em Entropy}, 16(3):1376--1395, 2014.

\bibitem[Dob59]{Dobrushin59}
R.~L. Dobru\v{s}in.
\newblock A general formulation of the fundamental theorem of {S}hannon in the
  theory of information.
\newblock {\em Uspehi Mat. Nauk}, 14(6 (90)):3--104, 1959.

\bibitem[DP17]{DaskalakisP17}
C.~Daskalakis and Q.~Pan.
\newblock Square {H}ellinger subadditivity for {B}ayesian networks and its
  applications to identity testing.
\newblock In {\em Proceedings of the 30th Conference on Learning Theory, {COLT}
  2017}, pages 697--703, 2017.

\bibitem[EO87]{Easley87}
D.~Easley and M.~O'Hara.
\newblock Price, trade size, and information in securities markets.
\newblock {\em Journal of Financial Economics}, 19(1):69 -- 90, 1987.

\bibitem[Fis24]{Fisher24}
R.~A. Fisher.
\newblock The distribution of the partial correlation coefficient.
\newblock {\em Metron}, 3:329--332, 1924.

\bibitem[Gol17]{Gol:17}
O.~Goldreich.
\newblock {\em Introduction to Property Testing}.
\newblock Cambridge University Press, 2017.

\bibitem[Gra80]{Granger80}
C.W.J. Granger.
\newblock Testing for causality: A personal viewpoint.
\newblock {\em Journal of Economic Dynamics and Control}, 2(Supplement C):329
  -- 352, 1980.

\bibitem[GS10]{GL10}
G.~Geenens and L.~Simar.
\newblock Nonparametric tests for conditional independence in two-way
  contingency tables.
\newblock {\em Journal of Multivariate Analysis}, 101(4):765--788, 2010.

\bibitem[HPS16]{HardtPNS16}
M.~Hardt, E.~Price, and N.~Srebro.
\newblock Equality of opportunity in supervised learning.
\newblock In {\em Advances in Neural Information Processing Systems 29: Annual
  Conference on Neural Information Processing Systems 2016}, pages 3315--3323,
  2016.

\bibitem[Hua10]{Huang10}
T.-M. Huang.
\newblock Testing conditional independence using maximal nonlinear conditional
  correlation.
\newblock {\em Ann. Statist.}, 38(4):2047--2091, 08 2010.

\bibitem[LG96]{LG96}
O.~Linton and P.~Gozalo.
\newblock {Conditional Independence Restrictions: Testing and Estimation}.
\newblock Cowles Foundation Discussion Papers 1140, Cowles Foundation for
  Research in Economics, Yale University, 1996.

\bibitem[LRR11]{LRR11}
R.~Levi, D.~Ron, and R.~Rubinfeld.
\newblock Testing properties of collections of distributions.
\newblock In {\em ICS}, pages 179--194, 2011.

\bibitem[MH59]{MS59}
N.~Mantel and W.~Haenszel.
\newblock Statistical aspects of the analysis of data from retrospective
  studies of disease.
\newblock {\em Journal of the National Cancer Institute}, 22(4):719--748, April
  1959.
\newblock {PMID:} 13655060.

\bibitem[Nea03]{Neapolitan:2003}
R.~E. Neapolitan.
\newblock {\em Learning Bayesian Networks}.
\newblock Prentice-Hall, Inc., 2003.

\bibitem[NUU17]{Natori17}
K.~Natori, M.~Uto, and M.~Ueno.
\newblock Consistent learning {B}ayesian networks with thousands of variables.
\newblock In {\em Proceedings of The 3rd International Workshop on Advanced
  Methodologies for Bayesian Networks}, volume~73 of {\em Proceedings of
  Machine Learning Research}, pages 57--68. PMLR, 20--22 Sep 2017.

\bibitem[Pan08]{Paninski:08}
L.~Paninski.
\newblock A coincidence-based test for uniformity given very sparsely-sampled
  discrete data.
\newblock {\em IEEE Transactions on Information Theory}, 54:4750--4755, 2008.

\bibitem[Pea88]{Pearl88}
J.~Pearl.
\newblock {\em Probabilistic Reasoning in Intelligent Systems: Networks of
  Plausible Inference}.
\newblock Morgan Kaufmann Publishers Inc., San Francisco, CA, USA, 1988.

\bibitem[Pin05]{Pinsker05}
M.~S. Pinsker.
\newblock On the estimation of information via variation.
\newblock {\em Problemy Peredachi Informatsii}, 41(2):3--8, 2005.

\bibitem[Rub12]{Rub12}
R.~Rubinfeld.
\newblock Taming big probability distributions.
\newblock {\em XRDS}, 19(1):24--28, 2012.

\bibitem[SGS00]{Spirtes2000}
P.~Spirtes, C.~Glymour, and R.~Scheines.
\newblock {\em Causation, Prediction, and Search}.
\newblock MIT press, 2nd edition, 2000.

\bibitem[Son09]{Song09}
K.~Song.
\newblock Testing conditional independence via {R}osenblatt transforms.
\newblock {\em Ann. Statist.}, 37(6B):4011--4045, 12 2009.

\bibitem[SW07]{SuWhite07}
L.~Su and H.~White.
\newblock A consistent characteristic function-based test for conditional
  independence.
\newblock {\em Journal of Econometrics}, 141(2):807 -- 834, 2007.

\bibitem[SW08]{SuWhite08}
L.~Su and H.~White.
\newblock A nonparametric {H}ellinger metric test for conditional independence.
\newblock {\em Econometric Theory}, 24(4):829--864, 2008.

\bibitem[SW14]{SuWhite14}
L.~Su and H.~White.
\newblock Testing conditional independence via empirical likelihood.
\newblock {\em Journal of Econometrics}, 182(1):27 -- 44, 2014.
\newblock Causality, Prediction, and Specification Analysis: Recent Advances
  and Future Directions.

\bibitem[TBA06]{Tsamardinos2006}
I.~Tsamardinos, L.~E. Brown, and C.~F. Aliferis.
\newblock The max-min hill-climbing bayesian network structure learning
  algorithm.
\newblock {\em Machine Learning}, 65(1):31--78, Oct 2006.

\bibitem[VV14]{VV14}
G.~Valiant and P.~Valiant.
\newblock An automatic inequality prover and instance optimal identity testing.
\newblock In {\em FOCS}, 2014.

\bibitem[WH17]{WH17}
X.~Wang and Y.~Hong.
\newblock Characteristic function based testing for conditional independence: a
  nonparametric regression approach.
\newblock {\em Econometric Theory}, pages 1--35, 2017.

\bibitem[Wyn78]{Wyner78}
A.~D. Wyner.
\newblock A definition of conditional mutual information for arbitrary
  ensembles.
\newblock {\em Inform. and Control}, 38(1):51--59, 1978.

\bibitem[ZPJS11]{Zhang11uai}
K.~Zhang, J.~Peters, D.~Janzing, and B.~Sch\"{o}lkopf.
\newblock Kernel-based conditional independence test and application in causal
  discovery.
\newblock In {\em Proceedings of the Twenty-Seventh Conference on Uncertainty
  in Artificial Intelligence}, UAI'11, pages 804--813. AUAI Press, 2011.

\end{thebibliography}

\appendix

\section{Testing with Respect to Mutual Information}\label{sec:mutualinfo}

We conclude by considering a slightly different model from the one considered thus far. In particular, while the total variation metric is a reasonable one to measure what it means for $X$ and $Y$ to be far from conditionally independent, there is another metric that is natural in this context: \emph{conditional mutual information}. Specifically, we modify the testing problem to distinguish between the cases where $X$ and $Y$ are conditionally independent on $Z$ and the case where $\condmutualinfo{X}{Y}{Z}\geq \eps$. Our picture here is somewhat less complete, but we are still able to say something in the case where $X,Y$ are binary.

\begin{theorem}\label{thm:binary:mutualinfo}
If $X$ and $Y$ are binary random variables and $Z$ has a support of size $n$, there exists a sample-efficient algorithm that distinguishes between $\condmutualinfo{X}{Y}{Z}=0$ and $\condmutualinfo{X}{Y}{Z}\geq \eps$ with sample complexity
\[
O(\max(\min(n^{6/7}\log^{8/7}(1/\eps)/\eps^{8/7}, n^{7/8}\log(1/\eps)/\eps),\sqrt{n} \log^2(1/\eps)/\eps^2))\,.
\]
\end{theorem} 
\begin{proof}
This follows immediately upon noting that by~\cref{lemma:relation:cmi:tv} (stated and proven later), that if $X$ and $Y$ are $\eps$-close in total variation distance from being conditionally independent on $Z$, then $\condmutualinfo{X}{Y}{Z} \leq O(\eps\log(1/\eps))$; or, by the contrapositive, that $\condmutualinfo{X}{Y}{Z} \geq \eps$ implies that  $X$ and $Y$ are $\Omega(\eps/\log(1/\eps))$-far in total variation distance from being conditionally independent on $Z$. Therefore, it suffices to run our existing conditional independence tester with parameter $\eps' \eqdef \Omega(\eps/\log(1/\eps))$. The sample complexity of this tester is as specified.
\end{proof}

\begin{remark}[On the optimality of this bound]
It is not difficult to modify the analysis slightly in order to remove the logarithmic factors from the first two
terms in the above expression. Intuitively, this is because these terms arise only when at least half of the mutual information comes from ``light''
bins, with mass at most $1/m$. In this case, these bins contribute at least $m^4 \sum_z \eps_z^2 \p_Z(z)^4 \gg m^4 \sum_z (\p_Z(z) \eps_z \log(1/\eps_z))^4
\gg m^4 \eps^4/ n^3$ to the expectation of $Z$, and the analysis proceeds from there as before.

It is also easy to show that in this regime our lower bounds still apply, as the hard instances also produced distributions with mutual information $\Omega(\eps)$.\footnote{I.e., the conditional mutual information of ``\no-distributions'' is easily seen to actually be $\Omega(\eps)$, while applying the relation between total variation distance and conditional mutual information as a black-box to the $\eps$ distance in total variation distance would incur a quadratic loss in $\eps$.}{} Therefore, we have matching upper and lower bounds as long as $\eps \gg n^{-3/8}/\log^2 n$.

However, it seems likely that the correct behavior in the small $\eps$ regime is substantially different when testing with respect to mutual information. The difficult cases for total variation distance testing actually end up with mutual information merely $\condmutualinfo{X}{Y}{Z} = O(\eps^2)$. It is quite possible that a better algorithm or a better analysis of the existing algorithm could give substantially improved performance when $\eps < n^{-3/8}$. In fact, it is conceivable that the sample complexity of $O(n^{7/8}/\eps)$ could be maintained for a broad range of $\eps$. The only lower bound that we know preventing this is a lower bound of $\Omega(\eps\log(1/\eps))$ by noting that there are distributions with $\condmutualinfo{X}{Y}{Z}\geq \eps$, but where $(X,Y,Z)$ is $O(\eps/\log(1/\eps))$-far in variation distance from being conditionally independent.
\end{remark}

\begin{lemma}\label{lemma:relation:cmi:tv}
  Assume $(X,Y,Z)\sim\p$, where $\p\in\distribs{ \domx\times\domy\times\domz }$ with $\abs{\domx}=\ell_1$, $\abs{\domy}=\ell_2$, and $\abs{\domz}=n$. Then, for every $\eps\in(0,1)$,
  \begin{itemize}
    \item If $\totalvardist{ \p }{ \condindprop{\domx}{\domy}{\domz} } \leq \eps$, then $\condmutualinfo{X}{Y}{Z} \leq O(\eps\log(\ell_1\ell_2/\eps))$;
    \item If $\totalvardist{ \p }{ \condindprop{\domx}{\domy}{\domz} } \geq \eps$, then $\condmutualinfo{X}{Y}{Z} \geq 2\eps^2$.
  \end{itemize}
\end{lemma}
\begin{proof}
  The second item is simply an application of Pinsker's inequality, recalling that 
  \[
      \condmutualinfo{X}{Y}{Z} = \kldiv{ (X,Y)\mid Z }{ (X \mid Z)\otimes (Y \mid Z) }\,.
  \]
  i.e. the Kullback--Leibler divergence between the joint distribution of $(X,Y\mid Z)$ and the product of
marginals $(X\mid Z)$ and $(Y\mid Z)$. As for the first, it follows from the relation between conditional mutual information and total variation distance obtained in~\cite{Pinsker05} (and~\cref{lemma:distance:product}).
\end{proof}
 \section{Deferred Proofs from~\cref{sec:prelims}}\label{sec:deferred:prelims}

\subsection{Proof of~\cref{lem:dtv-cond}}
The proof follows from the following chain of (in-)equalities:
\begin{align*}
2\totalvardist{\p}{\p'} 
  &= \sum_{(i, j, z) \in \domx\times\domy \times \domz} \abs{ \p(i,j,z) - \p'(i,j,z) }\\
  &=  \sum_{(i, j, z) \in \domx\times\domy \times \domz} \abs{ \p_Z(z) \cdot \p_z(i,j) - \p'_Z(z) \cdot \p'_z(i,j) }\\
  &= \sum_{(i, j, z) \in \domx\times\domy \times \domz} \abs{ \p_Z(z) \cdot (\p_z(i,j) - \p'_z(i,j)) + (\p_Z(z) - \p'_Z(z)) \cdot \p'_z(i,j) } \\
  &\leq \sum_{(i, j, z) \in \domx\times\domy \times \domz} \p_Z(z) \cdot \abs{ \p_z(i,j) - \p'_z(i,j) } 
    + \sum_{(i, j, z) \in \domx\times\domy \times \domz} \abs{\p_Z(z) - \p'_Z(z)} \cdot \p'_z(i,j) \\
  &= \sum_{z\in\domz} \left( \p_Z(z) \cdot \sum_{(i,j)\in\domx\times\domy} \abs{ \p_z(i,j) - \p'_z(i,j) } \right) 
    + \sum_{z\in\domz} \left( \abs{ \p_Z(z) - \p'_Z(z) } \cdot \sum_{(i,j)\in\domx\times\domy} \p'_z(i,j) \right) \\
  &= 2 \sum_{z\in\domz} \p_Z(z) \cdot \totalvardist{\p_z}{\p'_z} + 2\totalvardist{\p_Z}{\p'_Z} \;,
\end{align*}
where the fourth line used the triangle inequality and the last line used the fact that 
$\sum_{(i,j)\in\domx\times\domy} \p'_z(i,j) = 1$.
This completes the proof of the first part of the lemma. For the second part, we note that the equality in \eqref{eq:useful:decomposition} 
holds if and only if the triangle inequality in the fourth line above holds with equality, i.e., when $\p_Z=\p'_Z$.
This completes the proof of~\cref{lem:dtv-cond}. \qed

\subsection{Proof of~\cref{lemma:distance:product}}

Let $\p'\in\condindprop{\domx}{\domy}{\domz}$ be such that $\totalvardist{\p}{\p'}\leq \eps$ 
and $\q =  \sum_{z\in\domz}\p_Z(z) \q_z$. Since 
$\totalvardist{\p}{\q} \leq \totalvardist{\p}{\p'} + \totalvardist{\p'}{\q}
\leq \eps + \totalvardist{\p'}{\q}$, it suffices to show that $\totalvardist{\p'}{\q} \leq 3\eps.$
By~\cref{lem:dtv-cond}, we have that
\begin{align*}
\totalvardist{\q}{\p'}  
&\leq \sum_{z\in\domz} \q_Z(z) \cdot \totalvardist{\q_z}{\p'_z} + \totalvardist{\q_Z}{\p'_Z} \\
&= \sum_{z\in\domz} \p_Z(z) \cdot \totalvardist{\q_z}{\p'_z} + \totalvardist{\p_Z}{\p'_Z} \\
&= \sum_{z\in\domz} \p_Z(z) \cdot \totalvardist{\p_{z,X} \otimes \p_{z,Y}}{\p'_{z,X} \otimes \p'_{z,Y}}  + \totalvardist{\p_Z}{\p'_Z}   \\
&\leq  \sum_{z\in\domz} \p_Z(z) \cdot \left( \totalvardist{\p_{z,X}}{\p'_{z,X}}+ \totalvardist{\p_{z,Y}}{\p'_{z,Y}} \right)  + \totalvardist{\p_Z}{\p'_Z}    \\
&= \sum_{z\in\domz} \p_Z(z)\totalvardist{\p_{z,X}}{\p'_{z,X}}+ \sum_{z\in\domz} \p_Z(z)\totalvardist{\p_{z,Y}}{\p'_{z,Y}}  + \totalvardist{\p_Z}{\p'_Z}   \\
&\leq 3\eps  \;,
\end{align*}
where the second line uses the fact that $q_Z = p_Z$, the third line uses the fact that 
$\q_z = \p_{z,X} \otimes \p_{z,Y}$ (\cref{def:product:conditional:marginals}) 
and that $\p'_z = \p'_{z,X} \otimes \p'_{z,Y}$ (since  $\p' \in\condindprop{\domx}{\domy}{\domz}$), 
the fourth line uses the sub-additivity of total variation distance for product distributions, and
the last line uses the fact that each of the three terms in the fifth line is bounded from above by 
$\totalvardist{\p}{\p'}$. This completes the proof of~\cref{lemma:distance:product}. \qed

\subsection{Proof of~\cref{fact:split:distributions:l2norm:nonpoisson}}
This lemma is essentially shown in~\cite{DK:16}.
The only difference is that we require a proof for (ii)
when $S$ is a set of $m$ independent samples (as opposed to $\mathrm{Poi}(m)$ samples
from $\p$ in ~\cite{DK:16}). We show this by an explicit calculation below.

Let $a_i$ equal one plus the number of copies of $i$ in $S$, i.e. $a_i \eqdef 1+\sum_{j\in S} \indic{i=j}$.
We note that the expected squared $\lp[2]$-norm of $\p_S$ is
$
\expect{\sum_{i=1}^n \sum_{j=1}^{a_i} \p_i^2/a_i^2} = \sum_{i=1}^n \p_i^2 \expect{1/a_i}.
$
Further, $a_i$ is distributed as $1+X$ where $X$ is a $\binomial{m}{\p_i}$ random variable.
Therefore,
\begin{align*}
    \expect{\frac{1}{1+X}}
    &= \sum_{k=0}^m \frac{1}{k+1}\binom{m}{k} \p_i^k(1-\p_i)^{m-k}
    = \frac{1}{(m+1)\p_i}\sum_{k=0}^m \binom{m+1}{k+1} \p_i^{k+1}(1-\p_i)^{(m+1)-(k+1)}\\
    &= \frac{1}{(m+1)\p_i}\sum_{\ell=1}^{m+1} \binom{m+1}{\ell} \p_i^{\ell}(1-\p_i)^{(m+1)-\ell}
    = \frac{1-(1-\p_i)^{m+1}}{(m+1)\p_i} \leq \frac{1}{(m+1)\p_i}\,.
\end{align*}
This implies
$
\expect{[\normtwo{\p_S}^2} \leq \sum_{i=1}^n \p_i^2 / (m\p_i) = (1/m) \sum_{i=1}^n \p_i = 1/m.
$
which completes the proof. \qed

\subsection{Proof of~\cref{claim:variance:truncated:poisson}}

  Recalling that $\expect{N}=\lambda$ and $\expect{N^2}=\lambda+\lambda^2$, we get
  \[
  \expect{N \indic{ N \geq 4}} = e^{-\lambda}\sum_{k=4}^\infty k\frac{\lambda^k}{k!} = 
  \lambda - e^{-\lambda}\left(\lambda+\lambda^2+\frac{1}{2}\lambda^3\right) \eqdef f(\lambda) \;,
  \]
  and
  \[
  \var{N \indic{ N \geq 4}} = \left(\lambda+\lambda^2 - e^{-\lambda}\left(\lambda+2\lambda^2+\frac{3}{2}\lambda^3\right) \right)
  - \left(\lambda - e^{-\lambda}\left(\lambda+\lambda^2+\frac{1}{2}\lambda^3\right)\right)^2 \eqdef g(\lambda) \;.
  \]
  From these expressions, it is easy to check that 
  (i) $\lim_{x\to 0}\frac{f(x)}{g(x)} = \frac{1}{4}$, and
  (ii) $\lim_{x\to \infty}\frac{f(x)}{g(x)} = 1$. 
  From the definition as a variance of a non-constant random variable, it follows 
  that $g(x)>0$ for all $x>0$, from which we get that 
  (iii) $\frac{f}{g}$ is continuous and positive on $[0,\infty)$. 
  Combining these three statements, we get that $\frac{f}{g}$ achieves a minimum $c$ on $[0,\infty)$, 
  and that this minimum is positive. 
  This implies the result with $C\eqdef 1/c$. 
  The value $4.22$ comes from studying numerically this ratio, whose minimum is achieved for $x\simeq 1.1457$.

\subsection{Proof of~\cref{claim:variance:truncated:poisson:2}}

Let $a,b\geq 0$, $\lambda > 0$, and assume $X\sim\poisson{\lambda}$. 
Without loss of generality, suppose $0 < a\leq b$ (the case $a=0$ being trivial). 
We can rewrite $X \sqrt{\min(X,a)\min(X,b)}\indic{X\geq 4}$ as $Y$ with
      \[
          Y \eqdef X^2\indic{X\leq a} + \sqrt{a} X^{3/2} \indic{a < X\leq b}  + \sqrt{ab} X \indic{X > b}
      \]
      which implies
      \[
          Y^2 = X^4\indic{X\leq a} + a X^{3} \indic{a < X\leq b}  + ab X^2 \indic{X > b}\,.
      \]
     By linearity of expectation, the original claim boils down to proving there exists $C>0$ such that
     \begin{align*}
        \expect{X^4\indic{4\leq X\leq a}} &+\expect{a X^{3} \indic{a < X\leq b} \indic{X\geq 4} }+\expect{ab X^2 \indic{X > b} \indic{X\geq 4} } \\
        &\leq C\left(\expect{X^2\indic{4\leq X\leq a}} + \expect{\sqrt{a} X^{3/2} \indic{a < X\leq b} \indic{X\geq 4}}  + \expect{\sqrt{ab} X \indic{X > b} \indic{X\geq 4}}\right) \\
        &\quad+ \left(\expect{X^2\indic{4\leq X\leq a}} + \expect{\sqrt{a} X^{3/2} \indic{a < X\leq b} \indic{X\geq 4}}  + \expect{\sqrt{ab} X \indic{X > b} \indic{X\geq 4}}\right)^2
     \end{align*}
     and since $(x+y+z)^2 \geq x^2+y^2+z^2$ for $x,y,z\geq 0$, it is enough to show
     \[
          \expect{\beta^2 X^{2\alpha} \indic{X\in S}} \leq C\beta\expect{X^\alpha \indic{X\in S}} + \beta^2\expect{X^\alpha \indic{X\in S}}^2
     \]
     for $\alpha,\beta>0$, and $S\subseteq\R_+$ an interval. This in turn follows from arguments similar to that of the proof of~\cref{claim:variance:truncated:poisson}.

\subsection{Proof of~\cref{claim:expectation:truncated:poisson:squared:with:min}}

For $\lambda < 8$, we can take bound the expectation by the contribution of $X=4$ as
	  $\expect{ X \sqrt{\min(X,a)\min(X,b)}\indic{X\geq 4}} \geq 4 \sqrt{\min(4,a)\min(4,b)} \lambda^4/4! \geq \lambda^4/3$. Then $\lambda \sqrt{\min(\lambda,a)\min(\lambda,b)} \geq \lambda \sqrt{\min(\lambda,2)\min(\lambda,2)} \geq \lambda^2/4 \geq \lambda^4/256$. Thus  $\min(\lambda \sqrt{\min(\lambda,a)\min(\lambda,b)}, \lambda^4) \leq 256 \lambda^4$. Putting this together we have that for $\lambda < 8$,
	  
	  \[
            \expect{ X \sqrt{\min(X,a)\min(X,b)}\indic{X\geq 4}  } \geq (1/768) \min(\lambda \sqrt{\min(\lambda,a)\min(\lambda,b)}, \lambda^4) \;.
      \]
      To deal with the $\lambda \geq 8$ case, we claim that $\probaOf{X \geq \lfloor \lambda/2 \rfloor} \geq 1/2$ in this case.
	  To see this, we just need to expand $1=\exp(-\lambda) \sum_{k=0}^{\infty} f(k) \lambda^k/k!$ and note that for $1 \leq k \leq \lambda/2$, the ratio of the $k$ term to the $k-1$ term is at least $\lambda/k \geq 2$. Thus the sum of the first $\lfloor \lambda/2 \rfloor$ terms is smaller than the $k=\lfloor \lambda/2 \rfloor$ term and so
	  \[
	         \exp(-\lambda) \sum_{k=0}^{\lfloor \lambda/2 \rfloor-1}  \lambda^k/k! \leq \exp(-\lambda) \sum_{k=\lfloor \lambda/2 \rfloor}^{\infty} \lambda^k/k \; .
	  \]
	  The RHS is $\probaOf{X \geq \lfloor \lambda/2 \rfloor}$ and the LHS is $\probaOf{X < \lfloor \lambda/2 \rfloor}=1-\probaOf{X \geq \lfloor \lambda/2 \rfloor}$  and so we have that $\probaOf{X \geq \lfloor \lambda/2 \rfloor} \geq 1/2$ as claimed.
	  
	  For $8 \leq \lambda \leq 2 \min{a,b}$, we have
	  \begin{align*}
		\expect{ X \sqrt{\min(X,a)\min(X,b)}\indic{X\geq 4}} & \geq \expect{ X^2 \indic{X\geq \lfloor \lambda/2 \rfloor}} \\
		& \geq (1/2) (\lfloor \lambda/2 \rfloor)^2 \geq \lambda^2/3 \\
		& \geq \min(\lambda \sqrt{\min(\lambda,a)\min(\lambda,b)}/6, \lambda^4/3) \\
		& \geq (1/6) \min(\lambda \sqrt{\min(\lambda,a)\min(\lambda,b)}, \lambda^4)
	 \end{align*}
	 
	 For $\lambda \geq 2 \max{a,b,4}$, noting that for $X \geq \lambda/2$, $\sqrt{\min(X,a)\min(X,b)}=\sqrt{ab}$, we have
	 \begin{align*} 
		\expect{ X \sqrt{\min(X,a)\min(X,b)}\indic{X\geq 4}} & \geq \expect{ X \sqrt{\min(X,a)\min(X,b)}\indic{X\geq \lfloor \lambda/2 \rfloor}} \\
		& \geq  \lfloor \lambda/2 \rfloor \sqrt{ab} \\
		& \geq  \sqrt{ab}\lambda/3 \\
		& \geq (1/6) \min(\lambda \sqrt{\min(\lambda,a)\min(\lambda,b)}, \lambda^4)
	\end{align*}
	The final case we need to consider is when $\lambda$ is between $2a$ and $2b$ and the maximum of those is over $4$
	Supposing without loss of generality that $a \leq b$, for $\max{2a,4} \leq \lambda \leq 2b$, 	we have
	\begin{align*}
		\expect{ X \sqrt{\min(X,a)\min(X,b)}\indic{X\geq 4}} & \geq \expect{ X \sqrt{\min(X,a)\min(X,b)}\indic{X\geq \lfloor \lambda/2 \rfloor}} \\
		& \geq  \lfloor \lambda/2 \rfloor^{3/2} \sqrt{a}  \\
		& \geq \sqrt{ab} \lambda^{3/2}/4 \\
		& \geq (1/8) \min(\lambda \sqrt{\min(\lambda,a)\min(\lambda,b)}, \lambda^4)\;.
	\end{align*}

\end{document}